%% file: main.tex
\documentclass[5p,sort&compress]{elsarticle}
\usepackage[T1]{fontenc}
\usepackage[utf8]{inputenc}
\let\today\relax
\makeatletter
\def\ps@pprintTitle{%
\let\@oddhead\@empty
\let\@evenhead\@empty
\def\@oddfoot{\footnotesize\itshape
     {} \hfill\today}%
\let\@evenfoot\@oddfoot
}
\makeatother
\usepackage[htt]{hyphenat} 
\usepackage{silence} 
\usepackage{xcolor}
\usepackage{amsmath,amssymb}
\usepackage{subcaption}
\usepackage{bm}
\usepackage{mathtools}
\usepackage{enumitem}
\usepackage[hidelinks,unicode=true]{hyperref}

\usepackage[algo2e,ruled,noend]{algorithm2e}
\usepackage{listings}
\usepackage[capitalise]{cleveref}
\usepackage{bbold}
\usepackage{color}
\usepackage{soul}
\usepackage{lettrine}
\usepackage{placeins}
\usepackage{titlesec}

\usepackage{bookmark}
\usepackage{amsfonts}
\usepackage{booktabs} 
\usepackage{tikz}
\usetikzlibrary{quantikz}
\usepackage{amsthm} 
\newtheorem{theorem}{Theorem}
\newtheorem{lemma}[theorem]{Lemma} 
\theoremstyle{definition}
\newtheorem{definition}{Definition}
\newtheorem{corollary}[theorem]{Corollary} 

\DeclarePairedDelimiter{\E}{\langle}{\rangle}

\usepackage{xprintlen}

\title{Quantum walk informed variational algorithm design}

\author{Edric Matwiejew\corref{cor1}}
\cortext[cor1]{edric.matwiejew@uwa.edu.au}

\author{Jingbo B. Wang}
\address{Department of Physics, The University of Western Australia, WA, 6009}

\begin{document}

\begin{frontmatter}

\begin{abstract}
\input{abstract} 
\end{abstract}

\begin{keyword}
quantum optimisation \sep continuous-time quantum walks \sep portfolio rebalancing \sep parallel task scheduling 
\end{keyword}
\end{frontmatter}

\input{body_text}

\section*{Acknowledgements}

This work was supported by resources provided by the Pawsey Supercomputing Centre with funding from the Australian Government and the Government of Western Australia. EM acknowledges the support of the Australian Government Research Training Program Scholarship. Additionally, EM would like to thank Dr Kaushika De Silva for providing feedback on the presentation of key mathematical concepts and A/Prof. Du Huynh for discussions on the phase-based analysis of the optimisation QVA states.

\bibliographystyle{elsarticle-num}
\bibliography{bibliography}
\appendix
\input{appendix}
\end{document}

%% file: abstract.tex
We present a theoretical framework for the analysis of amplitude transfer in Quantum Variational Algorithms (QVAs) for combinatorial optimisation with mixing unitaries defined by vertex-transitive graphs, based on their continuous-time quantum walk (CTQW) representation and the theory of graph automorphism groups. This framework leads to a heuristic for designing efficient problem-specific QVAs. Using this heuristic, we develop novel algorithms for unconstrained and constrained optimisation. We outline their implementation with polynomial gate complexity and simulate their application to the parallel machine scheduling and portfolio rebalancing combinatorial optimisation problems, showing significantly improved convergence over preexisting QVAs. Based on our analysis, we derive metrics for evaluating the suitability of graph structures for specific problem instances, and for establishing bounds on the convergence supported by different graph structures. For mixing unitaries characterised by a CTQW over a Hamming graph on $m$-tuples of length $n$, our results indicate that the amplification upper bound increases with problem size like $\mathcal{O}(e^{n \log m})$.

%% file: body_text.tex
\section{Introduction}
Optimisation represents a critical and nearly ubiquitous challenge across various sectors, including logistics, finance, medicine, and compilation~\cite{Amaro2022Dec, bennett_quantum_2021, slate_quantum_2021,lozano_combinatorial_2019,kell_scientific_2012}. At their core, combinatorial optimisation problems (COPs) involve selecting a configuration from a finite set of discrete objects that minimises a given cost. In general, COPs are notoriously difficult to solve. They are often NP-hard due to their lack of easily exploitable structure and have a solution space that grows exponentially with the number of combinatorial variables~\cite{korte_2022,pardalos_complexity_1992}. For this reason, quantum computing emerges as a promising avenue to tackle these challenges~\cite{CZ12,FH16,boulebnane10solving,maslov2021quantum,bravyi2020quantum,wu2021strong,basso2021quantum}. By leveraging the exponential size of a Hilbert space of a relatively small number of qubits, quantum algorithms for combinatorial optimisation utilise superposition to search over the complete solution space in quantum parallel, creating an opportunity to exploit entanglement and interference as mechanisms to accelerate the identification of optimal or near-optimal solutions.  

Quantum Variational Algorithms (QVAs) are a class of hybrid quantum-classical algorithms for combinatorial optimisation ~\cite{Farhi14,marsh_quantum_2019,marsh_combinatorial_2020,HWO+19}. These algorithms utilise an alternating ansatz structure, comprised of interleaved parameterised phase-shift and mixing unitaries. The phase-shift unitary applies a phase rotation proportional to the solution costs of a COP to the basis states of a quantum superposition. The mixing unitary then drives the transfer of probability amplitude between states, where the phase-encoded costs result in interference. By adjusting the phase-shift and mixing parameters through classical optimisation, QVAs aim to induce constructive interference at low-cost solutions. As the number of ansatz iterations increases, so does the potential for convergence to optimal solutions, at the expense of a deeper quantum circuit and a larger optimisation space for the classical optimiser. Flexible circuit depth and robustness to system noise have flagged QVAs as a near-term application for quantum processors in the NISQ era~\cite{Preskill18,symons_practitioners_2023,zhou2020quantum}.

Researchers face several key challenges in the design of efficient QVAs. The first arises in constrained optimisation problems, where the space of valid solutions is restricted to a subset of the problem solution space~\cite{marsh_combinatorial_2020,marsh_quantum_2019,HWO+19}. Approaches based on the Quantum Approximate Optimisation Algorithm (QAOA) apply a penalty function that increases the cost associated with invalid solutions. Alternative QVAs specifically target constrained COPs, initialising a superposition over a space of valid solutions to which a mixing unitary is applied that preserves the problem constraints. An example of this approach is the Quantum Alternating Operator Ansatz (QAOAz), which uses the so-called XY-mixers to restrict the transfer of probability amplitude to states within the same permutation set~\cite{HWO+19,Brandhofer2023,CEB20}. Another example is the Quantum Walk-based Optimisation Algorithm (QWOA), in which an indexing unitary restricts its search to a lexicographically ordered canonical subspace of valid solutions~\cite{marsh_combinatorial_2020,marsh_quantum_2019}. More recently, the Unified Quantum Alternating Operator Ansatz defines constraint-preserving mixing unitaries according to the minimum non-zero Hamming distance between valid solutions, with state initialisation performed by post-selection conditioned on an oracle that marks the set of valid solutions~\cite{ruan_quantum_2023}.

Another challenge is identifying mixing unitaries that can achieve sufficient convergence to optimal solutions with a small number of ansatz iterations. The significance of this task is highlighted by considering the convergence behaviour of Grover's algorithm~\cite{Grover97}. It performs a deterministic search for a target state in a space of size $N$ that converges in $\sim\sqrt{N}$ iterations of its marking and diffusion unitaries; a result that has proven optimal for an unstructured quantum search~\cite{zalka_grovers_1999}. This has significant implications for the practicality of QVAs that employ mixing unitaries based on an unstructured coupling of the solution space, which includes the QWOA and variants of the QAOA, as they are expected to offer no more than a quadratic speedup over a uniform random sampling of the solution space~\cite{matwiejew_quantum_2023,bridi_analytical_2024,bennett2021quantum,symons_practitioners_2023}. As such, as the number of combinatorial variables grows, the corresponding exponential increase in the problem solution space necessitates an exponential increase in the number of ansatz iterations required by these algorithms to achieve a constant degree of convergence, rapidly resulting in impractical circuit depths.

Consequently, it is increasingly evident that an efficient QVA must use a mixing unitary that can exploit structural properties inherent to a target class of COPs. Notably, convergence exceeding that of Grover's algorithm has been observed for the QAOA applied to the maxcut and Boolean satisfiability problems~\cite{wurtz_maxcut_2021,basso2021quantum,boulebnane10solving}. More recently, the Quantum Multivariable Optimisation Algorithm (QMOA) demonstrated highly efficient convergence in the optimisation of high-dimensional oscillatory multivariable functions by employing a mixing unitary that couples across the coordinate dimensions of Cartesian space~\cite{matwiejew_quantum_2023}. The QAOA was also observed to achieve a similar degree of convergence on problems of this type, albeit with less consistent convergence to the globally optimal solution. 

Optimisation of variational parameters constitutes the final core challenge in QVA development. This task is inherently complex, as multivariate optimisation itself is generally NP-hard. As such, tuning the phase-shift and mixing parameters can entail a significant computational overhead. Parameter transfer offers an approach to mitigating this computational bottleneck. These strategies analyse the convergence behaviour of QVA parameters in the context of a specific type of COP to derive heuristics for parameter initialisation that can be generalised to higher ansatz iterations and larger problem instances~\cite{sureshbabu_parameter_2024,headley_problem-size-independent_2023}. This approach has shown considerable promise, particularly in QAOA formulations of the maxcut and weighted maxcut problems. Such heuristics have led to parameter initialisation schemes that generate initial parameters within 1.1\% of the best found through parameter optimisation~\cite{sureshbabu_parameter_2024,wang_quantum_2018}.

Addressing these challenges requires the development of models that describe the dynamics of QVAs in a manner that can guide the design of problem-specific mixing unitary structures and parameter initialisation schemes. A potential basis for this is the theoretical framework of continuous-time quantum walks (CTQWs), which describe a quantum evolution under a Hamiltonian determined by the adjacency structure of an undirected graph~\cite{ADZ93,venegas-andraca_quantum_2012,QWbook2014}. Continuous-time quantum walks have garnered significant attention in quantum algorithm research, particularly in quantum search algorithms aimed at amplifying the probability of single or multiple indistinguishable target states~\cite{CG04,marsh_quantum_2019,QMW+22,razzoli_universality_2022,novo_systematic_2015,marsh_deterministic_2021,SKW03}. These efforts often proceed by identifying symmetries within the graph structure, leading to a lower-dimensional description of the walk~\cite{novo_systematic_2015, razzoli_universality_2022}. However, the application of similar methods to the analysis of QVA mixing unitaries has been limited~\cite{wang_quantum_2018}. The QVA problem cost function results in complex phase distributions that, in general, preclude an exact description of its dynamics in terms of a lower-dimensional Hilbert space. 

Despite this apparent complexity, efforts to address the challenges of QVA design have established a strong connection between global structure in COPs, the convergence potential of QVA mixing unitaries, and the overhead associated with parameter optimisation. This motivates the focus of this work. We adopt a CTQW interpretation of the QVA mixing unitary and, within this picture, derive a compact description of its action on a superposition with phase-encoded solution costs. Our analysis is based on identifying the orbits of vertices under transformations in the graph automorphism group, which we find to be a natural approach for the defining graphs of many QVAs as they are distance or, at least, vertex-transitive. With this description, we prove that the ability of a particular mixing unitary to produce constructive interference is negatively correlated with the average variance over the phase angles of states within the same orbit. The resulting CTQW-based description of QVA dynamics is intuitive and highly instructive, providing explanations for the comparative performance of QVAs applied to the same COP class, practical heuristics for mixing unitary design, parameter initialisation, and efficiently computable figures of merit.

Based on this analytical description, we devise a heuristic for the design of mixing unitary structures in which the minimum non-zero Hamming distance among the valid solutions of a COP defines one or more underlying graphs. This is applied to derive novel QVAs for unconstrained and constrained COPs. The first is based on the QMOA, for which we identify the conditions under which the QMOA mixing unitary encompasses the class of Hamming graphs. In doing so, we establish a direct relationship between the QMOA and the QAOA, offering an explanation for the efficiency of the QAOA on COPs with binary variables. The resulting generalised QMOA can efficiently address unconstrained combinatorial optimisation problems with variables of arbitrary arity.

Moving to constrained COPs, our heuristic for mixing unitary design leads to a partitioning scheme based on permutations of multisets with multiplicities that correspond to the counts of combinatorial variable values for which a given constraint is satisfied. This leads to the development of a QWOA variant, the Quantum Walk-based Optimisation Algorithm over Combinatorial Subsets (QWOA-CS), which utilises a mixing unitary that is defined by two \emph{submixers} that first transfer amplitude within, and then between, subspaces that are defined by the partitioning scheme. We outline the efficient implementation of the QWOA-CS mixing unitary on a gate-based quantum processor through sparse Hamiltonian simulation and a novel Fourier-based method for the simulation of CTQWs on non-homogeneous $K$-partite graphs.

These novel algorithms are compared against pre-existing QVAs on two COPs of practical significance. The QMOA and QAOA are applied to the parallel machine (or task) scheduling problem, which aims to identify an optimal assignment of tasks to machines accounting for task completion time, task priority, non-homogeneous machine speeds, and energy consumption~\cite{xiao_branch_2021,henning_complexity_2020,azizoglu1999minimization,anand2012resource}. The QMOA exhibits a significant advantage over the QAOA, converging to the globally optimal solution in each considered problem instance. To the best of our knowledge, this is the first result indicating that QVAs may be efficiently applied to COPs of this type.

We also consider portfolio optimisation under the Mean-Variance Markowitz Model (MVMM)~\cite{Markowitz1952}. This COP, which is NP-hard with integer-constrained assets, has garnered much attention as a benchmarking application in the growing domain of quantum finance~\cite{buonaiuto_best_2023,baker_wasserstein_2022,slate_quantum_2021,hodson_portfolio_2019,OML19,qu_experimental_2024,hodson_portfolio_2019,Rebentrost18}. Two variants are commonly considered: budget-constrained optimisation and portfolio rebalancing~\cite{baker_wasserstein_2022,slate_quantum_2021}. The latter is characterised by ternary variables that do not have a one-to-one mapping to a Hilbert space of qubits. This results in a significant degree of solution degeneracy in a quantum formulation of the problem cost function~\cite{slate_quantum_2021,hodson_portfolio_2019}.

Previous work has found that the QWOA has a pronounced advantage over alternative methods based on the QAOA and QAOAz~\cite{slate_quantum_2021}. The QWOA indexing and unindexing unitaries remove degenerate solutions from the quantum search space. In comparison, the QAOA and QAOAz, are disadvantaged as their search domain is enlarged by the inclusion of degenerate solution states, which can bias their convergence to solutions away from the global optimum. Our numerical results indicate that the QWOA-CS can significantly outperform the original QWOA.

The structure of this article is as follows. In \cref{sec:background}, we provide theoretical background covering COPs, parallel machine scheduling, portfolio rebalancing, CTQWs, definitions from graph theory, quantum simulation of CTQWs, the QVA approach to combinatorial optimisation, relevant pre-existing QVAs, and the complexity of unitaries referenced in complexity analyses. In \cref{sec:theory}, we present our CTQW-based description of QVA dynamics, which, in \cref{sec:mixer-design}, leads to a heuristic for the construction of graphs for the QVA mixing unitary. In \cref{sec:figures_of_merit}, we introduce novel measures and analysis methods based on our theoretical framework. We then apply the design heuristic to develop the generalised QMOA and the QWOA-CS in \cref{sec:cart_qva,sec:perm_qva}. In \cref{sec:methods_and_results}, we present an analysis of the convergence supported by the graphs of the mixing unitaries considered in this work, benchmarks of QVA performance, and a parameter optimisation scheme based on analytical results. We discuss our results in \cref{sec:discussion} and provide concluding statements in \cref{sec:conclusion}.

\section{Background}
\label{sec:background}

In \cref{sec:cop}, we define the unconstrained and constrained combinatorial optimisation problems and their quantum encoding, followed by definitions of the parallel machine scheduling and portfolio rebalancing problems used for algorithm benchmarking in \cref{sec:methods_and_results}. \cref{sec:graphs_and_ctqws} provides relevant definitions from graph theory and introduces the CTQW. This is followed by an overview of efficient approaches to the quantum simulation of CTQWs in \cref{sec:ctqw_simulation}. In \cref{sec:qva} the QVA scheme for combinatorial optimisation is introduced together with an overview of the pre-existing algorithms that are relevant to this work. Finally, \cref{sec:complexity} provides a summary of the gate complexity and ancilla qubit overhead associated with non-standard unitary operations, which inform our complexity analyses in \cref{sec:cart_qva,sec:perm_qva,sec:methods_and_results}.

\subsection{Combinatorial Optimisation Problems} 
\label{sec:cop}

Consider a combinatorial optimisation problem (COP) with solutions represented as vectors $\bm{s} = (s_0, \dots, s_{n-1})$ of length $n$. Each element $s_i$ is drawn from a finite alphabet $\mathcal{X} = \{ x_0, \dots, x_{m-1} \}$ that contains $m$ distinct elements. The set of all possible solutions is denoted $\mathcal{S}$. Associated with $\mathcal{S}$ is a cost function, $C: \mathcal{S} \rightarrow \mathbb{R}$, that assigns a scalar \emph{cost} to each solution. An unconstrained COP has a valid solution space with a cardinality of $|\mathcal{S}| = m^n$. Constrained optimisation places additional requirements on $\bm{s}$ that restrict the space of valid solutions to a subset of $\mathcal{S}$ such that the space of valid solutions, denoted $\mathcal{S}^\prime$, has a cardinality less than $|\mathcal{S}|$.

This work considers both unconstrained and constrained COPs. In the latter case, we focus on the following equality constraint function~\cite{ruan_quantum_2023},
\begin{equation}
\label{eq:constraint}
h(\bm{s}) = \sum_{i=0}^{n-1} z(s_i) = A,
\end{equation}
where $z$ is an injective map $z: \mathcal{X} \rightarrow \mathbb{Z}$, and the space of valid solutions is defined as $\mathcal{S}' = \{ \bm{s} \in \mathcal{S} \; | \; h(\bm{s}) = A \}$.

In either case, the aim is to identify low-cost solutions that solve or closely approximate,
\begin{equation}
\bm{s}_{\text{opt}} = \Big\{\bm{s} \; | \; C(\bm{s}) = \min \{ C(\bm{s}) \; | \; \bm{s} \in \mathcal{S}'{}\}\Big\}.
\end{equation}
where, for unconstrained optimisation, $\mathcal{S} \equiv \mathcal{S}^\prime$.

\subsubsection{Quantum Encoding of the Problem Solution Space}
\label{sec:cop_encoding}

A quantum encoding of a COP solution space defines a mapping from the alphabet $\mathcal{X}$ to a space of bit-strings $\{0, 1\}^{\lceil \log m \rceil}$. The solution space $\mathcal{S}$ then has a bit-string encoding that is efficiently represented by a superposition over a Hilbert space $\mathcal{H}$ of $n \lceil \log m \rceil$ qubits. Under the condition that the problem cost function $C(\bm{s})$ is expressible as a polynomial function of binary variables, it has a $Z$-basis encoding that defines a diagonal quality operator $\hat{Q}$ where,
\begin{equation}
q_{\bm{s}} = \bra{\bm{s}}\hat{Q}\ket{\bm{s}}
\end{equation}
and $q_{\bm{s}} = C(\bm{s})$~\cite{nielsen_quantum_2010,matwiejew_quop_mpi_2022}.

For problems in which the mapping from $\mathcal{X}$ to $\{0, 1\}^{\lceil \log m \rceil}$ is not one-to-one, the set of feasible solutions, $\mathcal{S}$, may have a cardinality smaller than $\mathcal{H}$, resulting in degenerate or invalid solution states~\cite{slate_quantum_2021}.

For a COP to be efficiently solvable by a QVA, it is also required that $C(\bm{s})$ is computable in polynomial time and that its magnitude is polynomially bounded with respect to $|\mathcal{S}|$. The first requirement ensures that $\hat{Q}$ is computable efficiently in quantum parallel, and the second allows efficient computation of the expectation value of $\hat{Q}$ by the central limit theorem~\cite{crescenzi_structure_1999, matwiejew_quop_mpi_2022}.

\subsubsection{Parallel Machine Scheduling}
\label{sec:pms_background}

For a schedule of $n$ jobs that can be assigned to $m$ machines, the parallel machine scheduling problem (PMS) seeks a configuration of job assignments that minimises time and resource usage. For many variants, in particular those with discrete job assignments, the PMS problem is NP-hard~\cite{henning_complexity_2020,azizoglu1999minimization}. 

We consider a variant of the PMS problem that considers both the overall completion time and energy consumption~\cite{xiao_branch_2021,azizoglu1999minimization,anand2012resource}. The space of feasible solutions for this problem includes all configurations in which each job is assigned to exactly one machine, with multiple assignments to the same machine allowed. 

The solution space is encoded as,
\begin{equation}
    \ket{\mathcal{S}} = \bigotimes_{i=0}^{n-1}\left(\ket{x_{0,i}} + \ket{x_{1,i}} + \dots + \ket{x_{m-1,i}}\right)
\end{equation}
where $x_{j,i}$ denotes the assingment of machine $x_j$ to job $i$, $x_j \in \{0, 1, \dots, m-1\}$ and $\ket{s_i} = \ket{x_{j,i}}$. In relation to this encoding, the problem cost function is expressed as,
\begin{equation}
\label{eq:pms_costs}
\hat{Q} = \,  \sum_{i=1}^{n-1} \sum_{j=0}^{m-1} \left[ \eta \, \left(\frac{w_i\tau_{i}}{\kappa_j}\right) + (1 - \eta) \, \kappa_j^\alpha  \frac{\tau_{i}}{\kappa_j} \right] \ket{x_{j,i}}\bra{x_{j,i}}
\end{equation}
where $w_i \geq 0$ weights job $i$ based on its priority, $\tau_i$ is the processing time of job $i$, $\kappa_j$ is the speed of machine $j$, and $\alpha > 1$ defines the power function of the machines~\cite{xiao_branch_2021}. The weighting variable $\eta$ takes values between $0$ and $1$; when $\eta = 1$, the optimal solution minimises the weighted processing time, whereas for $\eta= 0$, the optimal solution minimises power usage.

\subsubsection{Portfolio Rebalancing}
\label{sec:portfolio}

Portfolio rebalancing is the task of adjusting asset positions in a portfolio to preserve an allocated level of risk and expected return in response to market drift~\cite{guastaroba_models_2009,Markowitz1952}. For a portfolio of $n$ assets, the possible positions are:

\begin{itemize}
    \item \textbf{Long}: Buying and holding an asset in anticipation of its value increasing.
    \item \textbf{Short}: Borrowing and selling an asset with the intention of repurchasing it at a lower price.
    \item \textbf{No}: Not including the asset in the portfolio.
\end{itemize}

A qubit encoding for a problem uses two qubits per asset in which the possible positions are assigned to~\cite{hodson_portfolio_2019},

\begin{itemize}
\item $\ket{01} \rightarrow$ long position.  
\item $\ket{10} \rightarrow$ short position. 
\item $\ket{00}$ and $\ket{11} \rightarrow $ no position.
\end{itemize}

With respect to the above encoding, the MVMM cost function takes the form,
\begin{equation} \label{eq:portfolio_cost_function}
\hat{Q} = \eta \sum_{i,j = 0}^{n-1} \sigma_{ij} \zeta_i \zeta_j - (1 - \eta) \sum_{i = 0}^{n-1} r_i \zeta_i,
\end{equation}
subject to the constraint,

\begin{equation}
A = \sum_{i = 1}^{n} \zeta_i,
\end{equation}
where $\zeta_i = \frac{1}{2}(\hat{Z}_{2i} - \hat{Z}_{2i+1})$ has eigenvalues $(-1, 1, 0)$ for long, short and no positions respectively, $\sigma_{ij}$ is the covariance between assets $s_i$ and $s_j$, and $r_i$ denotes the average return. The parameter $\eta \in \mathbb{R}$ allocates a balance between $r_i$ and $\sigma_{ij}$. As $\eta \rightarrow 0$, the model prioritizes maximum returns, while $\eta \rightarrow 1$ emphasizes minimising volatility~\cite{slate_quantum_2021}. 

The constraint $A$ manages the net position of the portfolio, influencing both its market exposure and its diversification. The balance of long and short positions is an important factor, in general, since short selling, which involves borrowing assets, carries inherent risks. Additionally, $A$ can be chosen to reflect the manager's view on market trends or preference for portfolio stability. For example, $A = 0$ is likely to produce $\bm{s}$ with expected returns that are less prone to market fluctuations~\cite{Markowitz1952,guastaroba_models_2009,green15,RL18}.

\subsection{Graphs and Continuous-Time Quantum Walks}
\label{sec:graphs_and_ctqws}

This work describes the QVA mixing unitary using the conceptual lense of CTQWs. Doing so allows the dynamics of a mixing unitary to be described using graph-theory concepts and terminology. This viewpoint underpins the approach to mixing unitary design presented in \cref{sec:mixer-design} and the algorithms introduced in \cref{sec:cart_qva,sec:perm_qva}.

\subsubsection{Graph Theory}
Here, we provide the formal definition of a graph and the graph properties relevant to this work~\cite{bener_2015,biggs_1974}. Let $G = (V, E)$ be a graph with vertices $V = \{ v \}$ and edges  $E = \{ (v, v^\prime)\}$ .  We require that $G$ be undirected, meaning that if $(v, v^\prime) \in G$ then $(v^\prime, v) \in G$.

The adjacency matrix of $G$ is a square matrix of size $(|V|, |V|)$ where the entry $G_{ij}$ is one if the vertices $v_i$ and $v_j$ are adjacent (that is, there exists an edge between $v_i$ and $v_j$), and zero otherwise. The degree of a vertex $v$ in a weighted graph is the count of edges incident to $v$, denoted $\text{deg}(v)$.

A graph is regular if all vertices have the same degree, which means $\text{deg}(v) = \text{deg}(v^\prime)$ for all $v, v^\prime \in V$. The distance between two vertices $v$ and $v^\prime$ in a graph is the minimum number of edges required to connect them, denoted $\text{dist}(v, v^\prime)$. The diameter of graph $G$ is defined as the largest distance between any pair of its vertices, denoted $\text{diam}(G) = \max \text{dist}(v, v^\prime)$. The shell (or neighbourhood) of a vertex $v$, denoted as $\mathcal{N}_d(v)$, refers to the set of vertices with distance $d$ from vertex $v$~\cite{fortunato_community_2012}. 

A graph automorphism is a bijection $g: V \rightarrow V $ that permutates (or ``relables'') the vertices of $G$ in a manner that preserves the adjacency structure of the graph. All such functions belong to the automorphism group of $G$, denoted $\text{Aut}(G)$. Members of $\text{Aut}(G)$ where $G$ is a finite graph have a representation as a permutation matrix that acts on the graph adjacency matrix as $\pi_{g}^{-1}{G}\pi_{g} = G^\prime$~\cite{biggs_1974}.

A graph is vertex-transitive if, for every pair of vertices $(v, v^\prime)$, there exists a non-trivial automorphism $g: V \rightarrow V$ such that $ g (v) = v^\prime$. A graph is distance-transitive if for any two pairs of vertices $(u, v)$ and $(u^\prime, v^\prime)$ that are the same distance apart, there is an automorphism, $g: V \times V \rightarrow V \times V$, that maps $u$ to $u^\prime$ and $v$ to $v^\prime$. The orbit of a vertex $v$ under the action of $g \in \text{Aut}(G)$ is the set of all $v^\prime \in V$ that can be obtained from $v$ by applying elements of $\text{Aut}(G)$~\cite{biggs_1974}. Every vertex belongs to one orbit. We denote the number of distinct orbits by $|\text{Orb}(G)|$.

\subsubsection{Continuous Time-Quantum Walks}

A quantum search by CTQW defines a correspondence, $\mathcal{S} \leftrightarrow V$, between the quantum solution space $\mathcal{S}$ and the set of vertexes of a graph $G$. A continuous-time quantum walk over $\mathcal{S}$ is then defined as,

\begin{equation}\label{eq:ctqw}
\ket{\psi(t)} = \exp(-\text{i} t \hat{H})\ket{\psi_0}    
\end{equation}
where $\ket{\psi_0} = \sum_{\bm{s} \in \mathcal{S}} c_{\bm{s}} \ket{\bm{s}}$ with complex coefficients $c_{\bm{s}}$. The Hamiltonian $\hat{H}$ is given by the adjacency matrix of $G$~\cite{jafarizadeh_investigation_2007}, or alternatively, the graph Laplacian, 
\begin{equation}\label{eq:laplacian}
\mathcal{L} = D - G
\end{equation}
where $D$ is a diagonal matrix of vertex degrees and $G$ takes the form of an adjacency matrix~\cite{razzoli_universality_2022}.

\subsection{Quantum Simulation of CTQWs}
\label{sec:ctqw_simulation}

Implementing CTQWs on digital quantum computers falls within the scope of quantum Hamilton simulation. As such, the mixing unitaries of an efficient QVA must correspond to some family of efficiently simulatable Hamiltonians. This is most straightforward when $G$ is efficiently expressed as the sum of a number of commuting Puali strings that grow at most polynomially with $n$~\cite{nielsen_quantum_2010}. Otherwise, the adjacency matrix of $G$ must be sparse or belong to some class of efficiently preparable dense Hamiltonians~\cite{marsh_combinatorial_2020,childs_limitations_nodate}.

\subsubsection{Sparse Graphs}
\label{sec:sparse_hamiltonian}

Let $G$ be a graph with a maximum vertex degree that grows at most polynomially with $n$. The adjacency matrix of $G$ will be sparse relative to $|\mathcal{S}|$ and therefore its CTQW efficiently implemented as a sparse Hamiltonian simulation. In this approach, $\hat{H}$ is prepared using oracles, 

\begin{equation}
    \text{COLUMN}\ket{i}\ket{y} = \ket{i}\ket{y \oplus j},
\end{equation}
where $i$ corresponds to a non-zero row and $j$ to a non-zero column, and

\begin{equation}
    \text{ELEMENT}\ket{i}\ket{j}\ket{z} = \ket{i}\ket{j}\ket{z \oplus \hat{H}_{ij}},
\end{equation}
where $\hat{H}_{ij}$ is a non-zero value at position $(i, j)$~\cite{berry2009black}.

Sparse Hamiltonian simulation based on quantum signal processing approximates the time-evolution operation to high precision by a truncated Joacobi-Anger expansion. An expansion of $R$ terms is computed a sequence of parameterised unitaries,
\begin{equation}
\ket{\psi(t)}\ket{\alpha} \approx \hat{U}_{\phi_0}\hat{U}_{\phi_1} \dots \hat{U}_{\phi_{R-1}}\ket{\psi_0}\ket{\alpha}
\end{equation}
comprised of single-qubit rotations on an ancilla register ${\ket{\alpha}}$, and a controlled walk operator that encodes $\hat{H}$ and performs SU(2) rotations on two-dimensional subspaces that span the eigenstates of $\hat{H}$~\cite{berry2009black,low2017optimal,low_hamiltonian_2019}. Simulation up to error $\epsilon$ requires,
\[R = \mathcal{O}\left( d ||\hat{H}||_{\text{max}} t + \frac{\log{(1/\epsilon)}}{\log{\log{(1/\epsilon)}}} \right),\] and $\lceil \log d \rceil + 2$ ancilla qubits, where $||\hat{H}||_\text{max}$ is the largest singular value of $\hat{H}$ and $d$ is the maximum number of non-zero enteries in a row of $\mathcal{H}$~\cite{low_hamiltonian_2019}. The singular values of graph adjacency matrices are upper bound by the maximum vertex degree ~\cite{biggs_1974}. 

\subsubsection{Cirulant Graphs}

Undirected circulant graphs can be efficiently simulated regardless of their sparsity~\cite{zhou_quantum_2017,zhou_efficient_2017}. This efficiency is due to the relationship between the eigensystem of circulant matrices and the discrete Fourier transform. The eigenvalues of a circulant matrix have the closed-form solution,
\begin{equation}
\lambda_j = b_0 + b_{N-1}\omega^j + b_{N-2} \omega^{2j} + \cdots + b_1 \omega^{(N-1)j},
\end{equation}
where $N$ is the size of the matrix, $b_{i=0,\ldots,N-1}$ defines the first row of the circulant matrix, $\omega = \exp\left(\frac{2 \pi \text{i}}{N}\right)$ is a primitive $N^\text{th}$ root of unity, and $j=0,\ldots,N-1$. The corresponding eigenvectors are given by,
\begin{equation}
\nu_j = \frac{1}{\sqrt{N}} ( \omega^j, \omega^{2j}, \ldots, \omega^{(N-1)j}),
\end{equation}
which form the matrix of the discrete Fourier transform~\cite{davis_1979}. Consequently, a continuous-time quantum walk (CTQW) can be efficiently simulated using,
\begin{equation}
\mathcal{F}^{-1} e^{-\text{i} t \hat{\Lambda}} \mathcal{F} \ket{\psi_0},
\end{equation}
where $\mathcal{F}$ is the quantum Fourier transform.

\subsection{Quantum Variational Algorithms}
\label{sec:qva}

Quantum variational algorithms address combinatorial optimisation using a variable sequence of ansatz iterations, summarised as,

\begin{equation} \label{eq:qaoa}
\ket{\bm{\theta}} = \prod_{i = 1}^{p} \hat{U}(\theta_i)\ket{\psi_0}.
\end{equation}
where positive integer $p$ is the number of ansatz iterations, $\hat{U}$ is the ansatz unitary, 
\[\bm{\theta} = ((\gamma_0, \bm{t}_0), (\gamma_1, \bm{t}_1), \dots, (\gamma_{p-1}, \bm{t}_{p-1}))\]
contains real-valued variational parameters for each ansatz iteration and $\ket{\psi_0}$ is an algorithm-specified superposition over $\mathcal{S}$~\cite{matwiejew_quop_mpi_2022}. 

The ansatz unitary encapsulates the phase-shift and mixing unitary, 
\begin{equation}
	\hat{U}(\theta_i)=\hat{U}_W(\bm{t}_i)\hat{U}_Q(\gamma_i).
\end{equation}
The phase-shift unitary, 
\begin{equation}
	\hat{U}_Q(\gamma) = \exp(-\text{i} \gamma \hat{Q}),
	\label{eq:phase-shift}
\end{equation}
applies a phase-rotation proportional to $q_{\bm{s}}$, encoding the solution costs into the phase of the solution states. The mixing unitary $\hat{U}_W$ performs one or more CTQWs with the walk times specified by the $\bm{t}$ parameter. 

A QVA is then executed by repeated preparation of $\ket{\bm{\theta}}$, optimising the variational parameters according to the objective function, 
\begin{equation}
    \label{eq:objective}
	\langle \hat{Q} \rangle = \bra{\bm{\theta}} \hat{Q} \ket{\bm{\theta}}.
\end{equation}

The mixing unitary drives the transfer of probability amplitude between $\ket{\bm{s}}$ during which phase differences encoded by $\hat{U}_Q$ are transferred, causing interference~\cite{matwiejew_quantum_2023}. 

\subsubsection{The Quantum Approximate Optimisation Algorithm}
\label{sec:QAOA}

The foundational Quantum Approximate Optimisation Algorithm addresses unconstrained COPs~\cite{Farhi14}. Its mixing unitary is defined by an $n$-dimensional hypercube graph,

\begin{equation}
    \label{eq:hypercube}
    \hat{G}_\text{QAOA} = \sum_{i=0}^{n \lceil \log m \rceil -1} \hat{X}_i,
\end{equation}
where $\hat{X}_i$ is a NOT gate applied to the $i$-th qubit~\cite{Farhi14,matwiejew_quop_mpi_2022}.

The initial state for the QAOA is an equal superposition over the entire solution space,

\begin{equation}
    \ket{\psi_{0-\text{QAOA}}} = \frac{1}{\sqrt{|\mathcal{S}}|} \sum_{\bm{s} \in \mathcal{S}} \ket{\bm{s}}.
\end{equation}

Application of the QAOA to constrained COPs is performed by defining an unconstrained version of the problem with the modified cost-function,
\begin{equation}
    C_{\text{QAOA}}(\bm{s}) = 
    \begin{cases} 
    C(\bm{s}) & \text{if } \bm{s} \in \mathcal{S}^\prime, \\
    C(\bm{s}) + \Omega(\bm{s}) & \text{otherwise},
\end{cases}
\end{equation}
where $\Omega(\bm{s}) > 0$ is a penalty function that increases the cost of invalid solutions~\cite{slate_quantum_2021,hodson_portfolio_2019}.

\subsubsection{The Quantum Multivariable Optimisation Algorithm}

The Quantum Multivariable Optimisation Algorithm (QMOA) was developed to address the unconstrained optimisation of continuous multivariable functions of dimension $n$~\cite{matwiejew_quantum_2023}. Its mixing unitary is a composite CTQW that mixes with respect to the dimensions of Cartesian space,
\begin{equation}
\hat{U}_{W\text{-QMOA}}(\bm{t}) = \prod_{i=0}^{n-1}\exp(-\text{i}t_i \hat{G}_i),
\end{equation}
where $\bm{t} = (t_0,...,t_{n-1})$ specifies the walk times for the CTQWs, and $\hat{G}_i$ is the adjacency matrix of a circulant graph that fully-connects solutions in dimension $i$. 

As each $\hat{G}_i$ is circulant, $\hat{U}_{W\text{-QMOA}}(\bm{t})$ is efficiently computable as, 
\begin{equation}
\label{eq:qmoa_mixer}
(\mathcal{F}^{-1})^{\otimes n} \exp\left(- \text{i} \hat{\Lambda}(\bm{t}) \right) \mathcal{F}^{\otimes n},
\end{equation}
where $\hat{\Lambda}(\bm{t})$ is defined by,
\begin{equation}
\label{eq:qmoa_eigenvalues}
\hat{\Lambda}(\bm{t}) = \sum_{i = 0}^{n-1}t_{i}\sum_{j=0}^{N_i - 1}\Lambda_{i,j}\ket{x_{i,j}}\bra{x_{i,j}}.
\end{equation}
Here, $\Lambda_{i,j}$ are the eigenvalues of $\hat{G}_i$, $N_i$ is the number of grid points and $\ket{x_{i,j}}$ is a discretised coordinate state, all in dimension $i$. The initial state of the QMOA is a uniform superposition over all coordinate states. The QMOA inherets a gate complexity of $\mathcal{O}\left(\text{polylog} \, \max(N_0, \dots, N_{n-1})\right)$ from the quantum Fourier transform~\cite{matwiejew_quantum_2023}.

\subsubsection{The Quantum Alternating Operator Ansatz}
\label{sec:qaoaz}

The Quantum Alternating Ansatz extends the QAOA to constrained combinatorial optimisation by, first, initialising the system as a superposition over valid solution states and, second, utilising a mixing unitary that couples between solutions satisfying the same constraint value~\cite{HWO+19}. 

Most relevant to this study are the so-called XY-mixers, which act on $\{(0,1),(0,1)\}$ subspaces in $\bm{s}$,
\begin{equation}
    \hat{G}_{XY} = \sum_{i,j \in E(G)}(\hat{X}_i\hat{X}_j + \hat{Y}_i\hat{Y}_j)
\end{equation}
where $E(G)$ is a set of edges associated with a graph that defines the XY-mixer coupling structure.

Two common forms of $\hat{G}_{XY}$ are the parity-ring and complete XY-mixers. The first defines $G(E)$ as two disjoint cycle graphs over alternating qubit pairs, mixing first on even qubit pairs and second on odd pairs. This has the effect of restricting transitions to states with the same parity. The complete XY-mixer defines $G(E)$ as a complete graph, performing a mixing operation over states within the same permutation subspace~\cite{wang_xy-mixers_2020,ruan_quantum_2023}. The gate complexity is $\mathcal{O}(n)$ and $\mathcal{O}(n^2)$ for the parity and complete XY-mixers respectively~\cite{Brandhofer2023}. 

In both instances, the state evolution is restricted to disjoint subgraphs, several of which may contain valid solutions. As such, the initial state is prepared as the Dicke state that populates subgraphs that statfy the problem constraint $A$~\cite{slate_quantum_2021,Brandhofer2023}. For example, portfolio rebalancing with the QAOAz parity or complete XY-mixers initialises $\ket{\psi_0}$ as,

\begin{equation}
    \ket{\psi_0} = \frac{1}{2^{(n-|A|)/2}}\ket{y}^{\otimes |A|}\otimes (\ket{00}+\ket{11})^{\otimes(n-|A|)} \, ,
\end{equation}
where $\ket{y}=\ket{01}$ if $A \geq 0$ or $\ket{10}$ if $A < 0$. This state is efficient to prepare but binomially distributes probability over the valid subgraphs~\cite{hodson_portfolio_2019,slate_quantum_2021}. 

\subsubsection{The Quantum Walk-based Optimisation Algorithm}
\label{sec:qwoa}
The Quantum Walk-based Optimisation Algorithm offers an alternative approach to constrained optimisation in which the evolution of the mixing unitary is restricted to a lexicographically ordered canonical subspace of valid solutions $\mathcal{S^\prime}$~\cite{marsh_combinatorial_2020,marsh_quantum_2019}. Given efficient algorithms $\text{id}$ and $\text{id}^{-1}$ that lexicographically index and unindex $\mathcal{S^\prime}$, it defines an indexing and conjugate unindexing unitary. The action of the indexing unitary is,

\begin{equation}
\label{eq:qwoa_index_5}
\hat{U}_{\#}\ket{\bm{s}}=\left\{\begin{array}{l}
\left|\mathrm{id}(\bm{s})\right\rangle \text{if } \bm{s} \in \mathcal{S}^{\prime}, \\
\ket{\bm{s}} \text { otherwise }.
\end{array}\right.
\end{equation}

As illustrated in \cref{qc:qwoa-indexing}, the indexing unitary is constructed using two registers of $\mathcal{O}(n \log m)$ qubits and a self-adjoint implementation of the indexing algorithm $\hat{U}_{\text{id}}$. It first prepares,
\begin{equation}
    \hat{U}_{\text{id}} \ket{\bm{s}}\ket{0} = \ket{\bm{s}}\ket{\text{id}(\bm{s})},
\end{equation}
followed by a swap operation. The desired state is then given by $\hat{U}_{\text{id}} \ket{\text{id}(\bm{s})}\ket{\bm{s}} = \ket{\text{id}(\bm{s})}\ket{0}$. The unindexing unitary $\hat{U}_\#^\dagger$ is similarly defined.

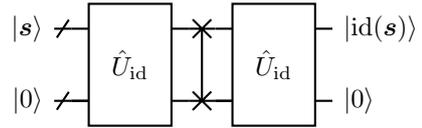
\begin{figure}[t!]
       \centering
        \resizebox{5.57978cm}{!}{%
        \begin{quantikz}[column sep = 0.23cm, row sep = 0.3cm]
        \lstick{\ket{\bm{s}}} & \qwbundle{} & \gate[2][1.1cm]{\hat{U}_{\text{id}}} & \swap{1} & \gate[2][1.1cm]{\hat{U}_{\text{id}}} & \qw  \rstick{\ket{\text{id}(\bm{s})}} \\
        \lstick{\ket{0}} & \qwbundle{} &  \qw & \targX{}  & \qw & \qw  \rstick{\ket{0}}
        \end{quantikz}
        }
    \caption[Circuit overview of the QWOA indexing unitary]{Circuit overview of the QWOA indexing unitary $\hat{U}_\#$.}
    \label{qc:qwoa-indexing}
\end{figure}

The unindexing unitary enables the preparation of an initial state that is equally superposed over $\mathcal{S}^\prime$,
\begin{equation}
\label{eq:qwoa_initial_state_5}
 \ket{\psi_{0\text{-QWOA}}}=\frac{1}{\sqrt{\left|\mathcal{S}^\prime\right|}}\sum\limits_{\bm{s} \in \mathcal{S}^\prime}\ket{\bm{s}}
\end{equation}
by first applying the Fourier transform modulo $|\mathcal{S}^\prime|$ (the quantum Fourier transform on the first $|\mathcal{S}^\prime|$ states), denoted $\mathcal{F}_{|\mathcal{S}^\prime|}$, to $\ket{0}$  followed by $\hat{U}_{\#}^\dagger$. There are highly depth-effcient circuits for the quantum Fourier transform modulo $N$ with gate complexity $\mathcal{O}(\text{polylog} \, N)$~\cite{cleve2000}. 

The QWOA mixing unitary $\hat{U}_{W\text{-QWOA}}$ then performs an indexed CTQW over $|\mathcal{S}^\prime|$ as,
\begin{equation}
\label{eq:qwoa_mixer_5}
\hat{U}_{W\text{-QWOA}}(t) = \hat{U}_\# \exp(-\text{i} t \hat{G}) \hat{U}_\#^\dag.
\end{equation}
where $\hat{G}$ is a complete graph of size $|\mathcal{S}^\prime|$.

\subsection{Complexity of Fundamental Unitaries}
\label{sec:complexity}

Algorithms presented in \cref{sec:cart_qva,sec:perm_qva,sec:methods_and_results} depend on unitary implementations for integer addition, division, and multiplication, as well as comparator and decrement unitaries. These are summarised in \cref{table:quantum_operations}, together with their associated gate complexity and ancillary qubit requirements. For the arithmetic operations, the reported complexity corresponds to Fourier-based methods due to their low ancilla qubit requirements and reversibility (i.e. no ``garbage'' is left in ancilla registers).

\begin{table*}[!ht]
\centering
\resizebox{\linewidth}{!}{%
\begin{tabular}{@{}lcccc@{}}
\toprule
\textbf{Operation} & \textbf{Definition} & \textbf{Register Sizes} & \textbf{Ancilla Qubits} & \textbf{Gate Complexity} \\
\midrule
Addition~\cite{ruiz-perez_quantum_2017} & $\hat{U}_+\ket{a}\ket{b}=\ket{a}\ket{a + b}$ & $\mathcal{O}(r)$ & None & $\mathcal{O}(r^2)$ \\
Division~\cite{khosropour_quantum_2011} & $\hat{U}_{/}\ket{a}\ket{b}\ket{0}=\ket{0}\ket{b}\ket{a/b}$ & $\mathcal{O}(r)$  & $\mathcal{O}(r)$ & $\mathcal{O}(r^3)$ \\
Multiplication~\cite{ruiz-perez_quantum_2017} & $\hat{U}_{\times}\ket{a}\ket{b}\ket{0}=\ket{a}\ket{b}\ket{ab}$ & $\mathcal{O}(r)$ & $\mathcal{O}(r)$ & $\mathcal{O}(r^3)$ \\
Compare~\cite{cuccaro_new_2004} & $\hat{U}_{>}\ket{a}\ket{b}\ket{0}=\ket{a}\ket{b}\ket{a>b}$ & $\mathcal{O}(r)$ & 1 & $\mathcal{O}(r)$ \\
Decrement \cite{douglas2009efficient} & $\hat{U}_{-1}\ket{a}=\ket{a-1}$ & $\mathcal{O}(r)$ & $\mathcal{O}(r)$ & $\mathcal{O}(r)$ \\
\bottomrule
\end{tabular}
}
\caption[Summary of fundamental gate complexity and ancilla requirements]{Summary of register sizes and gate complexity of fundamental operations on unsigned integers represented by $r$ qubits. A quantum encoding of the solution space of a COP with solutions of size $n$ and an alphabet of cardinality $|\mathcal{X}| = m$ requires $r = n \lceil \log m \rceil$ qubits (see \cref{sec:cop}).}
\label{table:quantum_operations}
\end{table*}

\section{Graph Subshells and Phase Discrepancy}
\label{sec:theory}

We now present the first contribution of this work, in which a CTQW description of QVA mixing unitaries that are defined by vertex-transitive graphs is used to develop a compact description of their action on phase-encoded solution states. This description is then used to derive conditions under which the combined action of $\hat{U}_W$ and $\hat{U}_Q$ can produce substantial constructive interference and, therefore, efficient convergence to an arbitrary solution state. We choose to focus on vertex-transitivity as such graphs do not inherently bias the convergence of a QVA and, as demonstrated in \cref{sec:cart_qva,sec:perm_qva,sec:methods_and_results}, encompass many graphs for which an efficient method for the simulation of their CTQW is known.

Let $G$ be a vertex-transitive graph. Relative to an arbitrary choice of reference vertex $\bm{s}$, we delineate a partitioning of its vertices into \emph{subshells}, which are defined as follows. 

\begin{definition}[Graph Subshells]
\label{def:subshells}
Let $\text{Aut}(G, \bm{s})$ be the subgroup of graph automorphisms that permute vertices $\bm{s}^\prime$ in shell $\mathcal{N}_d(\bm{s})$,
\begin{equation}
\text{Aut}(G, \bm{s}) = \{g \in \text{Aut}(G) \mid g(\mathcal{N}_d(\bm{s})) = \mathcal{N}_d(\bm{s})\},
\end{equation}
The subshells of $G$ are the sets of disjoint vertices generated by the orbits of $g \in \text{Aut}(G, \bm{s})$. The $k$-th subshell at a distance $d$ is denoted $\mathcal{N}_{d, k}(\bm{s})$ and the total number of orbits associated with the shell $\mathcal{N}_d(\bm{s})$ as $|\text{Orb}(\text{Aut}(G, \bm{s}))|$.
\end{definition}

The subshell partitioning of a distance-transitive $n=3$ hypercube graph and a vertex-transitive \emph{constrained permutation graph} (introduced in \cref{sec:perm_qva}) is illustrated in \cref{fig:subshells}. Based on this partitioning, we derive a compact expression for the amplitude transfer induced by a CTQW on a vertex-transitive graph. 

\begin{theorem}[Subshell Coefficients]
    \label{th:subshells}
    The amplitudes of a CTQW can be expressed as a linear combination of the complex coefficients of $\ket{\psi_0}$ associated with subspaces,
    \[\{ \ket{\bm{s}^\prime} \mid \bm{s}^\prime \in \mathcal{N}_{d, k}(\bm{s}) \} \subseteq \mathcal{H},\]
     grouped under time-dependant \emph{subshell coefficients} $w_{d,k}(t)$,
    \begin{equation}
    \label{eq:amp_dist}
    \E{\bm{s} | \psi(t)} = \sum_{d=0}^{D} \sum_{k=0}^{K_d-1}w_{d,k}(t)\sum_{\bm{s}^\prime \in \mathcal{N}_{d, k}(\bm{s})} c_{\bm{s}^\prime},
    \end{equation}
    where $D = \text{diam}(G)$, $K_d = |\text{Orb}\left({\text{Aut}(G, \bm{s})}\right)|$ and, 
    \[\ket{\psi_0} = \sum_{\bm{s} \in \mathcal{S}}c_{\bm{s}}\ket{\bm{s}} \]
    with time-independent complex coefficients $c_{\bm{s}}$.
\end{theorem}

\begin{proof}
The projection of $\hat{G}\ket{\psi_0}$ on $\ket{\bm{s}}$ is,
\begin{align*}
\bra{\bm{s}}\hat{G}\ket{\psi_0} &= \bra{\bm{s}} \left( \sum_{\bm{s}, \bm{s}^\prime} w_{\bm{s}, \bm{s}^\prime} \ket{\bm{s}} \bra{\bm{s}^\prime} \right) \sum_{\bm{s}^\prime} c_{\bm{s}^\prime} \ket{\bm{s}^\prime} \\ 
& = \sum_{\bm{s}^\prime} w_{\bm{s}, \bm{s}^\prime} c_{\bm{s}^\prime}.
\end{align*}
where $w_{\bm{s},\bm{s}^\prime} \in \mathbb{R}$.

Let $\pi_g$ be the permutation matrix representation of $g \in \text{Aut}(G, \bm{s})$. As $\pi_g$ and its inverse $\pi_g^{-1} = \pi_g^{T}$ reorder the basis states in a manner that preserves adjacency structure relative to $\bm{s}$, the action of $\hat{G}$ on $\ket{\psi}$ is invariant under conjugation by $\pi_g$. Therefore, 
\begin{align*}
    \bra{\bm{s}}\hat{G}^\prime \ket{\psi} & = \left(\bra{\bm{s}} \pi_g^{T}\right) \hat{G} \left(\pi_g\ket{\psi} \right) = \bra{\bm{s}}\hat{G}\ket{\psi}.
\end{align*}
A result that also holds for $\hat{G}^l$ and $(\hat{G}^\prime)^l$, where $l$ denotes the $l$-th power of the adjacency matrix, by the associativity of the matrix product. Let $w_{d, l, k}$ be the coefficient of $\hat{G}^l$ associated with subshell $\mathcal{N}_{d, k}$ then,

\begin{align*}
   \bra{\bm{s}}\hat{G}^l\ket{\psi} = \sum_{d=0}^{D} \sum_{k=0}^{K_d-1}w_{d, l, k}\sum_{\bm{s}^\prime \in \mathcal{N}_{d, k}} c_{\bm{s}^\prime}
\end{align*}
where $D$, $K_d$ and $c_{\bm{s}}$ are defined as in \cref{eq:amp_dist}. The subshell coefficients $w_{d,k}(t)$ are obtained from this expression by substituting $\hat{G}^l$ with the Taylor series definition of the matrix exponential $\exp(-\text{i} t \hat{G}) = \sum_{l=0}^\infty \frac{\left( -\text{i} t \hat{G} \right)^l}{l!}$. Sets $\{ \ket{\bm{s}^\prime} \mid \bm{s}^\prime \in \mathcal{N}_{d, k}(\bm{s}) \}$ span $\mathcal{N}_{d,k}(\bm{s})$ and are closed under scalar multiplication and addition, hence they are subspaces.
\end{proof}

\begin{figure*}[!ht]
    \centering
    \begin{subfigure}[t]{0.24\textwidth}
        \includegraphics[width=\linewidth]{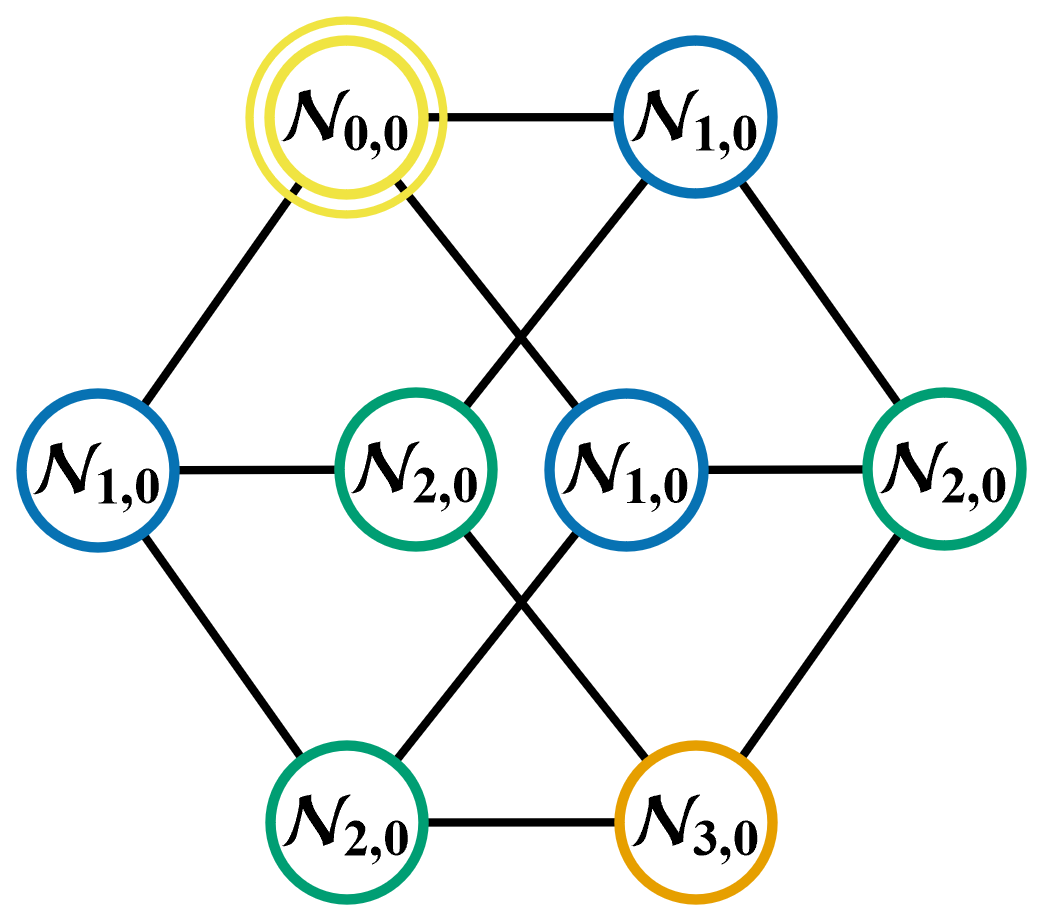}
        \caption{}
    \end{subfigure}
    \hfill
    \begin{subfigure}[t]{0.24\textwidth}
        \includegraphics[width=\linewidth]{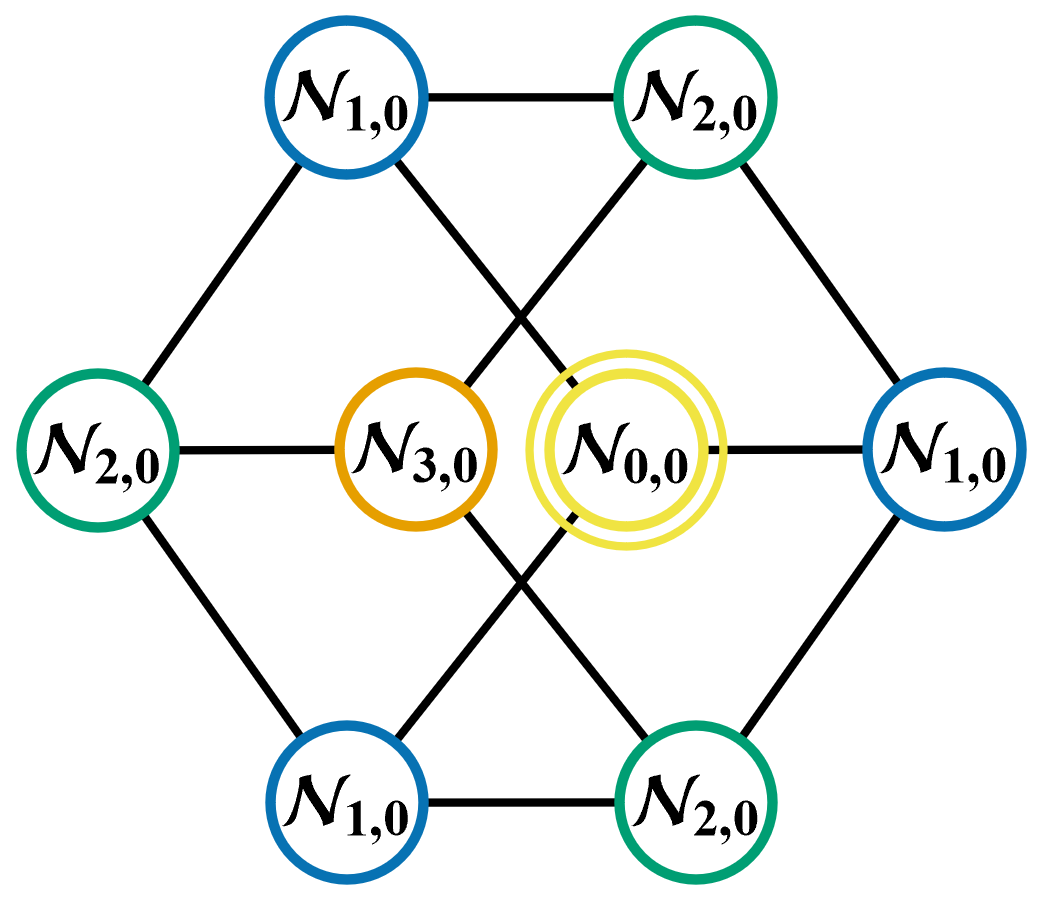}
        \caption{}
    \end{subfigure}
    \hfill
    \begin{subfigure}[t]{0.24\textwidth}
        \includegraphics[width=\linewidth]{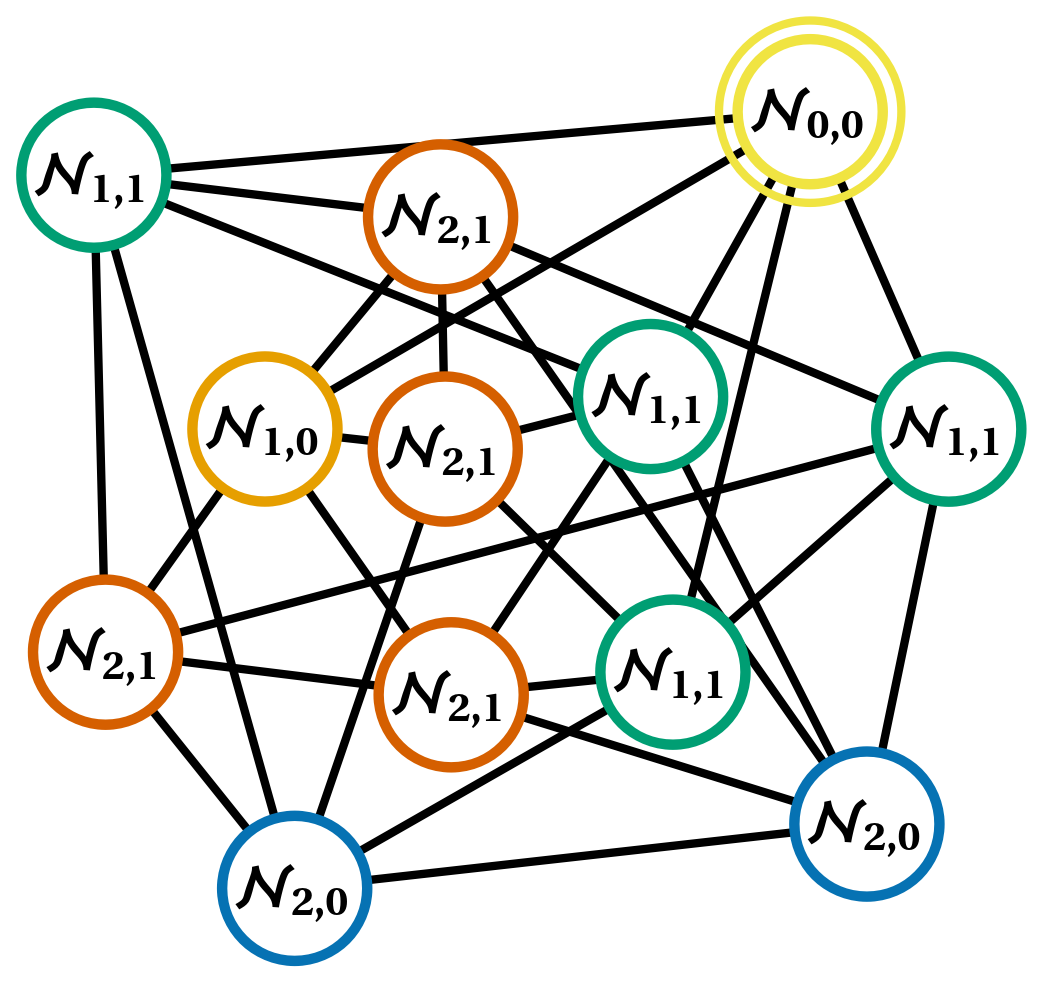}
        \caption{}
    \end{subfigure}
    \hfill
    \begin{subfigure}[t]{0.24\textwidth}
        \includegraphics[width=\linewidth]{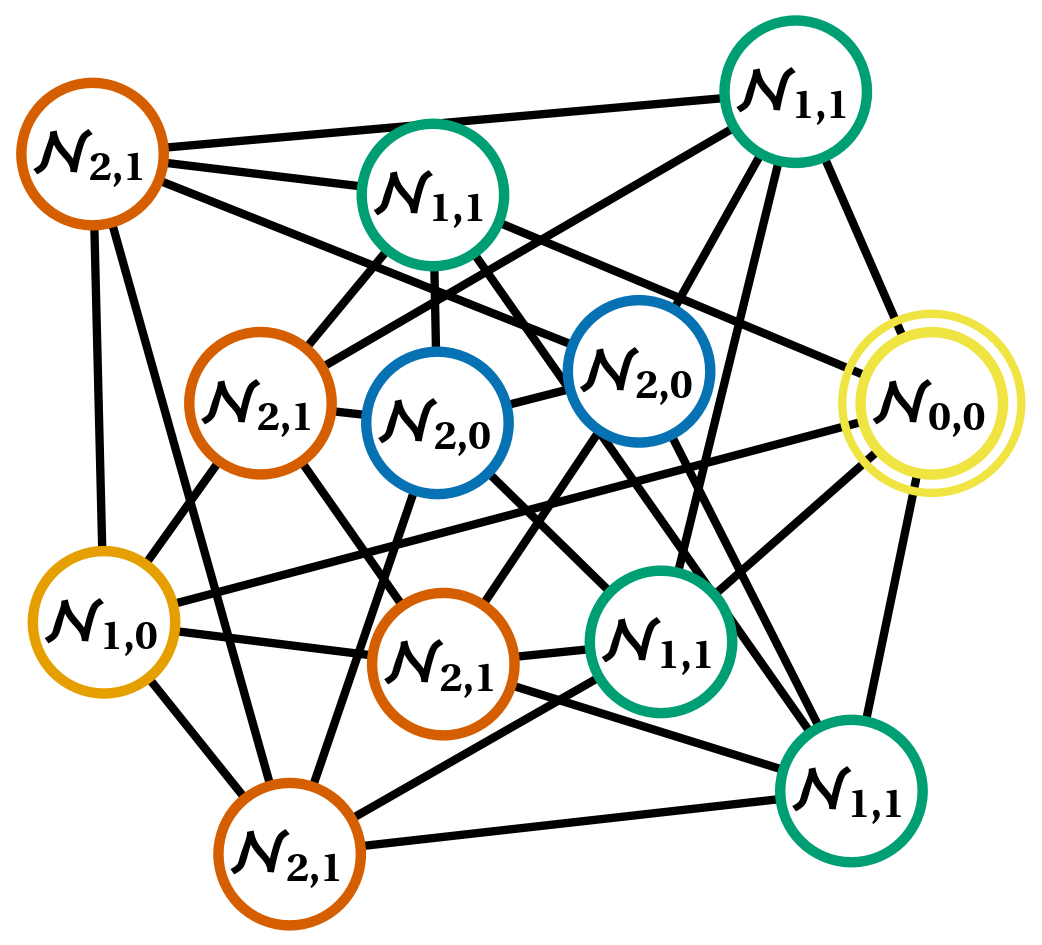}
        \caption{}
    \end{subfigure}
    \caption[Graph subshells in a hypercube and constrained permutation graph]{Graph subshell of (a and b) an $ n=3 $ hypercube graph (or $(3, 2)$ Hamming graph, see \cref{sec:cart_decomp}), and (c and d) a constrained permutation graph (see \cref{sec:perm_graph}) which represents the non-trivial permutations of $(-1, 0, 0, 1)$ that differ by the transposition of two elements. In each graph, subshells relative to a reference vertex $\bm{s}$ are labelled as $\mathcal{N}_{d,k}$, where $d$ is the distance from $\bm{s}$ and $k$ indexes the subshells at distance $d$. The reference vertex $\bm{s}$ is circled twice in yellow and labelled as $\mathcal{N}_{0,0}$ in each graph. Graph automorphisms $g \in \text{Aut}(G, \bm{s})$ (see \cref{def:subshells}) permute vertices within the same subshell, with $g(s) = s$ as $|\mathcal{N}_{0,0}\bm{s}|=1$. As both graphs are vertex-transitive, the adjacency structure relative to any choice of $\bm{s}$ is identical. The hypercube graph has the additional property of distance transitivity, so $\mathcal{N}_d(\bm{s}) = \mathcal{N}_{d,k}(\bm{s})$ where shell $\mathcal{N}_d(\bm{s})$ is the set of vertices at distance $d$ from $\bm{s}$. The constrained permutation graph is not distance-transitive, as the orbits of $g$ partition the vertices in shells $\mathcal{N}_1(\bm{s})$ and $\mathcal{N}_2(\bm{s})$ into subshells $\mathcal{N}_{1,0}(\bm{s})$, $\mathcal{N}_{1,1}(\bm{s})$, $\mathcal{N}_{2,0}(\bm{s})$ and $\mathcal{N}_{2,1}(\bm{s})$. These subshells have a clearly distinct adjacency structure relative to $\bm{s}$; for instance, vertices in $\mathcal{N}_{2,1}(\bm{s})$ are adjacent to the vertex in subshell $\mathcal{N}_{1,0}(\bm{s})$, while the vertices in shell $\mathcal{N}_{2,0}(\bm{s})$ are not.}
    \label{fig:subshells}
\end{figure*}

By this result and \cref{def:subshells}, it also follows that:

\begin{corollary}
A distance-transitive graph $G$ has $\text{diam}(G) + 1$ subshell coefficients.    
\end{corollary}

\begin{proof}
    The graph automorphisms $g \in \text{Aut}(G)$ of a distance-transitive graph $G$ satisfy $g((\bm{s}, \bm{s}^\prime)) \rightarrow (\bm{v}, \bm{v}^\prime)$ for all vertex tuples where $\text{dist}(\bm{s}, \bm{s}^\prime) = \text{dist}(\bm{v}, \bm{v}^\prime)$. Therefore, the number of orbits associated with shell $\mathcal{N}_d(\bm{s})$ is $|\text{Orb}(\text{Aut}(G, \bm{s}))| = 1$ and, by \cref{eq:amp_dist}, the number of subshells and subshell coefficients is $\text{diam}(G) + 1$.  
\end{proof}

We now use the subshell coefficient description of a CTQW to describe conditions in which a quantum search with the ansatz unitary $\hat{U}_W\hat{U}_Q$ may efficiently exploit constructive interference to produce high-convergence to an arbitrary state $\bm{s}$. 

\begin{lemma}
 The constructive interference at $\bm{s}$ due to interaction with $\bm{s}^\prime \in \mathcal{N}_{d,k}$ is maximised by minimising the \emph{phase discrepancy},

    \begin{equation}
        \label{eq:phase discrepancy}
        \begin{aligned}
        \Phi_{d, k}(\gamma, t) =& |\left( \phi_{0,0}(t) + \gamma q_{\bm{s}} \, \, \mathrm{mod} \, \, \pi \right) - \\
        & \left( \phi_{d,k}(t) + \gamma q_{\bm{s}^\prime} \, \, \mathrm{mod} \, \, 2 \pi \right)|,
        \end{aligned}
    \end{equation}
    where $\phi_{d,k}(t) = \text{Arg}(w_{d,k}(t))$.
\end{lemma}

\begin{proof}
Let the state of a CTQW at $t=0$ be,
    \[\ket{\psi_0} = \left(\frac{\hat{U}_Q(\gamma)}{\sqrt{|\mathcal{S}|}}\sum_{\bm{s} \in \mathcal{S}}\ket{\bm{s}} \right).\]
By \cref{eq:amp_dist}, the interaction between subshells $\mathcal{N}_{0,0}(\bm{s})$ and $\mathcal{N}_{d,k}(\bm{s})$ may be expressed as,
    \[ |w_{0,0}(t)|e^{-\text{i} (\phi_{0,0}(t) + \gamma q_{\bm{s}})} + |w_{d,k}(t)|e^{-\text{i}(\phi_{d,k}(t) + \gamma q_{\bm{s}^\prime})}\]
    For any value of $t$ where $|w_{0,0}(t)|$ and $|w_{d,k}(t)|$ are non-zero, the magnitude of this expression is maximised if two terms have the same complex phase. In this case, the absolute difference between the complex arguments of these terms, given by \cref{eq:phase discrepancy}, will be zero. 
\end{proof}

Based on this result, we propose a theorem that underpins the procedure for mixing unitary design presented in \cref{sec:mixer-design}. Let $q_{d,k}$ be the set of quality values $q_{\bm{s}^\prime}$ associated with vertices in subshell $\mathcal{N}_{d, k}(\bm{s})$ and $\overline{\phi}_{d,k}$ be the minimum average phase discrepancy between reference state $\bm{s}$ and $\bm{s}^\prime \in \mathcal{N}_{d,k}(\bm{s})$ obtained by variation of $\gamma$ and $t$.

\begin{theorem}
    \label{th:min_var}
    Minimisation of $\text{Var}(q_{d,k})$ reduces $\Phi_{d,k}$ for non-uniform $C(\bm{s})$.
\end{theorem}

\begin{proof}
    If $\overline{\phi}_{d, k} > 0$ then there must exist solutions in $\mathcal{N}_{d,k}$ for which $\phi_{d,k} > 0$. Denote the set costs associated with these solutions as $q^*$, where $q^* \subseteq q_{d,k}$. As $q^* \rightarrow \overline{q}_{d,k}$ then $\phi_{d_k} \rightarrow 0$. Therefore, $\text{Var}(q_{d,k}) \rightarrow 0$ reduces $\overline{\phi}_{d,k}$. 
\end{proof}

\cref{th:min_var} establishes that we can maximise the amplification potential of a QVA by choosing a graph that minimises, on average, the \emph{subshell variance} $\text{Var}(q_{d,k})$ for arbitrary choice of reference state $\bm{s}$. By \cref{th:subshells} it follows that it will be possible to achieve amplification at a target state due to constructive interference with a large proportion of the quantum search space if the number of subshells is small compared to the cardinality of the solution search space. Therefore, a good choice of $G$ is likely to be a distance-transitive or a vertex-transitive graph where the number of orbits in each of its shells is much smaller than the size of the shell at distances greater than one. We can also infer that the minimisation of subshell variance has the additional benefit of reducing the complexity associated with the optimisation of $\bm{\theta}$ as a (comparatively more) uniform response to changes in $t$ and $\gamma$ over each subshell will manifest as an objective function (see \cref{eq:objective}) with a smoother and more pronounced gradient. 

\section{A Heuristic for Mixing Unitary Design}
\label{sec:mixer-design}

The theoretical framework introduced in \cref{sec:theory} leads to a design heuristic for mixing unitaries that are tailored to specific classes of COPs based on the minimum non-zero Hamming distance between valid solution states. The Hamming distance between solutions is given by,

\begin{equation}
\label{eq:hamming}
H(\bm{s}, \bm{s}^\prime) = |\{ i \mid s_i \neq s_i^\prime, \, 0 \leq i < n \}|,
\end{equation}
where $s_i$, $s_i^\prime$ are the $i$-th elements of $\bm{s}$ and $\bm{s}^\prime$, respectively. The heuristic summarised as follows:

\begin{enumerate}
    \item Identify the minimum non-zero Hamming distance $d_{\text{min}}$ among all pairs of solutions in $\mathcal{S}$:
    \begin{equation}
        d_{\text{min}} = \min \{H(\bm{s}, \bm{s}^\prime) \mid \bm{s}, \bm{s}^\prime, \bm{s} \neq \bm{s}^\prime \}.
    \end{equation}

    \item Define a set of sets $S= \{S_k\}$ where each $S_k$ is a subset of solutions where each solution has at least one other distinct solution in the set with the minimum Hamming distance between them,
    \begin{equation}
        S_k = \{ \bm{s} \in \mathcal{S}^\prime \mid \exists \bm{s}^\prime \in \mathcal{S}^\prime, H(\bm{s}, \bm{s}^\prime) = d_{\text{min}} \}. 
    \end{equation}

    \item For each set $S_k$, define a graph $\mathcal{G}_k$ with vertices $V = S_k$ and edges 
    \[E = \{ (\bm{s}, \bm{s}^\prime) \mid H(\bm{s}, \bm{s}^\prime) = d_{\text{min}}, \, \bm{s}, \bm{s}^\prime \in S_k \}.\]
     The union of all such graphs is represented as $\mathcal{G}$.

     \item With $\mathcal{G}$, define a mixing unitary that preserves the structure of $S$,
     \begin{equation}
        \label{sec:designed_mixer}
        \hat{U}_{\mathcal{G}}(t) = \prod_{k=0}^{|\mathcal{G}| - 1} \exp(-\text{i} t \hat{\mathcal{G}}_k).
     \end{equation}
\end{enumerate}

This procedure follows from the reasoning that, since the problem cost function is polynomially bounded in $|\mathcal{S}|$, for two solutions $\bm{s}$ and $\bm{s}^\prime$ with a Hamming distance of $d_{\text{min}}$, it will generally be the case that
\[|C(\bm{s}) - C(\bm{s}^\prime)| < \Delta C,\]
where $\Delta C$ is the difference between the maximum and minimum solution costs in $C(\mathcal{S})$. Therefore, by defining a correspondence between the Hamming distance between solutions in $\mathcal{S}$ and the inter-node distance in $\mathcal{G}_k$, we increase the likelihood that its subshells are mapped to solutions with similar cost function values, reducing $\text{Var}(q_{d,k})$ (see \cref{th:min_var}) relative to the variance of the space of valid solutions.

The feasibility of this approach depends on the efficient identification of $S$ and algorithms for the quantum simulation of CTQWs on $\mathcal{G}_k \in \mathcal{G}$, with the associated design challenges depending on the particular problem at hand. In \cref{sec:cart_qva,sec:perm_qva}, we detail the development of two efficient QVAs for unconstrained and constrained COPs according to this guiding procedure.

\section{Convergence Potential and Mean Shell Variance}
\label{sec:figures_of_merit}

Here we draw on the analytical framework introduced in \cref{sec:theory} to propose new figures of merit and analysis methods that may be used to evaluate the relative performance of mixing unitaries based on vertex-transitive graphs within the alternating phase-walk QVA framework defined in \cref{sec:cop}.

\subsection{Convergence Potential}
\label{sec:convergence_def}

Consider an artificially constructed scenario in which subshell-associated costs $q_{d, k}$ are selected to minimise the phase discrepancy relative to a specific target vertex $\bm{s}$. The probability of measuring $\bm{s}$ in this scenario provides an upper bound on the convergence that the graph of a given mixing unitary can support under the parameterisation of $t$ and $\gamma$. We define this upper bound as the \emph{convergence potential} of the graph, denoted as $\text{Prob}^*$.

Relative to an initial uniform superposition over the nodes of the graph, $\text{Prob}^*$ is computed by first identifying a walk time $t^*$ that maximises the \emph{population-weighted} sum of the graph subshell coefficients,
\begin{equation}
\label{eq:opt_t}
t^* = \underset{t \in T}{\arg \max} \sum_{d,k} |w_{d,k}(t)| \times |\mathcal{N}_{d,k}|.
\end{equation}
The phase-optimal `cost' $q_{d,k}^*$ for each subshell $\mathcal{N}_{d,k}$ is then obtained by
\begin{equation}
q^*_{d,k} = -\text{Arg}\left(w_{d,k}(t^*)\right),
\end{equation}
which is incorporated into \cref{eq:amp_dist} as,
\begin{equation}
c_{\bm{s}} = \frac{1}{\sqrt{|N_{d,k}|}}e^{-\text{i}q^*_{d,k}},   
\end{equation}
thus aligning the phase of the probability amplitude contribution of each subshell. The convergence potential is then computed as folows,
\begin{equation}
\label{eq:convergence_potential}
\text{Prob}^* = \left|\sum_{d=0}^D\sum_{k=0}^{K_d-1}w_{d,k}(t^*)\left|\mathcal{N}_{d,k}\right|c_{\bm{s}}\right|^2,
\end{equation}
where $D$ and $K_d$ are defined as in \cref{eq:amp_dist}.

Furthermore, computation of the graph convergence potential makes possible the empirical investigation of the stability in convergence supported by vertex-transitive graphs at varying amounts of phase discrepancy $\Phi$. One such approach is to generate \emph{variance-adjusted} $q^*_{d,k}$ according to,
\begin{equation}
\label{eq:variance-adjusted}
\bm{q}^*(d, k, \sigma^2) = \left(q^*_{d,k} + Z_0, q^*_{d,k} + Z_1, \ldots, q^*_{d,k} + Z_{|\mathcal{N}_{d,k}|-1}\right),
\end{equation}
where each $Z_i \in \mathbb{R}$ is independently drawn from a uniform distribution over the interval
\[(-\sqrt{12} \sigma, \sqrt{12}\sigma),\]
which has variance $\sigma^2$. The resulting probability of measuring $\bm{s}$ is then given by $\text{Prob}(\bm{s}) = |\langle{\bm{s}|\bm{\theta}} \rangle|^2$, with
\begin{equation}
\label{eq:variance-adjusted-state}
\ket{\bm{\theta}} = \hat{U}_W(t)\hat{U}_{Q^*}(\gamma)\frac{1}{|V|}\sum_{\bm{s}^\prime \in G} \ket{\bm{s}^\prime},
\end{equation}
where $\hat{U}_W$ is a CTQW over the graph under consideration $G$ and $\hat{U}_{Q^*}$ is a phase-shift unitary with a quality operator $Q^*$ defined by $\bm{q}^*$. Optimal variational parameters $\bm{\theta} = (t, \gamma)$ may be determined using numerical optimisation.

\subsection{Mean Shell Variance}
\label{sec:MSV}

By \cref{th:min_var}, improved convergence behaviour is expected from QVAs whose mixing unitary is defined by a vertex-transitive graph that, on average, minimizes variance over its shells for an arbitrary choice of $\bm{s}$ relative to other vertex-transitive graphs. This \emph{shell variance} may be expressed as
\begin{equation}
    \overline{\sigma^2_{d}} = \frac{1}{|\mathcal{N}_{d}|} \sum_{\bm{s} \in \mathcal{S}} \text{Var}\left(C(\mathcal{N}_{d}(\bm{s}))\right),
\end{equation}
where $C(\mathcal{N}_{d}(\bm{s}))$ is the set of solution costs associated with the shell at distance $d$ from $\bm{s}$ and 
\[\text{Var}(C(\mathcal{N}_{d}(\bm{s}))) = 0,\]
if $|\mathcal{N}_{d}| = 1$. Based on this, we introduce the \emph{Mean Shell Variance} (MSV) as a figure of merit,
\begin{equation}
\label{eq:MSV}
    \text{MSV} = \text{Mean}\left(\overline{\sigma^2_0}, \overline{\sigma^2_1}, \dots, \overline{\sigma^2_D}\right),
\end{equation}
where $D = \text{diam}(G)$ is the diameter of the underlying graph.

The MSV is classically computable only if $|\mathcal{S}|$ is small. However, if $|\mathcal{N}_d|$ is known and $\text{dist}(\bm{s}, \bm{s}^\prime)$ is efficiently computable, it can be approximated accurately from a fixed number of uniformly sampled solutions $\bm{s}$ and their associated costs, given the polynomial bounds on both $|C(\bm{s})|$ and the complexity of its evaluation.

\section{Optimisation Over Cartesian Subspaces}
\label{sec:cart_qva}

We now apply the hueristic for the mixing unitary design outlined in \cref{sec:mixer-design} to unconstrained COPs, idenitfying the partitioning $S$, problem-class-specific graph $\mathcal{G}$ and developing a QVA based on the mixing unitary $\hat{U}_{\mathcal{G}}$ tailored for problems of this type. In the following subsection, we establish that the minimum non-zero Hamming distance $d_\text{min}$ for unconstrained COPs defines a correspondence between the solution space $\mathcal{S}$ and an $n$-ary Cartesian product of the alphabet $\mathcal{X}$, and that the resulting $\mathcal{G}$ is a Hamming graph. This leads to a generalisation of the QMOA in \cref{sec:gen_qmoa} an algorithm for the optimisation of unconstrained COPs with alphabets of arbitrary arity.

\subsection{Cartessian Partitioning of the Solution Space}
\label{sec:cart_decomp}

For unconstrained optimisation, the solution space consists of all possible $n$-combinations of the elements in $\mathcal{X}$. As such, the solutions $\bm{s}^\prime$ that are closest to $\bm{s}$ in terms of the Hamming distance are those that differ by exactly one of their combinatorial variances, therefore, $d_{\text{min}} = 1$. By this we observe that $\mathcal{S}$ is isomorphic to the $n$-ary Cartesian product,
\begin{equation}
\mathcal{S} \cong \mathcal{X}^{\otimes n} = \underbrace{\mathcal{X} \times \mathcal{X} \times \cdots \times \mathcal{X}}_{n\ \text{times}},
\end{equation}
where, relative to $\mathcal{X}^{\otimes n}$, a solution $\bm{s}$ specifies one of $|S| = m^n$ discrete coordinates. As all $\bm{s}^\prime \in \mathcal{S}$ can be reached from $\bm{s}$ by a sequence of at most $n$ solutions that each satisfy $d_{\text{min}}$, the natural choice for $S$ is simply the trivial partitioning $S = \{\mathcal{S}\}$. 

Consider a problem with one combinatorial variable ($n=1$). By step three of the design heuristic, the defining graph of the mixing unitary $\mathcal{G}$ is a complete graph of size $m$, which we denote here as $G_{(1,m)}$. It follows that, for arbitrary $n$, the resulting $\mathcal{G}$ is given by an $n$-dimensional Hamming graph on the set of $m$-tuples~\cite{andries_e_brouwer_spectra_2012}. Formally, the $(n, m)$ Hamming graph is defined by
\begin{equation}
    \mathcal{G} = G_{(n,m)} = G_{(1,m)}^{\square n},
\end{equation}
where $\square n$ is the graph Cartesian product of $G_{(1,m)}$ with itself $n$ times. The Cartesian product of graphs $G$ and $G^\prime$, $G \square G^\prime$, is a graph with the vertex set $V(G) \times V(G^\prime)$ and edges $E(G) \times E(G^\prime)$ where vertices $(v, v^\prime)$ and $(u, u^\prime)$ are adjacent if and only if $v$ and $u$ are adjacent in $G$ or $v^\prime$ and $u^\prime$ are adjacent in $G^\prime$. Consequently, the vertices of $G_{(n,m)}$ bijectively map to $S$, with edges between $\bm{s}$ and $\bm{s}^\prime$ only if they differ in exactly one of their $n$ combinatorial variables.

The $(n, m)$ Hamming graphs are known to be distance-transitive, have diameter $n$, and degree $n(m-1)$~\cite{andries_e_brouwer_spectra_2012}. In \cref{sec:methods_and_results}, we study the convergence potential supported by this family of graphs. For this, it is necessary to establish the size of the shells of such graphs, which we derive in the following lemma.

\begin{lemma}
    \label{th:hamming_shell_size}
    For an $(n, m)$ Hamming graph, the size of shell $\mathcal{N}_{d}$ is
    \begin{equation}
        |\mathcal{N}_{d}| = \binom{n}{d} \times (m - 1 )^d.
    \end{equation}
\end{lemma}

\begin{proof}
    Let $\bm{s}$ and $\bm{s}^\prime$ be two solutions with $|\mathcal{X}| = m$ and a Hamming distance of $d$. The number of combinations of $d$ vertices in which differences can occur is $\binom{n}{d}$. Each non-matching $s^\prime \in \bm{s}^\prime$ can be one of $m-1$ elements in $\mathcal{X}$, with a total of $(m-1)^d$ possible choices over each of the differing $s^\prime$.
\end{proof}

\subsection{The Generalised QMOA}
\label{sec:gen_qmoa}

\begin{figure*}[th!]
\centering
\begin{subfigure}[a]{0.90\linewidth}
      \resizebox{\linewidth}{!}{%
    \begin{quantikz}[column sep = 0.05cm, row sep = -0.35cm]
\lstick[4]{$\ket{0}^{\otimes n \lceil \log m \rceil}$} & \qwbundle[1]{\lceil \log m \rceil} & \phantomgate{HHHH}  & \gate[1][1.8cm][0.8cm]{\hat{U}_{\bm{s}}(m)} \gategroup[wires=5, steps=1, style={dotted, cap=round, inner sep=2pt, column sep = 0.1.5cm}, label style={label position=below, yshift=-0.6cm}]{$\ket{\psi_0}$} & \phantomgate{h} & \push{\,\cdots\,}  & \phantomgate{h} & \ctrl{1} \gategroup[wires=5, steps=2, style={dotted, cap=round, inner sep=2pt, column sep = 0.1.5cm}, label style={label position=below, yshift=-0.6cm}]{$\bm{\theta} = ((\gamma_0, t_0), (\gamma_1,t_1), \dots, (\gamma_{p-1}, t_{p-1}))$} & \gate[1][1.5cm][0.8cm]{\hat{U}_{\mathcal{L}_m}(t_i)}  & \phantomgate{h} & \push{\,\cdots\,}  & \phantomgate{h} & \meter{} & \qw\rstick[4, nwires=3]{$\bra{\bm{\theta}} \hat{Q} \ket{\bm{\theta}}$} \\
\qw & \qwbundle[1]{\lceil \log m \rceil} & \phantomgate{HHHH}  & \gate[1][1.8cm][0.8cm]{\hat{U}_{\bm{s}}(m)} & \qw & \push{\,\cdots\,}  & \qw & \ctrl{1}  & \gate[1][1.5cm][0.8cm]{\hat{U}_{\mathcal{L}_m}(t_i)}  & \qw & \push{\,\cdots\,}  & \qw & \meter{} & \qw \\
\ghost{H} & \ghost{H} & \ghost{H} & \ghost{H} & \ghost{H} & \ghost{H} & \ghost{H} & {\,\vdots\,} & \ghost{H} & \ghost{H} & \ghost{H} & \ghost{H} & \ghost{H} & \ghost{H} \\
\qw & \qwbundle[1]{\lceil \log m \rceil} & \phantomgate{HHHH}  & \gate[1][1.8cm][0.8cm]{\hat{U}_{\bm{s}}(m)} & \qw & \push{\,\cdots\,}  & \qw & \ctrl{1} \vqw{-1} & \gate[1][1.5cm][0.8cm]{\hat{U}_{\mathcal{L}_m}(t_i)}  & \qw & \push{\,\cdots\,}  & \qw & \meter{} & \qw \\
\lstick{$\ket{0}$} & \qwbundle[1]{\lceil n \log m \rceil} & \phantomgate{HHHH}  & \qw & \qw & \push{\,\cdots\,}  & \qw & \gate[1][1.8cm][0.8cm]{\hat{U}_Q(\gamma_i)} & \qw & \qw & \push{\,\cdots\,}  & \qw & \qw & \qw \rstick{$\ket{0}$}
\end{quantikz}
    }  
    \caption{}
\end{subfigure}

\vspace{0.4cm}
\centering
\begin{subfigure}[b]{0.8\linewidth}
        \resizebox{\linewidth}{!}
    {%
       \begin{quantikz}[column sep = 0.15cm, row sep = 0.15cm]
\lstick{$\ket{0}$} & \gate{H} & \gate[wires=4,nwires=3]{0 \leq i < m}\vqw{4} & \qw & \gate[wires=4,nwires=3]{0 \leq i < m}\vqw{4} & \gate{H} & \octrl{1} & \qw & \octrl{1} & \gate{H} & \qw\rstick[4]{$\frac{e^{\text{i}\phi_s^\prime}}{\sqrt{m}}\sum_{i=0}^{m-1}\ket{i}$} \\
\lstick{$\ket{0}$} & \gate{H} & \qw & \qw & \qw & \gate{H} & \octrl{1} & \qw & \octrl{1} & \gate{H} & \qw \\
\ghost{H} & \ghost{H} & \ghost{H} & \ghost{H} & \ghost{H} & \ghost{H} & {\,\vdots} & \ghost{H} & {\,\vdots} & \ghost{H} & \ghost{H} \\
\lstick{$\ket{0}$} & \gate{H} & \qw & \qw & \qw & \gate{H} & \octrl{1} \vqw{-1} & \qw & \octrl{1} \vqw{-1} & \gate{H} & \qw \\
\lstick{$\ket{0}$} & \qw & \targ{} & \gate{R_z(\phi_s)} & \targ{} & \qw & \targ{} & \gate{R_z(\phi_s)} & \targ{} & \qw & \qw\rstick{$\ket{0}$}
    \end{quantikz}
        }
        \caption{$\hat{U}_{\bm{s}}(m)$.}
\end{subfigure}

\vspace{0.4cm}
\centering
\begin{subfigure}[c]{0.8\linewidth}
       \resizebox{\linewidth}{!}{%
    \begin{quantikz}[column sep = 0.15cm, row sep = 0.15cm]
    \lstick[4]{$\ket{\psi}$} & \gate[wires=4,nwires=3]{\hat{U}_{\bm{s}}(m)} & \qw & \octrl{1} & \qw & \octrl{1} & \qw & \gate[wires=4,nwires=3]{\hat{U}_{\bm{s}}(m)^\dagger} & \qw\rstick[4]{$e^{-\text{i}t\hat{\mathcal{L}}_m}\ket{\psi}$}\\
    \qw & \qw & \qw & \octrl{1} & \qw & \octrl{1} & \qw & \qw & \qw\\
    \ghost{H} &  \ghost{H} &  \ghost{H} & {\,\vdots\,} & \ghost{H} & {\,\vdots\,} & \ghost{H} & \ghost{H} & \ghost{H} \\
    \qw & \qw & \gate{R_z(-\frac{m t}{2})} & \targ{} \vqw{-1}\vqw{1} & \gate{R_z(-\frac{m t}{2})} & \targ{} \vqw{-1}\vqw{1} & \gate{R_z(m t)} & \qw & \qw \\
    \lstick{$\ket{0}$} & \qw & \qw & \targ{} & \gate{R_z(-\frac{m t}{2})} & \targ{} & \qw & \qw & \qw \rstick{$\ket{0}$}
    \end{quantikz}
    }
    \caption{$\hat{U}_{\hat{\mathcal{L}_m}}(t)$}
    \end{subfigure}
    \caption{(a) Representative circuit diagram of the generalised QMOA. The initial state $\ket{\psi_0}$ is prepared as an equal superposition over the first $m$ computational basis states of $n$ registers that encode the combinatorial variables of $\bm{s}$. The parameterised solution costs $\gamma C(\bm{s})$ are phase-encoded to $r$ bits of precision on a register of $\mathcal{O}(r)$ qubits~\cite{childs2004quantum}. (b) $\hat{U}_{\bm{s}}$ prepares a uniform superposition over the first $m$ computational basis states~\cite{yoder2014fixed}, where $\phi_s = 2 \arcsin\left( \sqrt{\frac{2^m}{4 m}} \right)$, $\phi_s^\prime = \pi - \arccos\left( \sqrt{\frac{2^m}{4 m}} \right)$, and $H$ is a Hadamard gate. A CTQW on the Laplacian of an $(n, m)$ Hamming graph is simulated by $\left(\hat{U}_{\mathcal{L}_m}(t_i)\right)^{\otimes n}$. (c) Simulates a CTQW on the Laplacian of a complete graph that fully connects the first $m$ basis states~\cite{yoder2014fixed}.}
    \label{qc:qmoa}
\end{figure*}
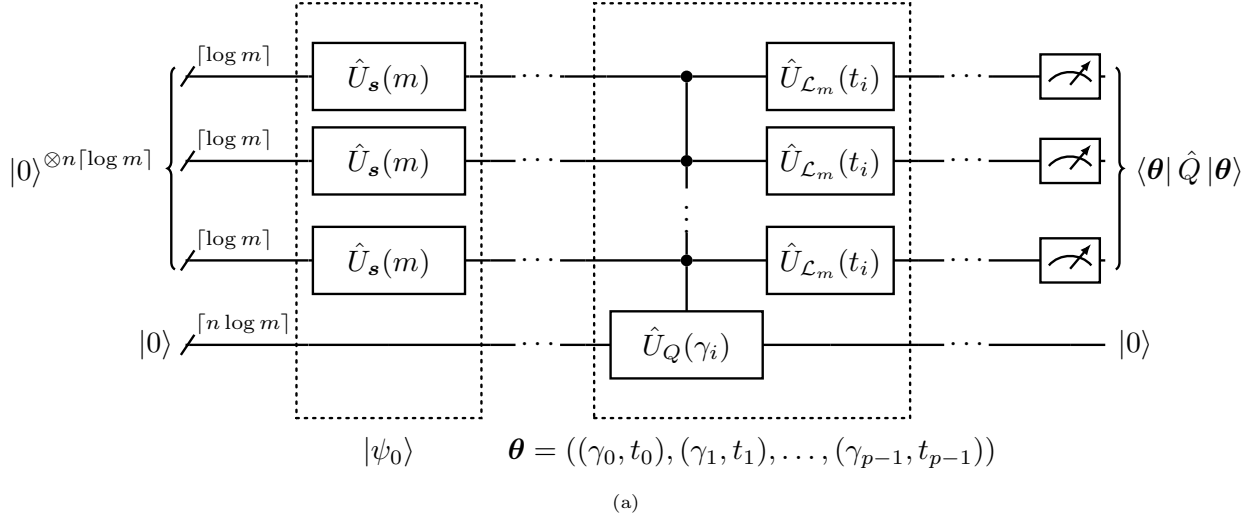
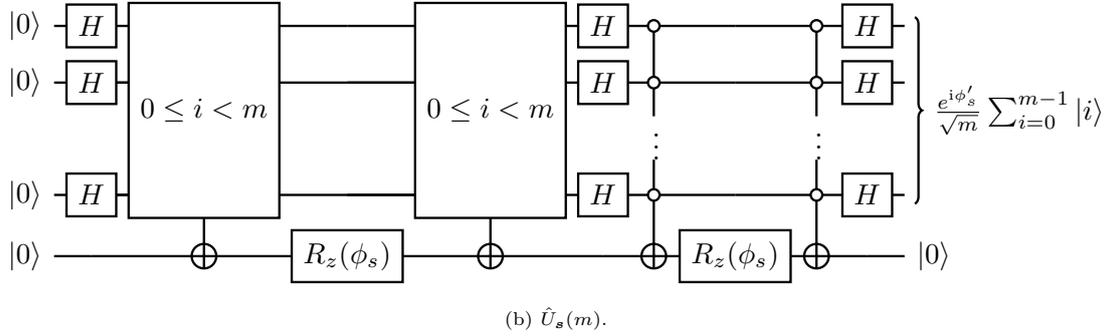
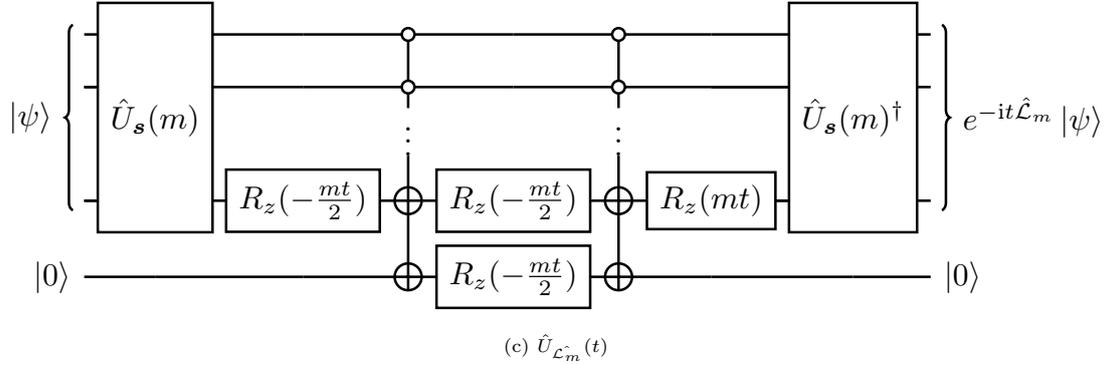

Defining $\mathcal{G}$ as the $(n, m)$ Hamming graph, arrives at a special case of the QMOA mixing unitary (see \cref{eq:qmoa_mixer}) in which the mixer for each coordinate is the complete graph, $\bm{t}$ specifies the same walk time $t$ in each coordinate, and the number of points in each dimension is $m$. Given a solution space where its size $|\mathcal{S}| = m^n$ is equal to an integer power of two, the initial state $\ket{\psi_0}$ may be prepared by applying a Hadamard gate to each of the $r = n\lceil \log m \rceil$ qubits that encode the solution space and the QMOA applied without modification. 

The original QMOA assumes that $m$ is equal to an integer power of two. For unconstrained COPs where this is not the case, we introduce the \emph{generalised} QMOA. This variant of the QMOA can be realised using the modulo $m$ quantum Fouier transform. However, the quantum Fourier transform has been found to be susceptible to noise, which may lead to a leakage of probability amplitude to states outside of $\mathcal{S}$. For this reason, we introduce an alternative method based on the work of Bennett et al. for the exact simulation of a CTQW on the Laplacian of a complete graph with $m$ vertices $\mathcal{L}_m$~\cite{bennett_quantum_2021}. It makes use of circuits developed for a fixed-point quantum search~\cite{yoder2014fixed}, which are informed by the following expression of the unitary of a CTQW on $\mathcal{L}_m$ up to global phase,
\begin{equation}
    \hat{U}_{\mathcal{L}_m}(t) = \mathcal{I} - (1 - e^{-\text{i} m t})\ket{\psi_m}\bra{\psi_m},
\end{equation}
which describes the walk as a rotation about an equal superposition over $m$ states $\ket{\psi_m}=\sum_{i,j=0}^{m-1}\ket{i}\bra{j}$. As the problem formulation in \cref{sec:cop} specifies that the alphabet $\mathcal{X}$ is shared by each of the combinatorial variables $n$, a CTQW on the Laplacian of an $(n,m)$ Hamming graph is exactly simulated by $\left(\hat{U}_{\mathcal{L}_m}(t)\right)^{\otimes n}$, which share the same global phase.

For arbitrary $m$, $\ket{\psi_m}$ is prepared by $\hat{U}_{\bm{s}}(m)\ket{0}$ where $\hat{U}_{\bm{s}} = H^{\otimes r}\hat{U}_0H^{\otimes r}\hat{U_1}H^{\otimes r}$ with,
\begin{equation}
\hat{U}_0=\left\{\begin{array}{ll}
e^{\text{i} \phi}\ket{i} & \text{if } 0 \leq i < m, \\
\ket{i} & \text{otherwise},
\end{array}\right.
\end{equation}
and,
\begin{equation}
\hat{U}_1=\left\{\begin{array}{ll}
e^{\text{i} \phi}\ket{i} & \text{if } i=0, \\
\ket{i}, & \text{otherwise}
\end{array}\right.
\end{equation}
where $\phi_s=2 \arcsin\left( \sqrt{\frac{2^m}{4 m}} \right)$. The initial state for the generalised QMOA is then obtained by $\left(\ket{\psi_0}=\hat{U}_{\bm{s}}(m)\right)^{\otimes n}$. Circuit diagrams illustrating the implementation of the generalised QMOA in this manner are shown in \cref{qc:qmoa}, which together have $\mathcal{O}(n^2)$ gate complexity~\cite{yoder2014fixed}. Parallel execution of CTQWs on $\mathcal{L}_m$ requires $n$ ancilla qubits; alternatively, sequential execution of $\hat{U}_{\mathcal{L}_m}$ may be performed with a single ancilla at the expense of increased circuit depth.

We note that this formulation frames the QMOA as an extension of the QAOA to higher-order Hamming graphs, as the hypercube graph (see \cref{eq:qmoa_mixer}) is exactly the $(n, 2)$ Hamming graph. The QAOA mixing unitary is obtained from the QMOA mixing unitary (see \cref{eq:qmoa_mixer}) if $m = 2$, in which case $\left(\mathcal{F}^{-1}\right)^{\otimes n}$ and $\mathcal{F}^{\otimes n}$ reduce to a tensor product of Hadamard gates $H^{\otimes n}$, and $\exp\left(-\text{i} t \hat{\Lambda}(t)\right) = \exp\left(-\text{i} t \left(\ket{0}\bra{0} - \ket{1}\bra{1}\right)^{\otimes n}\right)$, which together reduce to the well-known circuit for the QAOA mixing unitary, applying a controlled $X$ rotation to each of the $n$ qubits~\cite{medvidovic2021classical}.

\section{Optimisation Over Permutation Subspaces}
\label{sec:perm_qva}

This section applies the heuristic for mixing unitary design to constrained optimisation under integer equality constraints as defined in \cref{sec:cop}. In \cref{sec:multiset_partitioning}, we demonstrate that the heuristic defines a partitioning $S$ of the feasible solution space $\mathcal{S}$ into permutation sets, from which a subset that corresponds to the valid solution space $\mathcal{S}^\prime$ is efficiently identifiable for all COPs of this type. Based on this result, we build on the QWOA's indexing strategy (see \cref{sec:qwoa}), developing the Quantum Walk-based Optimisation Algorithm over Combinatorial Subsets (QWOA-CS), which searches over a canonical subspace of non-degenerate valid solutions. \cref{sec:indexing_perm} describes efficient algorithms for indexing and unindexing $\mathcal{S}^\prime$ in a manner that preserves the structure of the multiset partitioning. The QWOA-CS is introduced in \cref{sec:qwoa-cs}. This QVA employs indexing together with a mixing unitary that is defined by two \emph{submixers}. The first of these acts within the disjoint sets of $S$ and is defined by the graph $\mathcal{G}$ that follows from the design heuristic. The second submixer, which drives the transfer of probability amplitude between disjoint valid sets, is defined by a $K$-partite graph. We describe the properties of the disjoint subgraphs of $\mathcal{G}$ in \cref{sec:perm_graph} and outline the efficient implementation of the QWOA-CS mixing unitary in \cref{sec:perm_ctqw,sec:k-partite-graph}.

\subsection{Multiset Partitioning of the Solution Space}
\label{sec:multiset_partitioning}

Consider a solution $\bm{s}$ that satisfies the constraint $h(\bm{s}) = A$ (see \cref{eq:constraint}). Since $h$ is the sum of identical one-to-one mappings from $\mathcal{X}$ to $\mathbb{Z}$, changing a single element in $\bm{s}$ will result in an invalid solution. However, because $h$ depends only on the number of each $x_j \in \mathcal{X}$ in $\bm{s}$, permutations of $\bm{s}$ will also be valid. Hence, a valid solution with a minimum non-zero Hamming distance from $\bm{s}$ is obtained by transposition of two of its elements, with $d_{\text{min}} = 2$.

Unlike the unconstrained case, $S$ and, consequently, $\mathcal{G}$ are disjoint since by transposition we can only obtain solutions that are related by permutation. Therefore, it is necessary to identify a union of $S_k \in S$ that is equal to the valid solution space $\mathcal{S}^\prime$ to initialise a superposition that fully covers the problem search space. In the following, we show that this is achievable with polynomial complexity in $n$ by establishing a correspondence between $S$ and a set of generating multisets.

Let $\mathcal{P}_i = (\mathcal{X}, P_i)$ be a multiset of cardinality $n$, where $P_i: \mathcal{X} \rightarrow \mathbb{Z}^+$ specifies the multiplicity $x_j \in \mathcal{X}$ in the multiset. We notate the image of $P_i$ as, 
\begin{equation}
 \bm{P}_i = (P_{i,x_0}, P_{i,x_1}, \dots, P_{i,x_{m-1}})
\end{equation}
where $P_{i,x_j}$ is the multiplicity of $x_j$. The number of possible $\mathcal{P}_i$  is equal to the number of weak compositions of $n$ into $m$ non-negative integers~\cite{heubach_2009},
\begin{equation}
\label{eq:num_multisets_1}
|\mathcal{P}| = \binom{{n + m - 1}}{m}.
\end{equation}
where $\mathcal{P}$ is the set of all $\mathcal{P}_i$ with cardinality $n$. For constant $m$, \cref{eq:num_multisets_1} can be expressed as, 
    \begin{equation}
    \label{eq:num_multisets}
        |\mathcal{P}| = \frac{1}{m-1} (n + m -1)(n + m -2)\times\dots\times(n + 1)
    \end{equation}
using the multiplicative formula for the binomial coefficients~\footnote{The multiplicative formula for the binomial coefficients is, \[\binom{n}{k} = \frac{n \times (n-1) \times \dots \times (n - k + 1)}{k \times (k - 1) \times \dots \times 1}.\]} as a polynomial in $n$ of degree $m-1$.

As the set of permutations of a multiset $\text{Perm}(\mathcal{P}_i)$ contains every unique ordering of its members, the cardinality of $\text{Perm}(\mathcal{P}_i)$ is given by the multinomial coefficient~\cite{bener_2015},
\begin{equation}
\label{eq:multinomial}
|\text{Perm}(\mathcal{P}_i)| =  \frac{n!}{P_{i,x_0}! P_{i,x_1}! \dots P_{i,x_{m-1}}!}.
\end{equation}

Recall the linear equality constraint function $h$ defined in \cref{eq:constraint}. We establish the equivalence of its image on $\mathcal{S}$ and $\mathcal{P}$ in the following lemma.

\begin{lemma}
\label{th:multiset_covers}
The space of multisets $\mathcal{P}$ with cardinality $n$ provides a domain for the constraint function $h$ that covers the set of unique constraint values,
\[h(\mathcal{S}) = \{ A \mid A \in h(\bm{s}), \bm{s} \in \mathcal{S} \}.\]
\end{lemma}

\begin{proof}
As $h(\bm{s})$ is invariant under reordering of its terms $z(s_i)$, it is invariant under permutation of $\bm{s}$. Therefore, $h(\mathcal{S}) \equiv h(\mathcal{P})$.
\end{proof}

\begin{corollary}
The cardinality of $\mathcal{A}$ is upper bound by \cref{eq:num_multisets}.
\end{corollary}

\begin{proof}
As $h(\mathcal{S}) \equiv h(\mathcal{P})$, the image of $h(\mathcal{P})$, it will be maximised when $h$ is a bijection from $\mathcal{P}$ to $h(\mathcal{P})$.
\end{proof}

By these results and \cref{eq:num_multisets}, we establish that the set of multisets $\mathcal{P}^\prime = (\mathcal{P}_0, \dots, \mathcal{P}_{K-1})$ that satisfy a given constraint value $A$ is identifiable with a worst-case polynomial-time complexity of $\mathcal{O}(n^{m-1})$ and that, for any $A$, there exists a partitioning of $\mathcal{S}^\prime$ into permutation sets 
\[S = (S_0, S_1, \dots, S_{K-1})\] 
where $S_k = \text{Perm}(\mathcal{P}_k)$ and 
\begin{equation}
\mathcal{P}_k \in \{\mathcal{P}_i \mid h(\mathcal{P}_i) = A\}.    
\end{equation}
We refer to the disjoint $\mathcal{G}_k$ resulting from the design heuristic and this partitioning as the \emph{constrained permutation} graphs. 

\subsection{Indexing Algorithms for Permutation Sets}
\label{sec:indexing_perm}

To develop an efficient scheme for the lexicographic indexing of the multiset partitioning described in \cref{sec:multiset_partitioning}, we combined an indexing of the valid generating multisets $\mathcal{P}^\prime$ with a local index for $\bm{s} \in S_k$. The complexity associated with an indexing $\mathcal{P}^\prime$ is problem-specific. However,  by \cref{eq:num_multisets}, identification of the valid multisets for a given constraint value $A$ requires $\mathcal{O}(n^{m-1})$ queries to the constraint function $h$ in the worst case. Thus the indexing of $\mathcal{P}^\prime$ is efficient relative to the size of the solution space in general.

Indexing within $S_k$ follows from the observation that the number of $\bm{s} \in S_k$ starting with the same $x_i$ is given by the multinomial coefficient for $n - 1$ with the remaining counts of the unplaced values (i.e., all possible configurations for the remaining values). The resulting indexing algorithm $\text{id}_k$ recursively computes a lexicographically ordered index for each $s$ in $\bm{s}$ based on its position in $\bm{s}$ and the multiplicities given by $\bm{P}_k$. As shown in \cref{alg:index_sj}, it computes a prefix for each $s$ equal to the cumulative permutations occurring before it at the same position, which is calculated iteratively based on the encountered symbol counts within the sequence. The unindexing algorithm, shown in \cref{alg:unindex}, denoted $\text{id}_k^{-1}$, essentially reverses this indexing process.

An essential operation in the indexing and unindexing algorithms is the computation of multinomial coefficients. This can be achieved using the following factorisation of the multinomial coefficients as a product of binomial coefficients~\cite{bener_2015}, 
\begin{equation}
\binom{n}{P_{k,0}, P_{k,1}, \ldots, P_{k,m-1}} = \prod_{i=0}^{m-2} \binom{n - \sum_{j=0}^{i-1} P_{k,j}}{P_{k,i}}.
\end{equation}
A quantum circuit illustrating this process for $\binom{n}{{P_{k,x_0}, P_{k,x_1}, P_{k,x_2}}}$, which corresponds to a problem instance where $|X| = 3$, is shown in \cref{qc:multinomial}. Registers $\ket{b_0}$, $\ket{b_0}$ and $\ket{b_0}$ are of size $\mathcal{O}(\log m)$; consequently, the binomial coefficients are computable with $\mathcal{O}(n\log m)$ gate complexity~\cite{marsh_combinatorial_2020}. Multiplication of $n \log m$ bit integers is the most intensive operation in the circuit, leading to an overall gate complexity of $\mathcal{O}(n^3 \log^3 m)$. The equivalent circuit for any $n$ and $m$ requires $\mathcal{O}(m \log n)$ ancillia qubits and $\mathcal{O}(n)$ multiplications and divisions. As \cref{alg:index_sj,alg:unindex} have a recursion depth of $n$ and evaluate at most $m$ conditional expressions per call, they have a gate complexity of $\mathcal{O}(n^4 \log^4 m)$.

\begin{figure*} 
\centering
\resizebox{\linewidth}{!}{%
\begin{quantikz}[column sep = 0.1cm, row sep = 0.2cm]
        \lstick{\ket{b_0}} & \qwbundle{ } & \phantomgate{hhh} & \gate[4][2cm]{\hat{U}_+}\gateinput{$x_0$}\gateoutput{$x_0$} & \qw & \qw & \qw & \qw & \qw & \qw & \qw & \gate[4][2cm]{\hat{U}_+}\gateinput{$x_0$}\gateoutput{$x_0$}& \qw\rstick{\ket{b_0}} & \\
        \lstick{\ket{b_1}} & \qwbundle{} & \phantomgate{h} & \gateinput{$x_1$}\gateoutput{$x_1$} & \qw & \gate[5][1.6cm]{\mathcal{B}}\gateinput{$x_0$}\gateoutput{$x_0$} & \qw & \qw & \qw & \gate[5][1.6cm]{\mathcal{B}}\gateinput{$x_0$}\gateoutput{$x_0$} & \qw & \gateinput{$x_1$}\gateoutput{$x_1$} & \qw & \qw\rstick{\ket{b_1}} & \\
        \lstick{\ket{b_2}} & \qwbundle{} & \phantomgate{h} & \qw & \gate[3][2cm]{\hat{U}_+}\gateinput{$x_0$}\gateoutput{$x_0$} &  \qw & \qw & \qw & \qw & \qw & \gate[3][2cm]{\hat{U}_+}\gateinput{$x_0$}\gateoutput{$x_0$} & \qw & \qw\rstick{\ket{b_2}} & \\
        \lstick{\ket{0}} & \qwbundle{} & \phantomgate{h} & \gateinput{$y$}\gateoutput{${(x_0 {+} x_1)} {\oplus} y$} & \gateinput{$x_1$}\gateoutput{$x_1$} & \gateinput{$x_1$}\gateoutput{$x_1$} & \gate[4][1.6cm]{\mathcal{B}}\gateinput{$x_0$}\gateoutput{$x_0$} & \qw & \gate[4][1.6cm]{\mathcal{B}}\gateinput{$x_0$}\gateoutput{$x_0$} & \gateinput{$x_1$}\gateoutput{$x_1$} & \gateinput{$x_1$}\gateoutput{$x_1$} & \gateinput{$y$}\gateoutput{${(x_0 {+} x_1)} {\oplus} y$} & \qw\rstick{\ket{0}} & \\
        \lstick{\ket{0}} & \qwbundle{} & \phantomgate{h} & \qw & \gateinput{$y$}\gateoutput{${(x_0 {+} x_1)} {\oplus} y$} & \qw & \gateinput{$x_1$}\gateoutput{$x_1$} & \qw & \qw & \qw & \gateinput{$y$}\gateoutput{${(x_0 {+} x_1)} {\oplus} y$} & \qw & \qw\rstick{\ket{0}} & \\
        \lstick{\ket{0}} & \qwbundle{} & \phantomgate{h} & \qw & \qw & \ghost{H}\gateinput{$y$}\gateoutput{$\binom{x_1}{x_0} {\oplus} y$} & \qw & \gate[3][2cm]{\hat{U}_\times}\gateinput{$x_0$}\gateoutput{$x_0$} & \gateinput{$x_1$}\gateoutput{$x_1$} 
 & \ghost{H}\gateinput{$y$}\gateoutput{$\binom{x_1}{x_0} \oplus y$} & \qw & \qw & \qw\rstick{\ket{0}} & \\
        \lstick{\ket{0}} & \qwbundle{} & \phantomgate{h} & \qw & \qw & \qw & \ghost{H}\gateinput{$y$}\gateoutput{$\binom{x_1}{x_0} {\oplus} y$} & \gateinput{$x_1$}\gateoutput{$x_1$} & \ghost{H}\gateinput{$y$}\gateoutput{$\binom{x_1}{x_0} {\oplus} y$} & \qw & \qw  & \qw & \qw\rstick{\ket{0}} & \\
        \lstick{\ket{0}} & \qwbundle{} & \phantomgate{h} & \qw & \qw & \qw & \qw & \gateinput{$y$}\gateoutput{$(x_0 {\times} x_1) {\oplus} y$} & \qw & \qw & \qw & \qw & \qw\rstick{\ket{\mathcal{M}}} &
\end{quantikz}
}
    \caption[Circuit implementing computation of the multinomial coefficient]{Circuit for computation of the multinomial coefficient $\mathcal{M} = \binom{{b_0 + b_1 + b_2}}{{b_0, b_1, b_2}}$ as the product of two binomial coefficients.}
    \label{qc:multinomial}
\end{figure*}
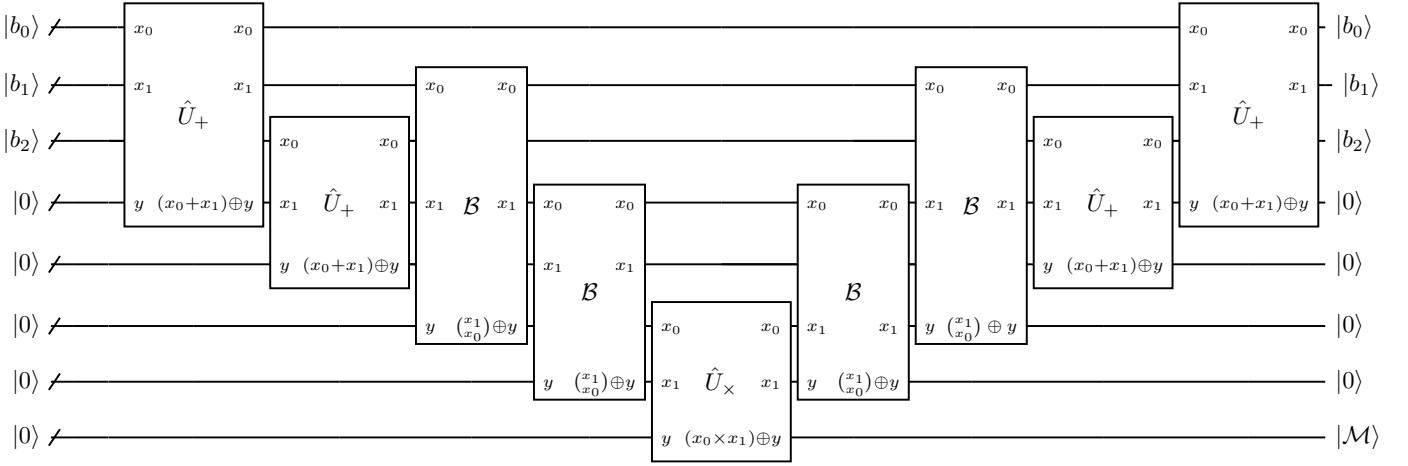

\begin{algorithm2e}[t]
\caption[QWOA-CS indexing algorithm for permutation sets]{$\text{id}_k(\bm{s})$ (QWOA-CS)}
\label{alg:index_sj} \SetKwFunction{subindex}{\textsc{IDK}} \SetKwProg{func}{function}{}{}
  \func{\subindex{$\bm{s}$}}{%
    $n \gets \text{size}(\bm{s})$\;
    $\bm{P}_k \gets (\text{count}(\bm{s}, x_0), \dots, \text{count}(\bm{s}, x_{m-1}))$\;
    \lIf{$n = 0$}{$\Return \; 0$}
    nPerm $\gets \binom{n}{{P_{k, 0}, \dots, P_{k, m-1}}}$\;
    $s \gets \bm{s}.\text{pop}(0)$ \;
    $\text{prefix} \gets 0$\;
    \ForEach{$x_i \text{ in } X$}{%
    \uIf{$s = x_i$}{%
      \Return  $\text{nPerm}\lfloor \frac{\text{prefix}}{n} \rfloor + \subindex{$\bm{s}$}$\;
    }
    $\text{prefix} \gets \text{prefix} + \bm{P}_{k,j}$\;
    }
    
  }
\end{algorithm2e}
\begin{algorithm2e}[t]
\caption[QWOA-CS unindexing algorithm for permutation sets]{$\text{id}_k^{-1}(i)$ (QWOA-CS)}
\label{alg:unindex}
\SetKwFunction{unindex}{\textsc{UNIDK}}
\SetKwProg{func}{function}{}{}
  \func{\unindex{$i, \bm{P}_k, \bm{s}$}}{%
    $n \gets \sum_{j=0}^{m-1}P_{k,j}$\;
    \If{$n = 0$}{\KwRet{$\bm{s}$}\;}
    nPerm $\gets \binom{n}{{P_{k, 0}, \dots, P_{k, m-1}}}$\;
    \ForEach{$P_{k,j} \text{ in } \bm{P}_k$}{
    $\text{prefix} \gets \text{nPerm}\lfloor \sum_{l=0}^{l \leq j} P_{k,l} \rfloor$\;
    \If{$i < \text{prefix}$}{
        $s.\text{append}(x_j)$\;
        $\bm{P}_{k,j} \gets \bm{P}_{k, j} - 1$\;
        $i \gets i - \text{prefix}$\;
        \KwRet{\unindex{$i, \bm{P}_k, \bm{s}$}}\;
    }
    }
  }
\end{algorithm2e}

We arrive at an indexing that preserves the structure of $S$ by defining the index offset function,
\begin{equation}
f(k) = 
\begin{cases} 
0 & \text{if } k = 0 \\
\sum_{i=0}^{k-1} |S_i| & \text{if } k > 0 
\end{cases}
\end{equation}
which computes the number of solutions based on the size of the $|S_k|$ with a lower index. An indexing algorithm for $\mathcal{S}^\prime$ with the desired grouping is then given by
\begin{equation}
    \label{eq:index}
    \text{id}(\bm{s}, k) = f(k) + \text{id}_k(\bm{s}),
\end{equation}
and the inverse unindexing function by
\begin{equation}
    \label{eq:unindex}
    \text{id}^{-1}(\text{id}(\bm{s}, k), k) = \text{id}_k^{-1}(\text{id}(\bm{s}, k) - f(k)),
\end{equation}   
where the time complexity of these algorithms is determined by the fastest-growing complexity among those of $\text{id}_k$, $\text{id}_k^{-1}$, and $f(x)$, all of which are $\mathcal{O}(\text{poly } n)$.

\subsection{The QWOA on Combinatorial Subsets}
\label{sec:qwoa-cs}

The Quantum Walk-based Optimisation Algorithm on Combinatorial Subsets (QWOA-CS), summarised in ~\cref{qc:qwoa-cs,fig:qwoa-cs-coupling}, performs a hierarchical search based on the multiset partitioning of $S$ described in \cref{sec:multiset_partitioning} and the multiset indexing scheme detailed in \cref{sec:indexing_perm}. 

As shown in \cref{qc:qwoa-cs-index}, the QWOA-CS indexing unitary $\hat{U}_{\#CS}$ introduces an additional register of size $\mathcal{O}(\log |\mathcal{P}^\prime|)$ that defines the multiset index $k$ of $\text{id}_k$. An index local to $S_k$ is first computed by,    
\begin{equation}
\label{eq:index_1}
\hat{U}_{\text{id}_k}\hat{\mathcal{S}}_{0,2}\hat{U}_{\text{id}_k}\ket{\bm{s}}\ket{k}\ket{0} = \ket{\text{id}_k(\bm{s})}\ket{k}\ket{0},
\end{equation}
where $\hat{\mathcal{S}}_{0,2}$ is a swap operation between the left and right registers. The unitary $\hat{U}_{f_k}$ then applies the index-offset,
\begin{equation}
\label{eq:index_2}
    \hat{U}_{f_k}\ket{\text{id}(\bm{s})}\ket{k} = \ket{f(k) + \text{id}(\bm{s})}\ket{k} = \ket{\text{id}(\bm{s}, k)}\ket{k},  
\end{equation}
For brevity, we refer to $\ket{\text{id}(\bm{s},k)}$ as $\ket{\text{id}(\bm{s})}$. Construction of the QWOA-CS unindexing unitary $\hat{U}_{\#CS}^\dagger$ naturally follows as the inverse of this procedure.

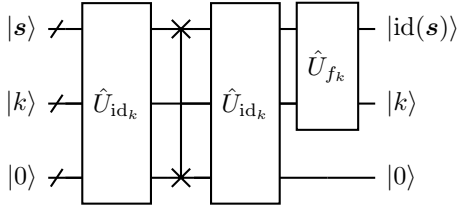
\begin{figure}
       \centering
        \resizebox{6.19975cm}{!}{%
        \begin{quantikz}[column sep = 0.23cm, row sep = 0.3cm]
        \lstick{\ket{\bm{s}}} & \qwbundle{} & \gate[3]{\hat{U}_{\text{id}_k}} & \swap{2} & \gate[3]{\hat{U}_{\text{id}_k}} & \gate[2]{\hat{U}_{f_k}} & \qw\rstick{\ket{\text{id}(\bm{s})}} \\
        \lstick{\ket{k}} & \qwbundle{} & \qw & \qw & \qw & \qw & \qw\rstick{\ket{k}} \\
        \lstick{\ket{0}} & \qwbundle{} & \qw & \targX{} & \qw & \qw & \qw\rstick{\ket{0}}
        \end{quantikz}
        }
    \caption[Circuit overview of the QWOA-CS indexing unitary]{Circuit overview of the QWOA-CS indexing unitary $\hat{U}_{\#CS}$.}
    \label{qc:qwoa-cs-index}
\end{figure}

The QWOA-CS ansatz unitary is summarised as

\begin{equation}
    \hat{U}_\text{QWOA-CS}(\gamma_i, t_{0,i}, t_{1,i}) = \hat{U}_Q(\gamma_i)\hat{U}_{W\text{-CS}}(t_{0,i}, t_{1,i})
\end{equation}
where,
\begin{equation}
    \label{eq:qwoa-cs-mixing-unitary}
    \hat{U}_{W\text{-CS}}(t_{0,i}, t_{1,i}) = \hat{U}_{\mathcal{P}^\prime}(t_{1,i})\hat{U}_{\mathcal{G}}(t_{0,i}),
\end{equation}
encapsulates two sequentially applied submixers, $\hat{U}_{\mathcal{G}}$ and $\hat{U}_{\mathcal{P}^\prime}$. 

The first submixer, $\hat{U}_{\mathcal{G}}$, implements CTQWs on the $K$ subgraphs of $\mathcal{G}$ (see \cref{sec:mixer-design,sec:multiset_partitioning}) that connect solutions in the valid solution space.

The second submixer $\hat{U}_{\mathcal{P}^\prime}$ performs a CTQW over the indexed solution space, 
\begin{equation}
    \hat{U}_{\mathcal{P}^\prime}(t) = \hat{U}_{\#\text{CS}}\hat{U}_{\mathcal{L}}(t) \hat{U}_{\#\text{CS}}^\dagger
\end{equation}
with the CTQW $\hat{U}_{\mathcal{L}}$ defined by the graph Laplacian of a nonhomogenous $K$-partite graph $\mathcal{L}_{|S_0|, |S_1|, \dots,|S_K|}$ with vertices $\mathcal{S}^\prime$ and edges, 

\begin{equation}
E = \{ (\bm{s}, \bm{s}^\prime) \mid \bm{s} \in S_k, \bm{s}^\prime \in S_{k^\prime}, S_k \bigcap S_{k^\prime} = \emptyset \}.
\end{equation}
We address the choice of the graph Laplacian for $\hat{U}_{\mathcal{L}}$ in \cref{sec:k-partite-graph}. The coupling structure of the QWOA-CS mixing unitary for an example portfolio rebalancing problem is shown in \cref{fig:qwoa-cs-coupling}.

\begin{figure*}[t!]
    \centering
    \includegraphics[width =\linewidth]{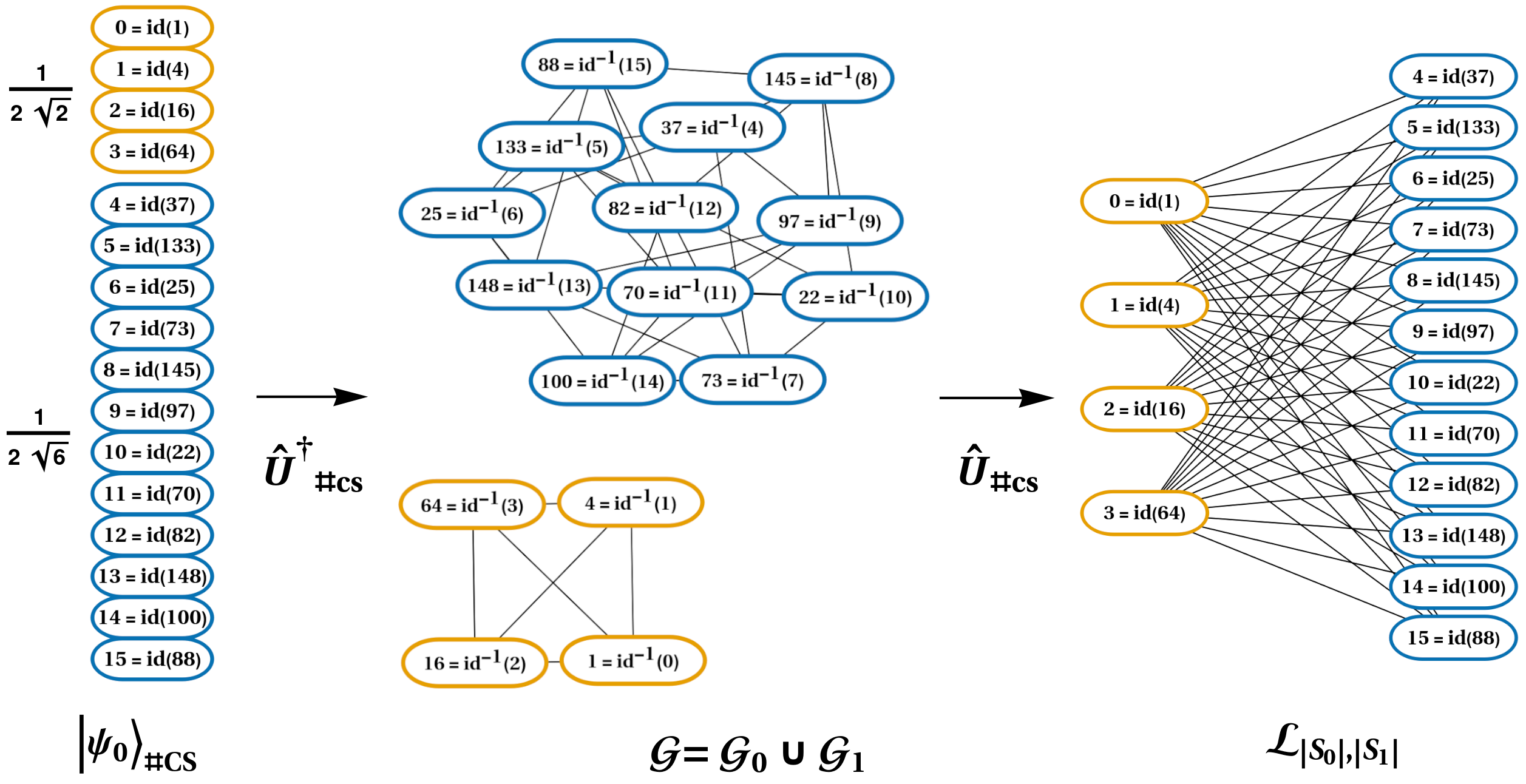}
    \caption[The QWOA-CS initial state and defining graphs]{Illustration of the QWOA-CS weighted initial state $\ket{\psi_0}_{\#CS}$, and the defining graphs of the QWOA-CS mixing unitary (\cref{eq:qwoa-cs-mixing-unitary}), which are a constrained permutation graph $\mathcal{G}$ and $K$-partite graph $\mathcal{L}_{|S_0|, |S_1|}$, for a portfolio rebalancing problem with $n=4$ assets and constraint value $A=-1$ (see \cref{sec:portfolio}) where $x = \text{id}(y)$ labels indexed states and $y = \text{id}^{-1}(x)$ indicates the corresponding unindexed state. Vertices outlined in orange belong to partition $S_0$ and those outlined in blue to $S_1$, which contain solutions generated by permutations of multisets of $\mathcal{X}=(1, -1, 0)$ with multiplicities $\bm{P}_1 = (0, 1, 3)$ and  $\bm{P}_1 = (1, 2, 1)$ respectively (see \cref{sec:multiset_partitioning}). Solutions are encoded with $n \lceil \log m \rceil = 2n$ qubits, so the size of the Hilbert space, $|\mathcal{H}|=256$, is several factors larger than the number of valid solutions, $|\mathcal{S}^\prime|=16$. Furthermore, as the mapping from $\mathcal{S}$ to $\mathcal{H}$ is not one-to-one, $\mathcal{H}$ contains degenerate encodings, with 56 states mapping to one of the 16 valid solutions. The weighted initial state $\ket{\psi_0}_{\#CS}$, defined in \cref{eq:subset_initial_state}, is prepared as an indexed state that populates a canonical subspace of valid non-degenerate solutions. The unindexing unitary $\hat{U}_{\#CS}^\dagger$ maps the populated indexed states to their corresponding permutations in the unindexed space as described in \cref{sec:indexing_perm}. For example, $\text{id}^{-1}(0)$ maps index 0 to the valid solution $\bm{s}_{1}=(-1, 0, 0, 0)$. According to the heuristic for mixing unitary design in \cref{sec:mixer-design}, $\bm{s}_{1}$ is connected in $\mathcal{G}_1$ to $\bm{s}_{4} = (0, -1, 0, 0)$, $\bm{s}_{16}=(0, 0, -1, 0)$, and $\bm{s}_{64} = (0, 0, 0, -1)$, which all have a Hamming distance of two from $\bm{s}_{1}$. The neighbours of any $\bm{s}$ are efficiently identifiable with a time complexity of $\mathcal{O}(n^2)$ in the unindexed space, enabling efficient implementation of $\hat{U}_{\mathcal{G}}$ via sparse Hamiltonian simulation as outlined in \cref{sec:perm_ctqw}. The unindexing unitary $\hat{U}_{\#CS}$ maps a set of non-degenerate valid solutions to their corresponding indexed states, where $U_{\mathcal{L}}$ performs a CTQW over $\mathcal{L}_{|S_0|,|S_1|}$, transferring probability amplitude between solutions in $S_0$ and $S_1$ by the method introduced in \cref{sec:k-partite-graph}.}
    \label{fig:qwoa-cs-coupling}
\end{figure*}

To prevent bias in the convergence of the QWOA-CS due to nonhomogeneity in the size of the valid paritions $S_k$, the system is initialised with a weighted probability distribution that is uniform with respect to the number of valid generating multisets $|\mathcal{P}^\prime|$. As depicted in \cref{qc:qwoa-cs}, this is achieved by applying $\mathcal{F}_{|\mathcal{P}^\prime|}$ to the second register, followed by $\mathcal{F}_{|S_k|}$ and the index offset function $f$ on the first register, to obtain the indexed state,
\begin{equation}
    \ket{\psi_0}_{\#\text{CS}} = \frac{1}{\sqrt{|\mathcal{P}^\prime|}} \sum_{k = 0}^{|\mathcal{P}^\prime|-1}\left( \frac{1}{\sqrt{|S_k|}} \sum_{i=f(k)}^{f(k+1) - 1}\ket{i}\right)\ket{k}.
    \label{eq:subset_initial_state}
\end{equation}
An un-indexed superposition with the desired probability weighting is then given by $\hat{U}_{\#CS}^\dagger\ket{\psi_0}_{\#\text{CS}}$.

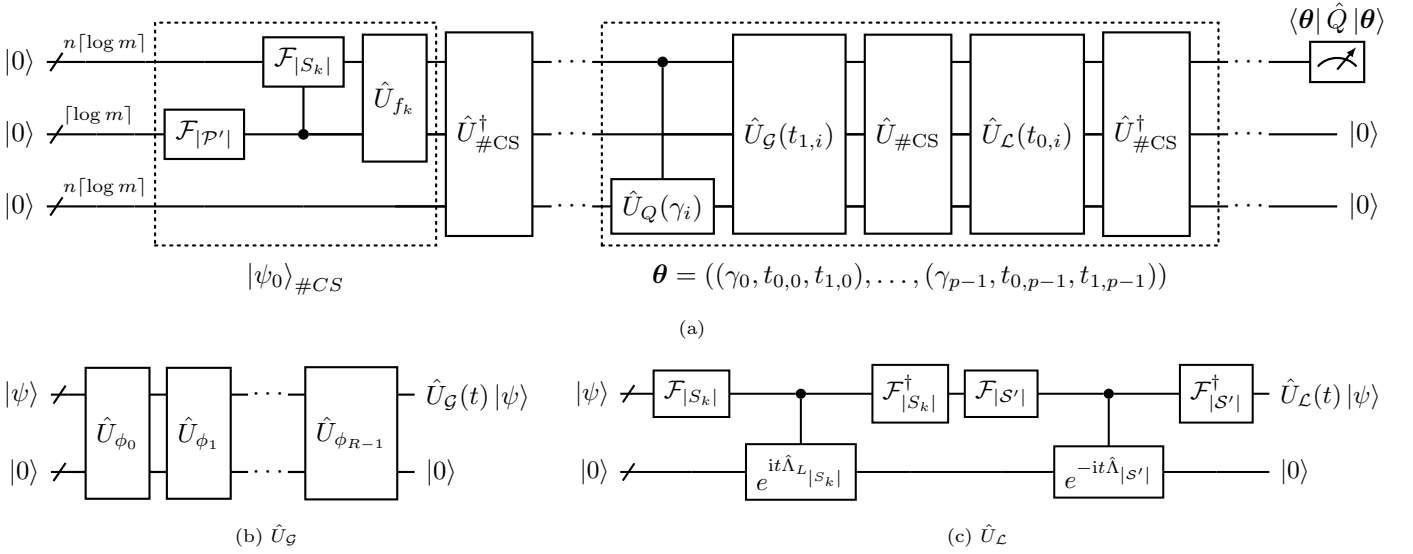
\begin{figure*}
\centering
\begin{subfigure}[c]{\linewidth}
\resizebox{\linewidth}{!}{%
\begin{quantikz}[column sep = 0.25cm, row sep = 0.2cm]
\lstick{\ket{0}} & \qwbundle{n \lceil \log m \rceil} & \phantomgate{HH} & \qw\gategroup[wires=3, steps=3, style={dotted, cap=round, inner sep=0pt, column sep = 0.2cm}, label style={label position=below, yshift=-0.6cm}]{$\ket{\psi_0}_{\#CS}$} & \gate{\mathcal{F}_{|S_k|}} & \gate[2]{\hat{U}_{f_k}} & \gate[3]{\hat{U}_{\#\text{CS}}^\dagger} & \push{\,\cdots\,} & \ctrl{2}\gategroup[wires=3, steps=6, style={dotted, cap=round, inner sep=0pt, column sep = 0.2cm}, label style={label position=below, yshift=-0.6cm}]{$\bm{\theta} = ((\gamma_0, t_{0,0}, t_{1,0}),\dots,(\gamma_{p-1}, t_{0,p-1}, t_{1,p-1}))$} & \gate[3]{\hat{U}_{\mathcal{G}}(t_{1,i})} &  \gate[3]{\hat{U}_{\#\text{CS}}}  &  \gate[3]{\hat{U}_{\mathcal{L}}(t_{0,i})} &  \gate[3]{\hat{U}_{\#\text{CS}}^\dagger} & \qw  & \push{\,\cdots\,} & \meter{$\bra{\bm{\theta}} \hat{Q} \ket{\bm{\theta}}$}  \\
\lstick{\ket{0}} & \qwbundle{\lceil \log m \rceil} & \phantomgate{h} & \gate{\mathcal{F}_{|\mathcal{P}^\prime|}} & \ctrl{-1} & \qw & \qw & \push{\,\cdots\,} & \qw & \qw    & \qw     & \qw  & \qw &  \qw & \push{\,\cdots\,}  & \qw\rstick{\ket{0}}          \\
\lstick{\ket{0}} & \qwbundle{n \lceil \log m \rceil} & \phantomgate{h} & \qw & \qw & \qw & \qw & \push{\,\cdots\,} & \gate{\hat{U}_Q(\gamma_i)} & \qw     & \qw  & \qw    & \qw & \qw & \push{\,\cdots\,}    & \qw\rstick{\ket{0}}          
\end{quantikz}
}
\caption{}
\end{subfigure}

\vspace{0.5em}

\begin{subfigure}[b]{0.39\linewidth}
       \centering
        \resizebox{\linewidth}{!}{%
        \begin{quantikz}[column sep = 0.23cm, row sep = 0.3cm]
        \lstick{\ket{\psi}} & \qwbundle{} & \gate[2]{\hat{U}_{\phi_0}} & \gate[2]{\hat{U}_{\phi_1}} & \push{\,\cdots\,} & \gate[2]{\hat{U}_{\phi_{R-1}}} & \qw\rstick{$\hat{U}_{\mathcal{G}}(t)\ket{\psi}$} \\
        \lstick{\ket{0}} & \qwbundle{} & \qw & \qw & \push{\,\cdots\,} & \qw & \qw\rstick{\ket{0}}
        \end{quantikz}
        }
        \caption{$\hat{U}_{\mathcal{G}}$}
\end{subfigure}
\hfill 
\begin{subfigure}[b]{0.59\linewidth}
       \centering
        \resizebox{\linewidth}{!}{%
        \begin{quantikz}[column sep = 0.23cm, row sep = 0.3cm]
        \lstick{\ket{\psi}} & \qwbundle{} & \gate[1]{\mathcal{F}_{\left|S_k\right|}} & \ctrl{1} & \gate[1]{\mathcal{F}_{
        \left|S_k\right|}^\dagger} & \gate[1]{\mathcal{F}_{\left|\mathcal{S}^\prime\right|}} & \ctrl{1} & \gate[1]{\mathcal{F}_{\left|\mathcal{S}^\prime\right|}^\dagger} & \qw\rstick{$\hat{U}_{\mathcal{L}}(t)\ket{\psi}$} \\
        \lstick{\ket{0}} & \qwbundle{} & \qw & \gate[1]{e^{\text{i} t \hat{\Lambda}_{L_{\left| S_k \right|}}}} & \qw & \qw & \gate[1]{e^{-\text{i} t \hat{\Lambda}_{\left|\mathcal{S}^\prime\right|}}} & \qw & \qw\rstick{\ket{0}}
        \end{quantikz}
        }
        \caption{$\hat{U}_{\mathcal{L}}$}
\end{subfigure}

\caption[Circuit overview of the QWOA-CS and its CTQW sub-circuits]{(a) Circuit overview of the QWOA-CS described in \cref{sec:qwoa-cs}. $\mathcal{F}_N$ denotes a quantum Fourier transform of size $N$, $\hat{U}_{f_k}$ is the index offset unitary, $\hat{U}_{\#\text{CS}}$ is the QWOA-CS indexing unitary (see \cref{qc:qwoa-cs-index}), and $\hat{U}_{\#\text{CS}}^\dagger$ is the unindexing unitary. The initial state $\ket{\psi_0}_{\#CS}$ is prepared as a uniform superposition over $K$ permutation subsets $S = (S_0, S_1, \dots, S_{K-1})$ whose union corresponds to a canonical subspace of valid solutions $\mathcal{S}^\prime$ (see \cref{eq:subset_initial_state}). Phase-encoded solution costs are computed on the third register as described in \cref{qc:qmoa}. (b) $\hat{U}_{\mathcal{G}}$ simulates a CTQW over a constrained permutation graph via sparse Hamiltonian simulation (see \cref{sec:sparse_hamiltonian,sec:perm_ctqw}). (c) $\hat{U}_{\mathcal{L}}$ simulates a CTQW over the $K$-partite graph $\mathcal{L}_{|S_0|, |S_1|, \dots, |S_{K-1}|}$ in the indexed solution space (see \cref{sec:k-partite-graph}).}
\label{qc:qwoa-cs} 
\end{figure*}

\subsection{Properties of Constrained Permutation Graphs}
\label{sec:perm_graph}

The constrained permutation graph, here denoted $\mathcal{G}_k$, is conceptually similar to the complete XY-mixer used by QAOAz (see \cref{sec:qaoaz}) as both graphs are defined by permutation sets in which two permutations are connected if and only if they differ by one transposition of their fundamental elements. However, since $\mathcal{G}_k$ is defined over permutations of a multiset $\text{Perm}(\mathcal{P}_k)$, the two graphs are not equivalent. Therefore, it is necessary to describe the structural properties of the constrained permutation graph to establish concretely our expectations for efficient convergence under \cref{th:min_var} and to determine the complexity associated with the quantum simulation of CTQWs on graphs of this type (see \cref{sec:perm_ctqw}).

Let $I_x: \mathcal{S} \rightarrow I_{x_i} $ be a mapping from $\mathcal{S}$ to a space of index subsets $I_{x_i} = \{ j \mid s_j = x_i \}$ (i.e., subsets containing the indices in $\bm{s}$ that are equal to $x_i$). The union of these sets, $ I \bigcup_{i=0}^{m-1} I_{x_i} $, covers the indexing set $(0, 1, \dots, n-1)$.

\begin{lemma}
Graphs $\mathcal{G}_k$ are regular, have degree,

\begin{equation}\label{eq:degree}
    \text{deg}(\mathcal{G}_k) = \sum_{i > j}^{m-1}|I_{x_i}||I_{x_j}|,
\end{equation}
and are vertex-transitive.
\end{lemma}

\begin{proof}
     As sets $ I_{x_i} $ are disjoint, $ I_{x_i} \times I_{x_j} $, where $i \neq j$, yields $ |I_{x_i}| |I_{x_j}| $ unique index pairs $ (i, j) $, each corresponding to a valid transposition $(i \, j)$ of $s_i, s_j \in \bm{s}$. Noting that $(i \, j) \equiv (j \, i)$ the degree of $\mathcal{G}_k$ is then given by summation over $ |I_{x_i}| |I_{x_j}| $ for $ i > j $. As $|I_{x_i}|$ is invariant under permutation of $\bm{s}$, this result holds for all $\bm{s} \in S_k$, so $\mathcal{G}_k$ is regular. Vertex-transitivity follows from the symmetry in the transpositions $(i \, j)$ that define the edges of $\mathcal{G}_k$.
\end{proof}

\begin{lemma}
The degree of $\mathcal{G}_k$ is upper bound by,
    \begin{equation}
    \label{eq:deg_bound}
        \text{deg}(\mathcal{G}_k) \leq \left\lfloor \frac{n^2}{M} \right\rfloor,
    \end{equation}
where $M = \frac{1}{2}m(m-1)$ is the number of unique unordered pairs of non-equivalent $I_{x_i}$.
\end{lemma}

\begin{proof}
By the inequality of arithmetic and geometric means, we know that the arithmetic mean of a set of non-negative numbers is greater than or equal to the geometric mean of the set~\cite{earl_2021}. With respect to \cref{eq:degree}, the inequality takes the form,

\begin{equation}
    \frac{1}{M} \sum_{i>j}^{m-1} |I_{x_i}||I_{x_j}| \geq \left( \prod_{i>j}^{m-1} |I_{x_i}||I_{x_j}| \right)^{\frac{1}{M}},
\end{equation}
for which equality occurs when all $|I_{x_i}||I_{x_j}|$ are equal. As the $I_{x_i}$ collectively index $\bm{s}$, it is always the case that $n = \sum_{i=0}^{m-1}|I_{x_i}|$, meaning that \cref{eq:degree} is maximised when $|I_{x_i}| = |I_{x_j}| = \frac{n}{M}$.
\end{proof}

\begin{lemma}
The diameter of $\mathcal{G}_k$ scales linearly in $n$ and is determined by $n$ and the maximum of value in $\bm{P}_k$,
\begin{equation}\label{eq:diameter}
    \text{diam}(\mathcal{G}_k)=\min(n - \bm{P}_k).
\end{equation}
\end{lemma}

\begin{proof}
Consider a sequence of $ (i \, j) $ permutations from initial configuration $ \bm{s} \in S_k $. Let $\bar{I}_{x_i}$ be sets containing the indices previously assigned to value $x_i$ in $\bm{s}$ and its permutations $\bm{s}^\prime$. Mapping the permutation sequence to a traversal on $\mathcal{G}_k$, we maximise the inter-node distance between $\bm{s}$ and the currently occupied vertex $\bm{s}^\prime$ by choosing transitions that minimise overlap between $\bm{s}^\prime$ and all of the previously traversed vertices. This corresponds to choosing $ (i \, j) $ such that $\bar{I}_{x_i}$ and $\bar{I}_{x_j}$ increase in size. For this to occur, both sets must contribute a unique index to the other. Hence, the total number of $(i \, j)$ permutations corresponding to an increase in internode distance between $\bm{s}$ and $\bm{s}^\prime$ is $\sum_{i=0}^{m-1}\min{\left( |\bar{I}_{x_i}|, |\bar{I}_{x_j}| \right)}$. \cref{eq:diameter} is obtained as the result of this sum at the starting permutation $\bm{s}$ with $j$ chosen so $|\bar{I}_{x_i}| \leq |\bar{I}_{x_j}|$, noting that $|I_{x_i}| = \bm{P}_{k, i}$ at the start of the permutation sequence.
\end{proof}

\subsection{CTQWs on Constrained Permutation Graphs}
\label{sec:perm_ctqw}

Sparse Hamiltonian simulation offers an efficient method for the implementation of the QWOA-CS submixer $\hat{U}_{\mathcal{G}}$ (see \cref{sec:sparse_hamiltonian,sec:qwoa-cs}). As $\mathcal{G}_k$ are undirected and unweighted, $\text{ELEMENT}$ is a trivial constant-time operation. To compute the $l$-th non-zero column of $\mathcal{G}_k$ in row $\bm{s}$ we define,

\begin{equation}
    \begin{split}
    \text{COLUMN}_k \ket{\bm{s}}\ket{l}\ket{I_{x_0}}\ket{I_{x_1}}\dots\ket{I_{x_{m-1}}}\ket{0} \\
    \rightarrow \ket{\bm{s}}\ket{l}\ket{I_{x_0}}\ket{I_{x_1}}\dots\ket{I_{x_{m-1}}}\ket{\bm{s}_l^\prime}
    \end{split}
\end{equation}
where $\ket{I_{x_j}}$ are of size $\mathcal{O}(\lceil \log n \rceil)$ and hold bit string representations of $I_{x_i}$ (see \cref{sec:perm_graph}) in which a one corresponds to index $j \in I_{x_i}$.

Since we know the size of $I_{x_i}$ and the degree $\text{deg}(\mathcal{G}_k)$ from $\bm{P}_k$ and \cref{eq:degree} respectively, we can compute the $l$-th non-trivial $(i \, j)$ permutation with $\mathcal{O}(n^2)$ time complexity, as shown in \cref{alg:swaps}. The state $\bm{s}^\prime_{l}$ is computed by performing the corresponding swap operation on a register containing a copy of $\bm{s}$. Thus, a CTQW over $\mathcal{G}_k$ can be simulated with
    \[\mathcal{O}\left( {n}^4  t + \frac{\log{(1/\epsilon)}}{\log{\log{(1/\epsilon)}}} \right)\]
gate complexity.

\begin{algorithm2e}
\caption[Generation of neighbours in the constrained permutation graph]{The $i$-th non-trivial permutation of $\bm{s} \in \mathcal{S}^\prime$.}\label{alg:swaps}
    $y \gets 0$\;
    \ForEach{\{$I_{x_k}, I_{x_j} \}$}{
    $y \gets y + |I_{x_k}||I_{x_j}|$\;
    \If{$ i < y$}{
        $l \gets i \, \, \mathrm{mod} \, \, |I_{x_k}|$\;
        $m \gets i \, \, \mathrm{mod} \, \, |I_{x_j}|$\;
        \KwRet{$(I_{{x_j}, \, l} \, \, I_{{x_k}, \, m})$}
    }
    }
\end{algorithm2e}

\subsection{CTQWs on Non-Homogeneous Partite Graphs}
\label{sec:k-partite-graph}

With respect to the indexed solution space $\text{id}(\mathcal{S}^\prime)$, let $\mathcal{L}_{|\mathcal{S}^\prime|}$ be the graph Laplacian of a complete graph and $\mathcal{L}_{|S_k|}$ be the graph Laplacian of the union of $K$ complete graphs connecting solutions within the same $S_k$. The graph Laplacian for a non-homogeneous $K$-partite graph with partite sets $S_k$ can then be expressed as,
\begin{equation}
\mathcal{L}_{|S_0|, |S_1|, \dots, |S_K|} = \mathcal{L}_{|\mathcal{S}^\prime|} - \mathcal{L}_{|S_k|}.
\end{equation}

\begin{lemma}
    For any partitioning of $\text{id}{(\mathcal{S}^\prime)}$ into sequentially indexed disjoint sets, $\left[\mathcal{L}_{|\mathcal{S}^\prime|}, \mathcal{L}_{|S_k|}\right]= 0$.
\end{lemma}
\begin{proof}
\begin{align}
\label{eq:commute}
    \left(\mathcal{L}_{|\mathcal{S}^\prime|} \cdot \mathcal{L}_{|S_k|}\right)_{ij} &= \sum_{k=0}^{|\mathcal{S}^\prime|-1} \left(\mathcal{L}_{|\mathcal{S}^\prime|}\right)_{ik} \left(\mathcal{L}_{|S_k|}\right)_{kj} \nonumber \\
    = \sum_{k=0}^{f(i+1)-1} & \left(\mathcal{L}_{|\mathcal{S}^\prime|}\right)_{ik} \left(\mathcal{L}_{|S_k|}\right)_{kj} \nonumber \\
    = \sum_{k=0}^{f(i+1)-1} & \left(\mathcal{L}_{|S_k|}\right)_{ik} \left(\mathcal{L}_{|\mathcal{S}^\prime|}\right)_{kj}  = \left(\mathcal{L}_{|S_k|} \cdot \mathcal{L}_{|\mathcal{S}^\prime|} \right)_{ij}
\end{align}
where the second line holds as $ \left(\mathcal{L}_{|S_k|}\right)_{ik} $ is non-zero only when $ f(i) \leq k < f(i + 1) $, and the last line holds as $ \mathcal{L}_{|\mathcal{S}^\prime|}$ and $\mathcal{\mathcal{L}}_{|S_k|} $ are symmetric. 
\end{proof}

Equation (\ref{eq:commute}) underpins the requirement that the walk be defined by the graph Laplacian, instead of the adjacency matrix as if $\left(\mathcal{L}_{|S_k|}\right)_{ii} = 0$ the product of  $\mathcal{L}_{|S_k|} \cdot \mathcal{L}_{|\mathcal{S}^\prime|}$ has off-diagonal elements if the size of any of the diagonal blocks in $\mathcal{L}_{|S_k|}$ is not a factor of $|\mathcal{S}^\prime|$.

As $\mathcal{L}_{|\mathcal{S}^\prime|}$ is a circulant matrix, it is diagonalised by the discrete Fourier transform mod $ |\mathcal{S}^\prime| $, $\mathcal{F}_{|\mathcal{S}^\prime|}$, with eigenvalues, 
\begin{equation}
\Lambda_{\mathcal{L}_{|\mathcal{S}^\prime|}} = \left(1, |\mathcal{S}^\prime| + 1,  |\mathcal{S}^\prime| + 1, \ldots, |\mathcal{S}^\prime| +1 \right). 
\end{equation}
 Furthermore, as $ \mathcal{L}_{|S_k|}  $ is block circulant, it is diagonalised by the block diagonal operator $ \mathcal{F}_{|S_k|} $ in which the $k$-th block is the discrete Fourier transform  mod $ |S_k| $, with eigenvalues,
\begin{equation}
\Lambda_{\mathcal{L}_{|S_k|}} = \left( \Lambda_{\mathcal{L}_{|S_0|}}, \Lambda_{\mathcal{L}_{|S_1|}}, \ldots, \Lambda_{\mathcal{L}_{|S_k|}} \right),   
\end{equation}
where $ \Lambda_{|S_k|} = \left(1, |S_0| + 1,  |S_1| + 1, \ldots, |S_{K-1}| + 1 \right)$. 

\begin{theorem}
 A CTQW over the graph Laplacian of an inhomogeneous $K$-partite graph is simulated exactly by,
\begin{equation}\label{eq:evolution}
\begin{split}
    & \hat{U}_{\mathcal{L}_{|S_0|, |S_1|, \dots, |S_K|}}(t) = \\
    & \mathcal{F}_{|\mathcal{S}^\prime|} \exp\left(-\text{i} t \hat{\Lambda}_{\mathcal{L}_{|\mathcal{S}^\prime|}}\right) \mathcal{F}_{|\mathcal{S}^\prime|}^\dagger \mathcal{F}_{|S_k|} \exp\left(\text{i} t \hat{\Lambda}_{\mathcal{L}_{|S_k|}}\right) \mathcal{F}_{|S_k|}^\dagger.
\end{split}
\end{equation}
\end{theorem}

\begin{proof}
    The factorisation, 
    \[\exp\left(-\text{i} t \hat{\mathcal{L}}_{|S_0|, |S_1|, \dots, |S_K|} \right) = \exp\left(-\text{i} t \hat{\mathcal{L}}_{|\mathcal{S}^\prime|})\exp(\text{i} t \hat{\mathcal{L}}_{|S_k|}\right),\]
    follows from the commutativity of $ \mathcal{L}_{|\mathcal{S}^\prime|} $ and $ \mathcal{L}_{|S_k|} $. The circulant and block-circulant structure of $ \mathcal{L}_{|\mathcal{S}^\prime|} $ and $ \mathcal{L}_{|S_k|} $ ensures that their time-evolutions are diagonalised by $\mathcal{F}_{|\mathcal{S}^\prime|}$ and $ \mathcal{F}_{|S_k|} $ respectively.
\end{proof}

The efficiency of this implementation follows from its use of the quantum Fourier transform, due to which $\hat{U}_{\mathcal{P}^\prime}$ has a gate complexity of $ O(n \log n) $, requiring $ O(n) $ ancilla qubits for computation of the eigenvalues. We note that \cref{eq:evolution} also holds for homogeneous $K$-partite graphs; however, in such instances, $\mathcal{F}_{|\mathcal{S}^\prime|}$ is a diagonalizing unitary for both $ \mathcal{L}_{|\mathcal{S}^\prime|} $ and $ \mathcal{L}_{|S_k|} $ so \cref{eq:evolution} may be reduced to $\mathcal{F}_{|\mathcal{S}|}\exp{\left[-\text{i} \left(\hat{\Lambda}_{\mathcal{L}_{|\mathcal{S}^\prime|}} - \hat{\Lambda}_{\mathcal{L}_{|S_k|}}\right)\right]}\mathcal{F}_{|\mathcal{S}|}^\dagger$.  

\section{Methods and Results}
\label{sec:methods_and_results}

 Here we apply the theoretical framework detailed in \cref{sec:theory,sec:figures_of_merit} to the analysis of QVA mixing unitaries, as well as a comparison of the generalised QMOA (\cref{sec:cart_qva}) and QWOA-CS (\cref{sec:perm_qva}) to pre-existing QVAs using the parallel machine scheduling and portfolio optimisation problems as benchmarking applications. For this, we consider QWOA-CS with and without its second submixer $\hat{U}_{\mathcal{L}}$, with the latter case denoted as QWOA-CS~(disjoint). In \cref{sec:CTQW-subshells,sec:convergence}, we present a numerical analysis of the structure and convergence behaviour of the graphs that define the mixing unitaries of the QAOA, QMOA, QWOA, QWOA-CS and QAOA, followed by a detailed analysis of the $(n, m)$ Hamming graphs in \cref{sec:conv_hamming}. In \cref{sec:applications}, we present numerical simulation results for the parallel machine scheduling and portfolio rebalancing problems. Finally, \cref{sec:hybrid-optimisation-results} presents a hybrid optimisation scheme for cases where a closed form solution for the graph subshell coefficients is known.

\begin{table*}[!ht]
\centering
\caption[Comparison of graph characteristics]{Comparison of graph characteristics (see \cref{sec:CTQW-subshells,sec:convergence}).}
\label{table:graph_characteristics_ordered_deg}
\resizebox{\linewidth}{!}{%
\begin{tabular}{@{}lccccccc@{}}
\toprule
Graph  Type & $|V|$ & $\text{deg}(G)$ & $\text{diam}(G)$ & $|\mathcal{N}_{d,k}|$ & $\text{Prob}^*$ & $\min\text{Prob}$ & $\text{HWHM} (\sigma^2)$ \\
\midrule
($7, 2$) Hamming & 128 & 7 & 7 & 8 & 1.00 & 0.038 & 0.21 \\
($3, 5$) Hamming & 125 & 12 & 3 & 4 & 0.91 & 0.037 & 0.23 \\
Constrained Permutation & 168 & 17 & 3 & 10 & 0.84 & 0.031 & 0.061 \\
Parity & 126 & 20 & 4 & 5 & 0.94 & 0.035 & 0.018 \\
Permutation & 128 & 28 & 4 & 5 & 0.62 & 0.040 & 0.021 \\
Complete & 128 & 127 & 1 & 2 & 0.069 & 0.0091 & 0.056 \\
\bottomrule
\end{tabular}
}
\end{table*}

\subsection{Mixing Unitary Graph Analysis}
\label{sec:CTQW-subshells}

We examined the structure of graphs that define the mixing unitaries of the QVAs considered in this work. As the mixing unitaries of the QWOA-CS~(disjoint), the QAOAz~(XY-parity), and the QAOAz~(XY-complete) involve a CTQW over several disjoint subgraphs, one subgraph was selected for each of these algorithms, which are referred to as the constrained permutation, parity, and permutation graphs. The graph size, degree, and diameter were obtained using the \texttt{IGraphM} package~\cite{horvat2023igraphm}, and are listed in \cref{table:graph_characteristics_ordered_deg}. The constrained permutation graph corresponds to a constrained COP with $\bm{P}_k=(1, 5, 2)$ (see \cref{sec:multiset_partitioning}). The corresponding graphs for the QWOA, QMOA, and QAOA are the complete ($7, 2$) Hamming and ($3, 5$) Hamming graphs. Graph subshell coefficients were determined by symbolic computation of \cref{eq:amp_dist} using the Mathematica \texttt{MatrixExp} function. Unique coefficients were identified and verified for occurrence in no more than one shell of the graph~\cite{Mathematica}. For all graphs aside from the constrained permutation graph, the number of subshells $|\mathcal{N}_{d,k}|$ is equal to $\text{diam}(G) + 1$, which is the expected result for distance-transitive graphs.

\subsection{Convergence Potential and Phase Discrepancy}
\label{sec:convergence}

The convergence potential $\text{Prob}^*$ (defined in \cref{eq:convergence_potential}) of the graphs described in \cref{sec:CTQW-subshells} are listed in \cref{table:graph_characteristics_ordered_deg}. These were calculated using optimal walk times $t^*$, which were obtained by numerical optimisation of \cref{eq:opt_t} using the \texttt{NMaximize} function with the \texttt{DifferentialEvolution} method in Mathematica~\cite{Mathematica}. This method was also used for all other numerical optimisation tasks in this and the following subsection.

\begin{figure*}[t!]
    \centering
    \includegraphics[width=\linewidth]{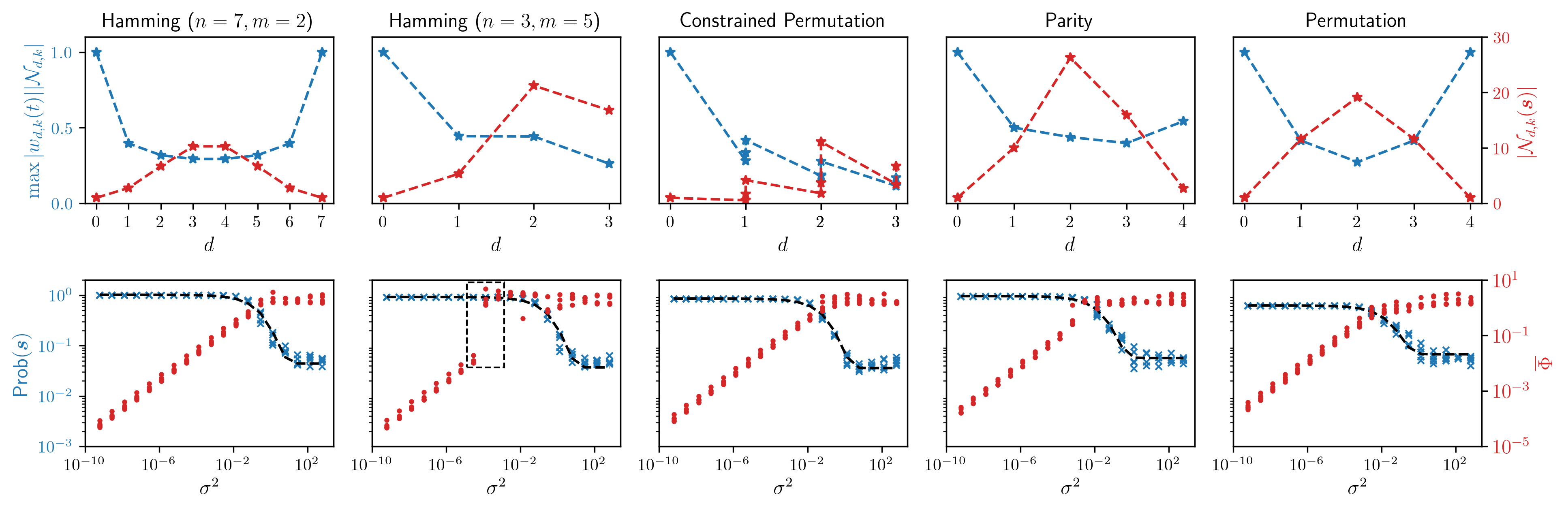}
    \caption[Magnitude of amplitude transfer and graph convergence]{As described in \cref{sec:convergence}, the top row shows the subshell amplitude contributions $|w_{d,k}(t)| |\mathcal{N}_{d,k}|$ to the convergence potential values given in \cref{table:graph_characteristics_ordered_deg} together with the subshell size $|\mathcal{N}_{d,k}|$ against inter-node distance $d$. The bottom row shows convergence to the target state $\text{Prob}(\bm{s})$ and phase discrepancy $\Phi$ with variance-adjusted phases $\bm{q}^*$ as function of the variance $\sigma^2$. Results within the dashed box in the lower $(n=3, m=5)$ Hamming graph plot are also shown in \cref{fig:transfer-and-discrepancy-hamming} (see \cref{sec:CTQW-subshells,sec:convergence}).}
    \label{fig:transfer-and-discrepancy}
\end{figure*}

\begin{figure}[ht!]
    \centering
    \includegraphics[width=6.19975cm]{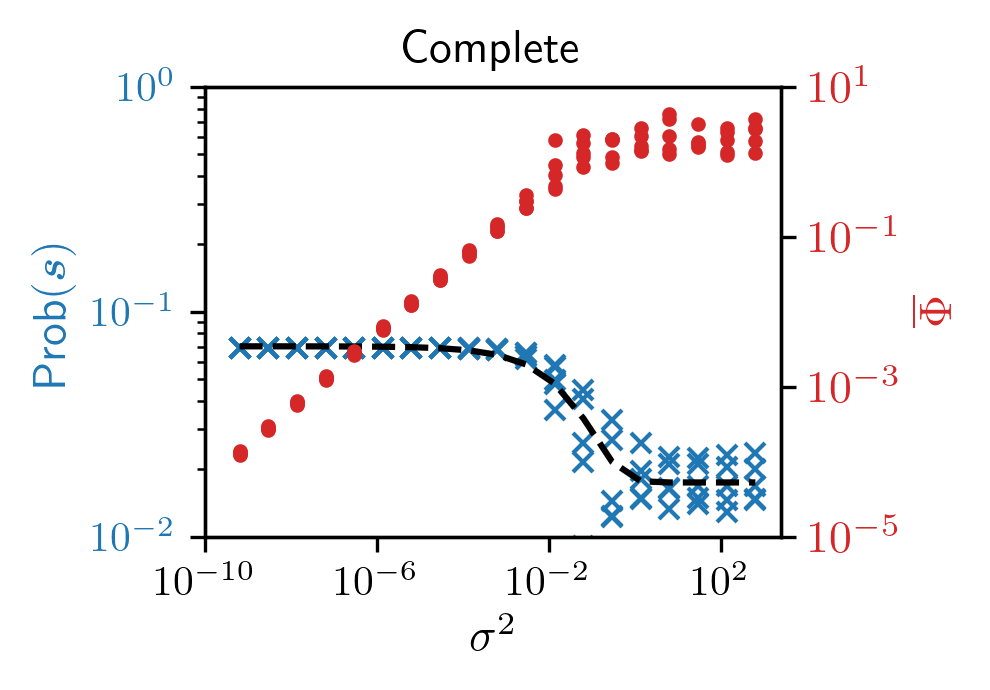}
    \caption[Complete graph convergence]{Convergence to the target state $\text{Prob}(\bm{s})$ and phase discrepancy $\Phi$ with increasing variance $\sigma^2$ in variance adjusted $q^*_{d,k}$ for the complete graph (see \cref{fig:transfer-and-discrepancy}).}
    \label{fig:transfer-and-discrepancy-complete}
\end{figure}

The first row of \cref{fig:transfer-and-discrepancy} shows the magnitude of amplitude contribution from the graph subshells $|w_{d,k}(t)| |\mathcal{N}_{d,k}|$ as a function of inter-node distance, together with the subshell sizes $|\mathcal{N}_{d,k}|$ for the graph test set, excluding the complete graph. The second row of \cref{fig:transfer-and-discrepancy} shows the maximum convergence obtained with the variance-adjusted phases $\bm{q}^*$ following the method described in \cref{sec:convergence_def} (see \cref{eq:variance-adjusted,eq:variance-adjusted-state}). For all graphs, maximum convergence was obtained at $\sigma^2 = 0$.

The convergence results as a function of $\sigma^2$ were fit to the exponential function $a + b\exp(-(\sigma^2 - c))$, from which we obtained an estimation of the variance at which the maximum possible convergence is half of $\text{Prob}^*$. This is reported as the half-width at half maximum ($\text{HWHM}(\sigma^2)$) in \cref{table:graph_characteristics_ordered_deg}. The minimum obtained convergence is listed in the same table under ``$\min\text{Prob}$''.

\begin{figure}[ht!]
    \centering
    \includegraphics[width = 6.19975cm]{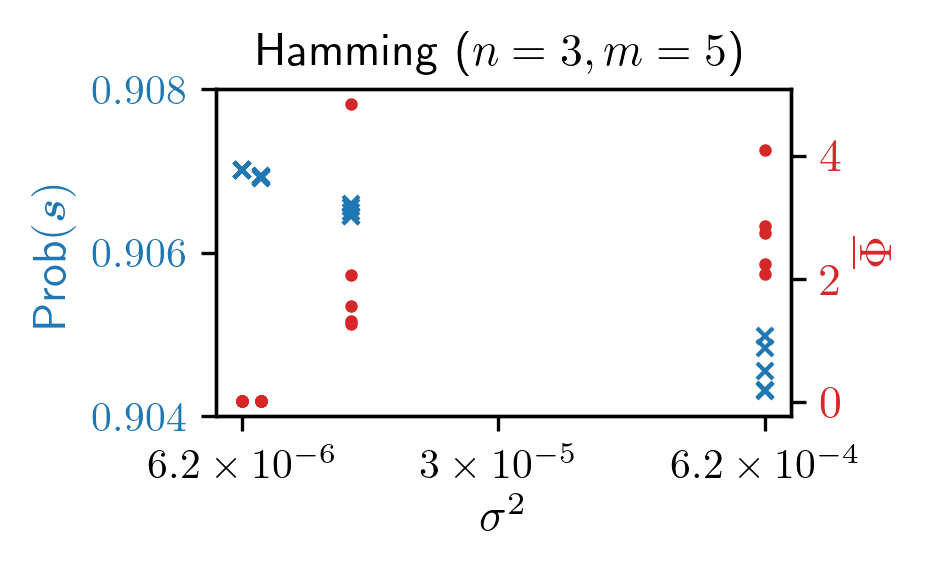}
    \caption[$(3,5)$ Hamming graph convergence]{Convergence to the target state $\text{Prob}(\bm{s})$ and phase discrepancy $\Phi$ with increasing variance $\sigma^2$ in variance adjusted $q^*_{d,k}$ for the $(3, 5)$ Hamming graph from $\sigma^2 = 6.2 \times 10^{-6} $ to $6.2 \times 10^{-4}$ (see \cref{fig:transfer-and-discrepancy}).}
    \label{fig:transfer-and-discrepancy-hamming}
\end{figure}

Convergence in the $(3, 5)$ Hamming graph remains relatively stable over the range $\sigma^2 = 10^{-4}$ to $\sigma^2 = 10^{-2}$, despite a sharp increase in phase discrepancy over this interval. As shown in \cref{fig:transfer-and-discrepancy-hamming}, $\text{Prob}(\bm{s})$ decreases by less than $0.05$ despite the phase discrepancy increasing by an order of magnitude. The convergence of the complete graph, shown in \cref{fig:transfer-and-discrepancy-complete}, is lower on average than that of the other graphs for all values of $\sigma^2$.

\subsection{Convergence Potential of $(n,m)$ Hamming Graphs}
\label{sec:conv_hamming}

Following the procedure of subshell coefficeint computation decribed in \cref{sec:CTQW-subshells}, the subshell coefficeints of the $(n, m)$ Hamming graph, here referred to as $G_{(n,m)}$, were found in closed-form as,
\begin{equation}
\label{eq:hamming-coefficeints}
\begin{aligned}
    w_{d,0}(t) & = \frac{(-1)^d}{m^n}e^{-\text{i} t n (m-1)} \times \\
    & (e^{\text{i} t m}-1)^d (1 + (m-1)e^{\text{i}t m})^{n-d}.
\end{aligned}
\end{equation}
Which, for $G_{(1,m)}$ and $G_{(n,2)}$ reproduces previously derived expressions for the coefficients of the complete and hypercube graphs~\cite{matwiejew_quantum_2023}.

Using \cref{eq:hamming-coefficeints} together with the method for identification of $t^*$ in \cref{sec:convergence}, the convergence potential of $G_{(n,m)}$, denoted $\text{Prob}^*_{(n,m)}$, was numerically solved for $n=2$ to $10$ with $m=2$ to $9$. These results are reported in \cref{fig:hamming-grovers} as:
\begin{equation}
\text{Amplification} =  |\mathcal{S}|\text{Prob}^*_{(n,m)},
\end{equation}
which is defined as the increase in probability amplitude relative to an equal superposition over $\mathcal{S}$. This metric provides an intuitive description of the performance advantage of a particular $G_{(n,m)}$ with increasing problem size. Ideally, amplification should scale according to $|\mathcal{S}^\prime|$.

\begin{figure}[ht!]
    \centering
    \includegraphics[width=8.85679cm]{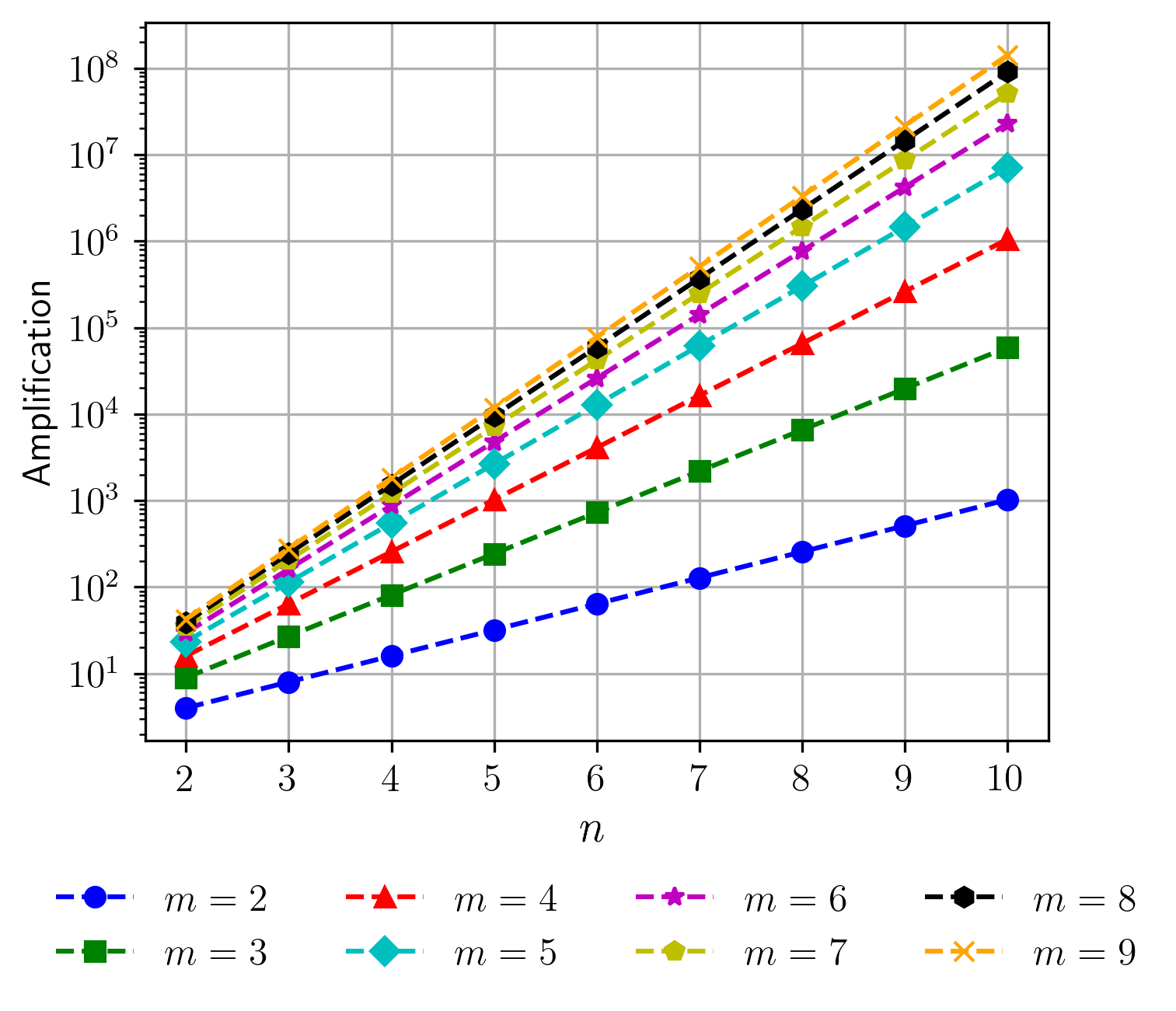}
    \caption[Amplification potential of $(n,m)$ Hamming graphs]{Amplification achieved by maximisation of the graph convergence potential for CTQWs on $(n, m)$ Hamming graphs. Markers indicate $\text{Prob}_{(n,m)}^*$ and the dashed lines $\left(\text{Prob}_{(1,m)}^*\right)^n$ (see \cref{sec:conv_hamming}).}
    \label{fig:hamming-grovers}
\end{figure}

We observed the following relationship between the convergence potential of $G_{(n,m)}$ and the complete graph $G_{(1,m)}$:
\begin{equation}
    \label{eq:hamming_grovers_convergence}
    \text{Prob}^*_{(n,m)} = \left(\text{Prob}^*_{(1,m)}\right)^n,
\end{equation}
which appears to be exact, up to the limits of numerical precision. 
For $m \geq 4$, the amplification of $G_{(n,m)}$ is closely approximated as:
    \[\left( 9 m^{n-1} - 24 m^{n-2} + 16 m^{n-3}\right)^n,\]
where the base is obtained as a solution to $\text{Prob}^*_{(1,m)}$ with $t^*=\pi/m$ and subshell optimal phase-factors $(q^*_{0,0} = 0, q^*_{1,0} = \pi)$, multiplied by $m^n$. By this, we see that the amplification of $G_{(n,m)}$ scales as $\mathcal{O}(e^{n \log m})$. In terms of measurement probability, the convergence potential decreases polynomially with $m$ and exponentially with $n$. However, in the latter case, as $m \rightarrow \infty$, $\text{Prob}^*_{(n,m)}$ is greater than $1/m^n$ by a factor of $9^n$.

Hamming graphs with $m \leq 4$ have $\text{Prob}_{(1,2)}^* = 1$. By \cref{eq:hamming_grovers_convergence}, the convergence potential of the $G_{(n,m \leq 4)}$ graphs is equal to one for all $n$. This entails an increase in amplification that is directly proportional to the size of the solution space. Numerical results for $\text{Prob}^*_{(n,m \leq 4)}$ from $n=1$ to $20$ were consistent with this up to the limits of numerical precision.

\subsection{QVA Benchmarks}
\label{sec:applications}

This section presents a comparative study of the QVAs introduced in \cref{sec:background,sec:cart_qva,sec:perm_qva}, focusing on their mixing structure, convergence behaviour, and ability to amplify the probability of the optimal solution under optimisation of their variational parameters. Algorithms for unconstrained optimisation, including the generalised QMOA and the QAOA, were applied to the parallel machine scheduling problem. The QWOA-CS, QWOA, QAOAz~(XY-parity), and QAOAz~(XY-complete) were applied to the portfolio optimisation problem. Additionally, we include results for the QWOA-CS without the $K$-partite mixing unitary $\hat{U}_{\mathcal{L}}$, denoted as QWOA-CS~(disjoint).

Algorithm performance is primarily evaluated in terms of the approximation ratio, which is defined as the ratio of the average cost of solutions found by sampling from the optimised QVA state to the globally optimal cost function value in the valid solution space. In the context of minimisation, we expressed this as,
\begin{equation}
\text{Approximation Ratio} = \frac{\langle \hat{Q} \rangle - \max C(\mathcal{S}^\prime)}{\min C(\mathcal{S}^\prime) - \max C(\mathcal{S}^\prime)},    
\end{equation}
where, for unconstrained optimisation, $\mathcal{S}^\prime = \mathcal{S}$.

Numerical results were obtained using \texttt{QuOp\_MPI}, which computes a noise-free simulation of the QVA phase-shift and mixing unitaries~\cite{matwiejew_quop_mpi_2022,matwiejew_quop_mpi_2022}. Variational parameters were optimised using the adaptive Nelder-Mead method included with \texttt{SciPy}, with an iteration limit of $10^3$ and a convergence tolerance of $10^{-9}$~\cite{jones_scipy_2001,gao2012implementing}. This optimisation scheme aimed to determine the limiting behaviour of the algorithms given an ``optimal'' choice of variational parameters. For each $p$, five repeats were carried out with $\bm{\theta}$ initialised from a uniform distribution between $0$ and $2\pi$.

\subsubsection{Parallel Machine Scheduling}
\label{sec:pms-results}

Two problem instances were considered with job completion times $\tau$ and machine speed $v$ from uniform distributions over $(25, 100]$ and $(24, 100]$, based on the range of these quantities in~\cite{xiao_branch_2021}. The first problem instance has $n=6$ tasks and $m=5$ machines, and the second problem instance has $n=7$ tasks and $m=4$ machines. We refer to these as Schedule A and Schedule B, respectively. The size of the search space for Schedule A is $|\mathcal{S}^\prime| = 15625$ ($\sim 13.9$ qubits) for the QMOA and $|\mathcal{S}| = 262144$ ($18$ qubits) for the QAOA. The search space for Schedule B is the same for both the QMOA and the QAOA at $|\mathcal{S}^\prime| = 16384$ ($14$ qubits). For QMOA and QAOA, solution costs were scaled to be within $0$ and $2 \pi$ by dividing by the mean of $C(\mathcal{S})$. By \cref{eq:hamming-coefficeints} and the procedure outlined in \cref{sec:convergence} we found the transfer potential for the QMOA graphs to be $0.823$ (Schedule A) and $1$ (Schedule B). The job priority $w$, completion time $\tau$ and machine speed $v$ for the two problem instances are given in Appendix~\ref{app:parallel_machine_data}. 

To apply the QAOA to Schedule A, where the machine count is not a power of two, we expanded the QAOA solution space to include $2^{\lceil \log m \rceil}$ machines and modified the cost function by the addition of a quadratic penalty function,
\begin{equation}
    \label{eq:penalty-qaoa}
    C_{\text{QAOA}}(\bm{s}) = 
    \begin{cases} 
    C(\bm{s}) & \text{if } \max \bm{s} \leq m, \\
    C(\bm{s}) + a (m - \max \bm{s})^2 & \text{otherwise},
\end{cases}
\end{equation}
where $\max \bm{s}$ is the maximum machine number in the solution, and $m$ is the highest valid machine number. Machines over $m$ were assigned speeds equal to the minimum of the valid machines. 

\begin{figure}[t!]
    \centering
    \includegraphics{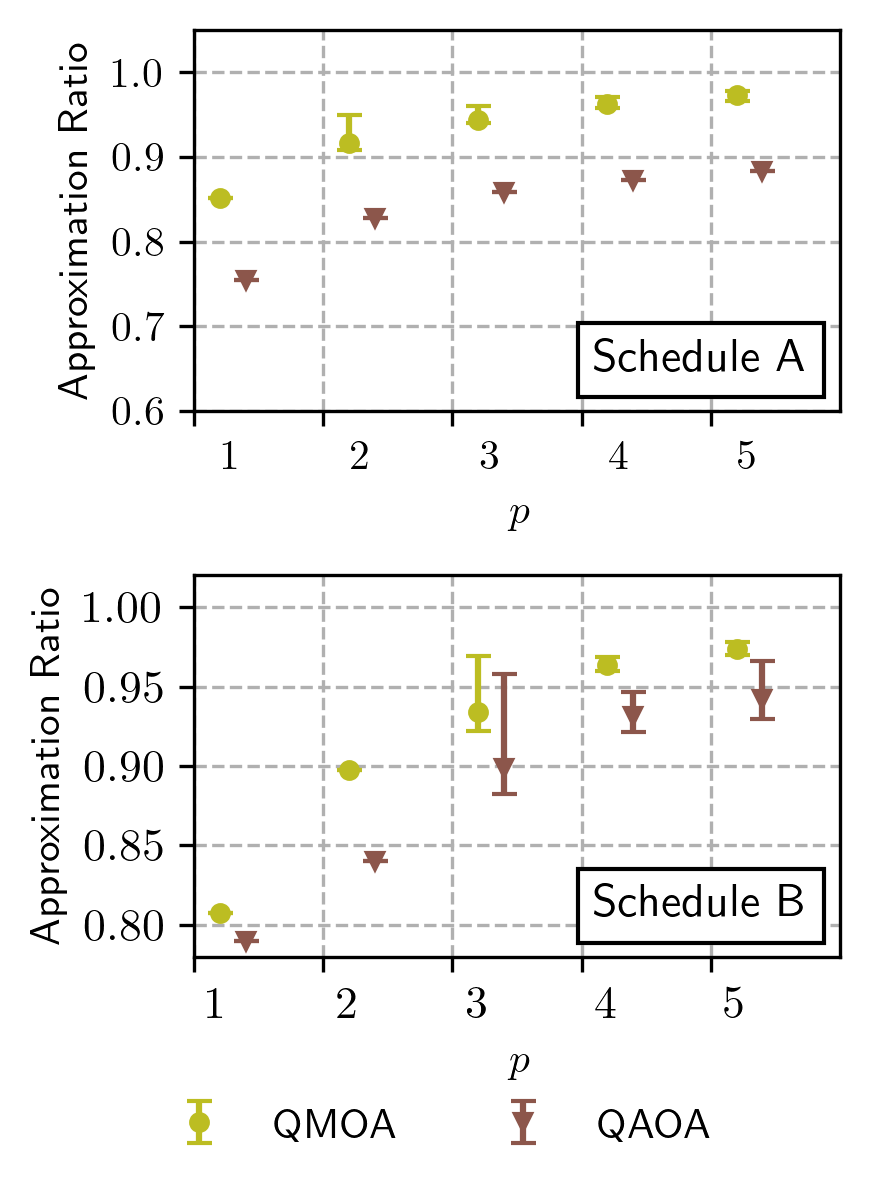}
    \caption[QVA approximation ratio (parallel machine scheduling)]{Approximation ratio with increasing ansatz depth $p$ of the optimised QMOA and QAOA states for the parallel machine scheduling problem.}
    \label{fig:pms-depth-plot}
\end{figure}

The average approximation ratio from $p=1$ to $5$ is shown in \cref{fig:pms-depth-plot}. The QMOA achieves the highest average at each ansatz depth, which, at $p=5$, is $0.973$ for both Schedule A and Schedule B. The corresponding QAOA averages are $0.883$ (Schedule A) and $0.942$ (Schedule B).

\begin{figure}[t!]
    \centering
    \includegraphics[width = 8.85679cm]{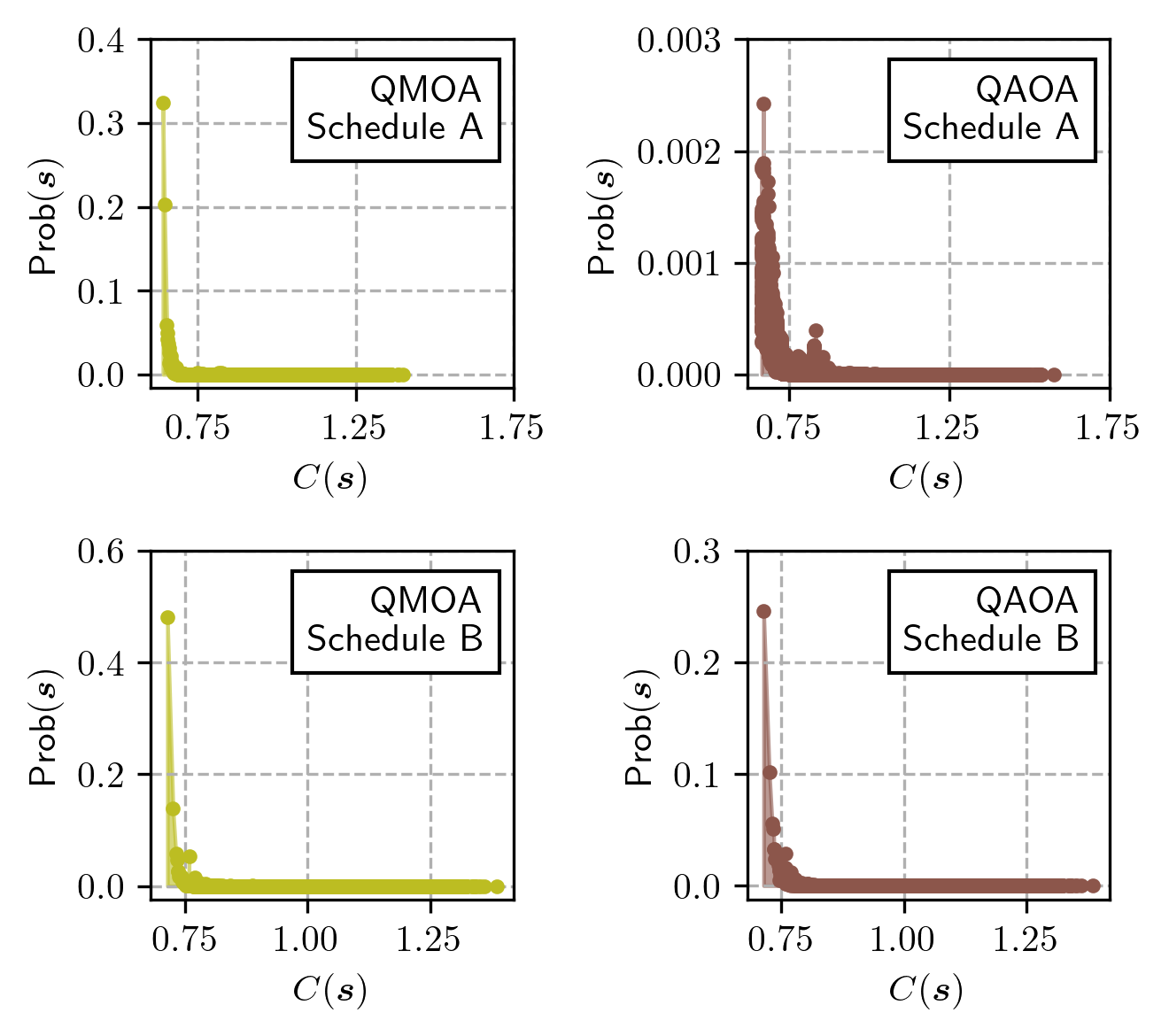}
    \caption[QVA probability distributions (parallel machine scheduling)]{QMOA and QAOA probability distributions for the parallel machine scheduling problem that achieved the highest approximation ratio at $p=5$. The QAOA has a higher upper bound in its solution costs for Schedule A because its search space includes invalid solutions, which are modified by the addition of a quadratic penalty function (see \cref{eq:penalty-qaoa}).}
    \label{fig:pms-states}
\end{figure}

The QMOA and QAOA states for Schedule A and Schedule B with the highest approximation ratio at $p=5$ are shown in \cref{fig:pms-states}. The QMOA states converge to the globally optimal solution with a measurement probability of $0.325$ (Schedule A) and $0.481$ (Schedule B). The highest probability solution in the Schedule A QAOA state has the 91st lowest cost and a probability of $0.00242$. For Schedule B, the QAOA also converges to the globally optimal solution with a measurement probability of $0.246$.

\subsubsection{Portfolio Rebalancing}
\label{sec:portfolio-results}

\begin{figure*}[t!]
    \centering
    \includegraphics[width=\linewidth]{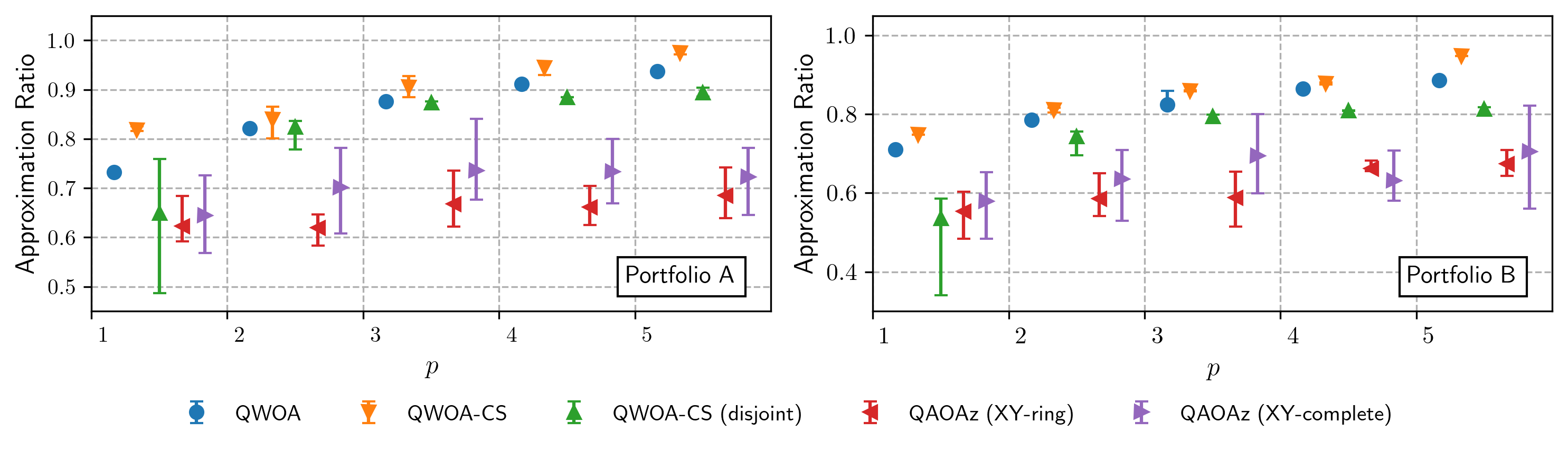}
    \caption[QVA approximation ratio (portfolio rebalancing)]{Approximation ratio of QVA solutions to the portfolio rebalancing problem with increasing ansatz depth $p$ for Portfolio A (left) and Portfolio B (right). Each point represents the mean objective function value after parameter optimisation over five repeats at each ansatz depth, and the error bars depict the minimum and maximum values.}
    \label{fig:portfolio-depth-plots}
\end{figure*}

For simulated optimisation of the portfolio rebalancing problem, solution costs were computed according to \cref{eq:portfolio_cost_function}, based on adjusted closing prices obtained from Yahoo Finance using the \texttt{yfinance} Python package for assets listed on the Australian Securities Exchange~\cite{yahoo_finance_2024}. Two portfolios were considered with $n=6$ and $8$ assets, both with $A=2$. The size of the search space for the QWOA, QWOA-CS and QWOA-CS~(disjoint) is $|\mathcal{S}^\prime| = 90$ ($6.5$ qubits, Portfolio A) and $|\mathcal{S}^\prime| = 784$ ($\sim9.6$ qubits, Portfolios B). The search space for the QAOAz~(XY-parity) and QAOAz~(XY-complete) is $|\mathcal{S}| = 4096$ ($12$ qubits, Portfolio A) and $|\mathcal{S}| = 65536$ ($16$ qubits, Portfolio B). The selected assets, date ranges and model parameters are provided in Appendix~\ref{app:portfolio_data}.

For the QWOA, QWOA-CS, and QWOA-CS~(disjoint), the valid solution space $\mathcal{S}^\prime$ was precomputed. The QWOA and QWOA-CS indexing unitaries do not necessarily produce the same ordering. However, since solution ordering does not influence the dynamics of the QWOA mixing unitary, the same set of solution costs was used for all three algorithms.

An indexing algorithm for the generating multisets $\mathcal{P}_k \in \mathcal{P}^\prime$ of the permuation paritioning of $\mathcal{S}^\prime$ (see \cref{sec:multiset_partitioning}) is given in \cref{alg:index_s} with complexity $\mathcal{O}(n)$. It follows from observing that, to maintain $A$, $\bm{s} \in \mathcal{S}^\prime$ can be modified only by substituting two \emph{no} positions with one long and one short position. Consequently,
\begin{equation}\label{eq:nsubsets}
    \left|\mathcal{P}^\prime\right| = \left\lfloor \frac{1}{2}\left(n - \left| A \right|\right) \right\rfloor + 1.
\end{equation}
A quantum circuit implementing the QWOA-CS indexing unitary for the portfolio rebalancing problem is shown in \cref{qc:index_sj}. The registers $\ket{P_{k,x_0}}$, $\ket{P_{k,x_1}}$, and $\ket{P_{k,x_2}}$ are of size $\mathcal{O}(\log n)$; consequently, the binomial coefficients are computable with a gate complexity of $\mathcal{O}(n\log n)$. By the complexity analysis in \cref{sec:indexing_perm}, the overall gate complexity is $\mathcal{O}(n^4)$.

\begin{algorithm2e}[ht!]
\caption[QWOA-CS multiset indexing (portfolio rebalancing)]{QWOA-CS multiset indexing (portfolio rebalancing).}\label{alg:index_s}
    \uIf{$A \geq 0$}{
        $P_{k,x_0} \gets k$ \tcp*{short positions}
        $P_{k,x_1} \gets A + k$ \tcp*{long positions}
        $P_{k,x_2} \gets n - A + 2k$ \tcp*{no positions}
    }
    \Return $j$\;
\end{algorithm2e}

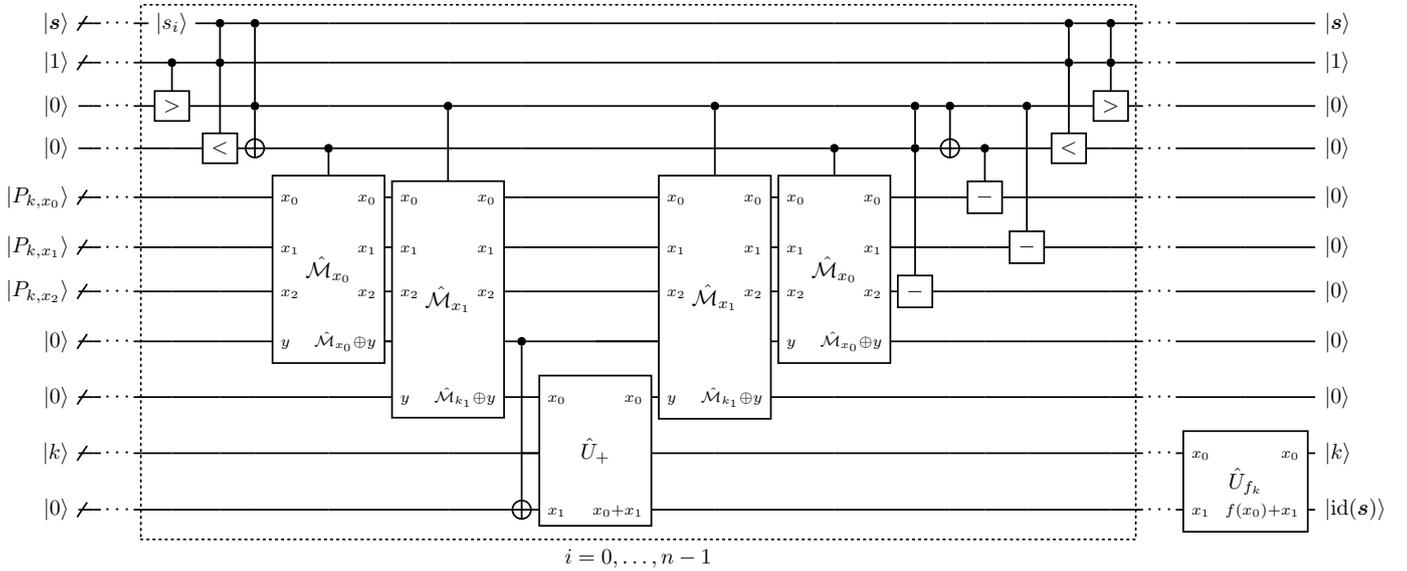
\begin{figure*}[t!]
\centering
\resizebox{\linewidth}{!}{%
\begin{quantikz}[column sep = 0.12cm, row sep = 0.2cm]
    \lstick{\ket{\bm{s}}} & \qwbundle{} & \qw & \push{\,\cdots\,} & \qw & \push{\,\,\ket{s_i}\,\,} \gategroup[wires=11, steps=15, style={dotted, cap=round, inner sep=0.00cm, column sep = 0.010cm, row sep = 1cm}, label style={label position=below, yshift=-0.5cm}]{$i=0,\dots,n-1$} & \ctrl{1} & \ctrl{2} & \qw & \qw & \qw & \qw & \qw & \qw & \qw & \qw & \qw & \qw & \ctrl{2} & \ctrl{1} & \qw & \push{\,\cdots\,} & \qw & \qw\rstick{\ket{\bm{s}}} \\
    \lstick{\ket{1}} & \qwbundle{} & \qw & \push{\,\cdots\,} & \qw & \ctrl{1} & \ctrl{2} & \ghost{H}\qw & \qw & \qw & \qw & \qw & \qw & \qw & \qw & \qw & \qw & \qw & \ctrl{2} & \ctrl{1} & \qw & \push{\,\cdots\,} & \qw & \qw\rstick{\ket{1}} \\
    \lstick{\ket{0}} & \qw & \qw & \push{\,\cdots\,} & \qw & \gate{>} & \qw & \ctrl{1} & \qw & \ctrl{2} & \qw & \qw & \ctrl{2} & \qw & \ctrl{1} & \ctrl{1} & \qw & \ctrl{3} & \qw & \gate{>} & \qw & \push{\,\cdots\,} & \qw & \qw\rstick{\ket{0}} & \\
    \lstick{\ket{0}} & \qw & \qw & \push{\,\cdots\,} & \qw & \qw & \gate{<} & \targ{} & \ctrl{1} & \qw  & \qw & \qw & \qw & \ctrl{1} & \ctrl{3}  & \targ{} & \ctrl{1} & \qw & \gate{<} & \qw & \qw & \push{\,\cdots\,} & \qw & \qw\rstick{\ket{0}} & \\
    \lstick{\ket{P_{k,x_0}}} & \qwbundle{} & \qw & \push{\,\cdots\,} & \qw & \qw & \qw & \qw & \gate[4][1.8cm]{{\hat{\mathcal{M}}}_{x_0}}\gateinput{$x_0$}\gateoutput{$x_0$} & \gate[5,disable auto height][1.8cm]{{\hat{\mathcal{M}}}_{x_1}}\gateinput{$x_0$}\gateoutput{$x_0$} & \qw & \qw & \gate[5][1.8cm]{{\hat{\mathcal{M}}}_{x_1}}\gateinput{$x_0$}\gateoutput{$x_0$} & \gate[4][1.8cm]{{\hat{\mathcal{M}}}_{x_0}}\gateinput{$x_0$}\gateoutput{$x_0$} & \qw & \qw & \gate{-} & \qw & \qw & \qw & \qw & \push{\,\cdots\,} & \qw & \qw\rstick{\ket{0}} & \\ 
    \lstick{\ket{P_{k,x_1}}} & \qwbundle{} & \qw & \push{\,\cdots\,} & \qw & \qw & \qw & \qw & \gateinput{$x_1$}\gateoutput{$x_1$} & \gateinput{$x_1$}\gateoutput{$x_1$} & \qw & \qw & \gateinput{$x_1$}\gateoutput{$x_1$} & \gateinput{$x_1$}\gateoutput{$x_1$} &  \qw &  \qw & \qw & \gate{-} & \qw & \qw & \qw & \push{\,\cdots\,} & \qw & \qw\rstick{\ket{0}} & \\
    \lstick{\ket{P_{k,x_2}}} & \qwbundle{} & \qw & \push{\,\cdots\,} & \qw & \qw & \qw & \qw & \gateinput{$x_2$}\gateoutput{$x_2$} & \gateinput{$x_2$}\gateoutput{$x_2$} & \qw & \qw & \gateinput{$x_2$}\gateoutput{$x_2$} & \gateinput{$x_2$}\gateoutput{$x_2$} &  \gate{-} &  \qw &  \qw & \qw & \qw & \qw & \qw & \push{\,\cdots\,} & \qw & \qw\rstick{\ket{0}} & \\
    \lstick{\ket{0}} & \qwbundle{} & \qw & \push{\,\cdots\,} & \qw & \qw & \qw & \qw & \gateinput{$y$}\gateoutput{${\hat{\mathcal{M}}}_{x_0} {\oplus} y$} & \qw & \ctrl{3} & \qw & \qw & \gateinput{$y$}\gateoutput{${\hat{\mathcal{M}}}_{x_0} {\oplus} y$} & \qw & \qw & \qw & \qw & \qw & \qw & \qw & \push{\,\cdots\,} & \qw & \qw\rstick{\ket{0}} & \\
    \lstick{\ket{0}} & \qwbundle{} & \qw & \push{\,\cdots\,} & \qw & \qw & \qw & \qw & \qw & \ghost{{\hat{\mathcal{M}}}_\text{1}}\gateinput{$y$}\gateoutput{${\hat{\mathcal{M}}}_{k_1} {\oplus} y$} & \qw & \gate[3][1.8cm]{\hat{U}_+}\gateinput{$x_0$}\gateoutput{$x_0$} & \gateinput{$y$}\gateoutput{${\hat{\mathcal{M}}}_{k_1} {\oplus} y$} & \qw & \qw & \qw & \qw & \qw & \qw &  \qw & \qw & \push{\,\cdots\,} & \qw &\qw\rstick{\ket{0}} & \\
    \lstick{\ket{k}} & \qwbundle{} & \qw & \push{\,\cdots\,} & \qw & \qw & \qw & \qw & \qw & \qw & \qw & \qw & \qw & \qw & \qw & \qw & \qw  & \qw & \qw & \qw & \qw & \push{\,\cdots\,} & \gate[2][2cm]{\hat{U}_{f_k}}\gateinput{$x_0$}\gateoutput{$x_0$} & \qw\rstick{\ket{k}} & \\
    \lstick{\ket{0}} & \qwbundle{} & \qw & \push{\,\cdots\,} & \qw & \qw & \qw & \qw & \qw & \qw & \targ{} & \ghost{H}\gateinput{$x_1$}\gateoutput{$x_0 {+} x_1$} & \qw & \qw & \qw & \qw & \qw  & \qw & \qw & \qw & \qw & \push{\,\cdots\,} & \gateinput{$x_1$}\gateoutput{$f(x_0) {+} x_1$} &  \qw\rstick{\ket{\text{id}(\bm{s})}} &
\end{quantikz}
}
    \caption[QWOA-CS indexing circuit (portfolio rebalancing)]{Circuit implementing the QWOA-CS indexing algorithm (\cref{alg:index_sj})  for the portfolio rebalancing problem (see also \cref{eq:index_1,eq:index_2}). $\mathcal{M}_{x_j}$ computes the number of permutations remaining after the choice of $x_j$, e.g., $\hat{\mathcal{M}}_{x_0}\ket{x_0, x_1, x_2}\ket{y}=\ket{x_0,x_1,x_2}\ket{\binom{x_0 + x_1 + x_2 - 1}{x_0 -1, x_1, x_2} \oplus y}$, based on the circuit presented in \cref{qc:multinomial}. The CNOT operation between the fourth lowest and lowest registers writes the output of $\hat{\mathcal{M}}_{x_0}$ to the lowest register with $\mathcal{O}(n)$ CNOT gates.}
    \label{qc:index_sj}
\end{figure*}

As shown in \cref{fig:portfolio-depth-plots}, the approximation ratio for all of the QVAs improves with increasing $p$. For Portfolio A at $p=5$, the QWOA, QWOA-CS, and QWOA-CS~(disjoint) have mean approximation ratios of $0.937$, $0.975$, and $0.893$, respectively, which are noticeably higher than the means of the QAOAz~(XY-parity) and QAOAz~(XY-complete), which are $0.685$ and $0.722$, respectively. A similar trend is observed for Portfolio B at $p=5$, with the QWOA, QWOA-CS, and QWOA-CS~(disjoint) having mean approximation ratios of $0.887$, $0.948$, and $0.814$, compared to the QAOAz~(XY-parity) and QAOAz~(XY-complete) means of $0.674$ and $0.705$. At $p=2$ and higher, the QWOA has the smallest variance in its approximation ratio, followed by the QWOA-CS, QWOA-CS~(disjoint), QAOAz~(XY-parity), and QAOAz~(XY-complete). This is also the case at $p=1$, with the exception being the QWOA-CS~(disjoint), which has a larger range than the rest of the QVAs.

\begin{figure*}
    \centering
    \includegraphics[width=\linewidth]{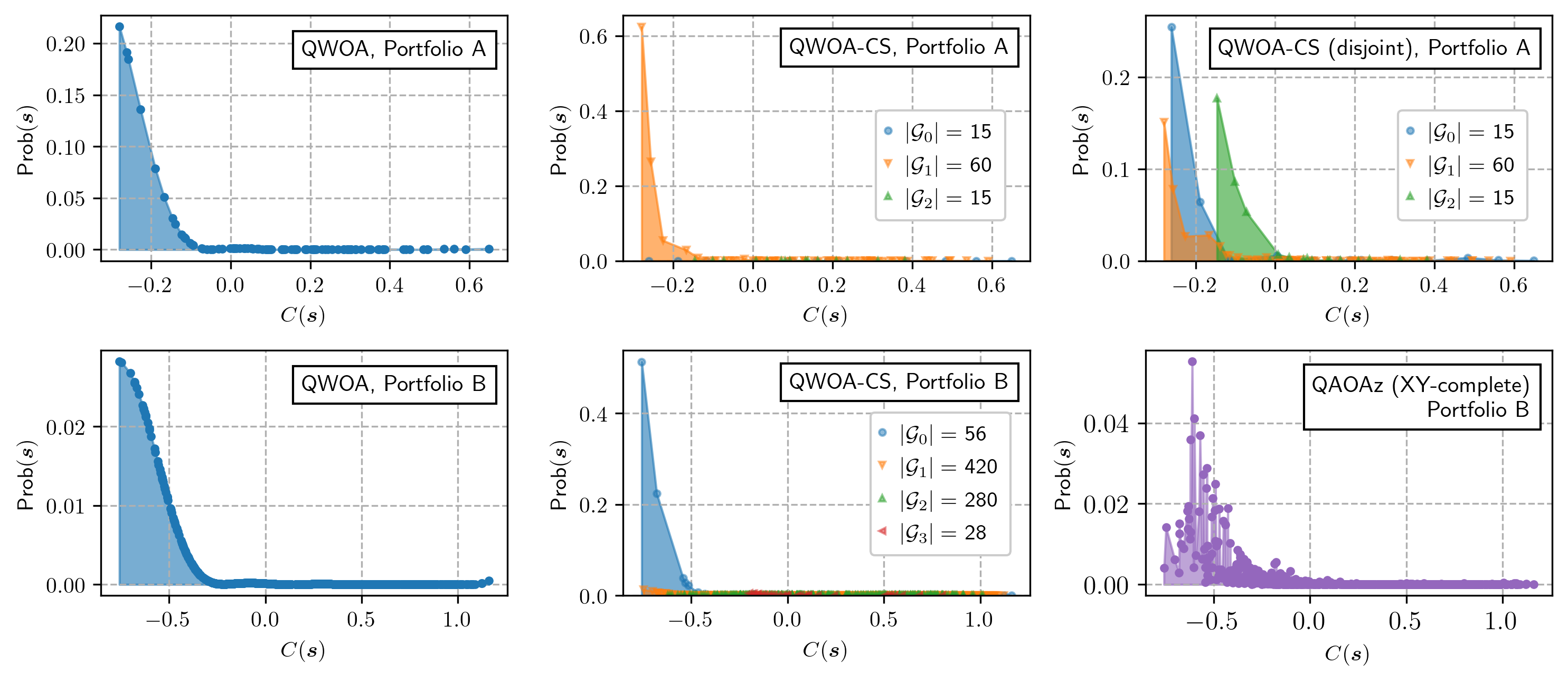}
    \caption[Optimised QVA probability distributions (portfolio rebalancing)]{Selected QVA solution probability at $p=5$. Each state corresponds to the upper bound in approximation depicted in \cref{fig:portfolio-depth-plots}. The QWOA-CS and QWOA-CS~(disjoint) subplots depict the probability distribution across the disjoint subgraphs of the constrained permutation graph. The QAOAz~(XY-complete) also evolves over disjoint subgraphs (the permutation graph); however, they are not shown due to significant overlap in probability distribution over these graphs.}
    \label{fig:portfolio-states}
\end{figure*}

While the QWOA, QWOA-CS, and QWOA-CS~(disjoint) have comparable approximation ratios, their convergence behaviour is distinct. Examining the probability distribution for the Portfolio A states that achieved the highest approximation ratio at $p=5$, as shown in \cref{fig:portfolio-states}, we see that the QWOA-CS has converged to the globally optimal solution with a measurement probability of $0.624$ and that its probability density is concentrated in the subgraph containing this solution. The QWOA-CS~(disjoint) shows convergence to the optimal solution in each subgraph, with the second-best solution having the highest probability of $0.248$. Like the QWOA-CS, the QWOA has also converged to the globally optimal solution; however, it has a lower probability of $0.216$.

Moving to Portfolio B, we see that the difference is more pronounced. Measurement of the globally optimal solution from the QWOA state, shown in \cref{fig:portfolio-depth-plots} at $p=5$, has a probability of $0.0282$ compared to a probability of $0.513$ for the QWOA-CS, which again shows an almost complete transfer of probability density to the corresponding subgraph.

\begin{figure}[ht!]
    \centering
    \includegraphics[width = 6.9083cm]{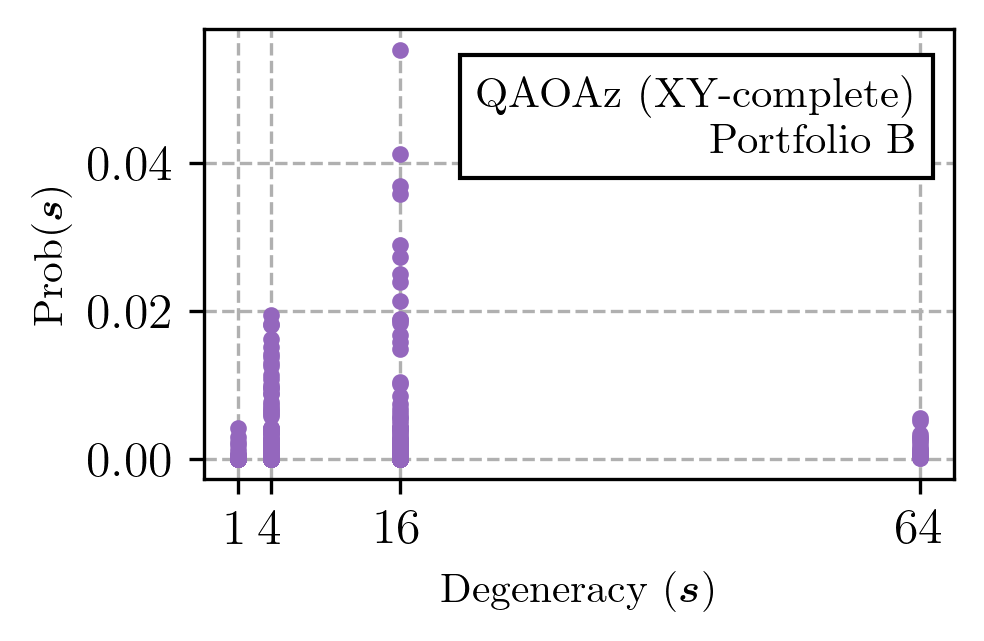}
    \caption[QAOAz~(XY-complete) probability and solution degeneracy]{Portfolio B solution probability and degeneracy for the QAOAz~(XY-complete) state that achieved the highest approximation ratio at $p=5$ (see also \cref{fig:portfolio-depth-plots,fig:portfolio-states}). }
    \label{fig:qaoaz_degen}
\end{figure}

The QAOAz~(XY-complete) has an upper bound in its approximation ratio comparable to the QWOA-CS~(disjoint) at $p=3$ and $p=5$. The corresponding QAOAz~(XY-complete) probability distribution is shown in the lower left of \cref{fig:portfolio-states}. It shows strong convergence to solutions away from the global minimum, with a maximum probability of $0.007$ at the state that corresponds to the 27th lowest cost. \cref{fig:qaoaz_degen} shows solution probability against degeneracy in the quantum search space, where two states $\ket{\bm{s}}$ and $\ket{\bm{s}^\prime}$ are considered degenerate if they encode the same portfolio configuration. Convergence appears biased towards solutions with a degeneracy of 16, with the next highest group being solutions with a degeneracy of four, followed by those with a degeneracy of one.

\subsubsection{Shell Variance}
\label{sec:shell-variance-results}

\begin{figure}[ht!]
    \centering
    \includegraphics[width=8.85679cm]{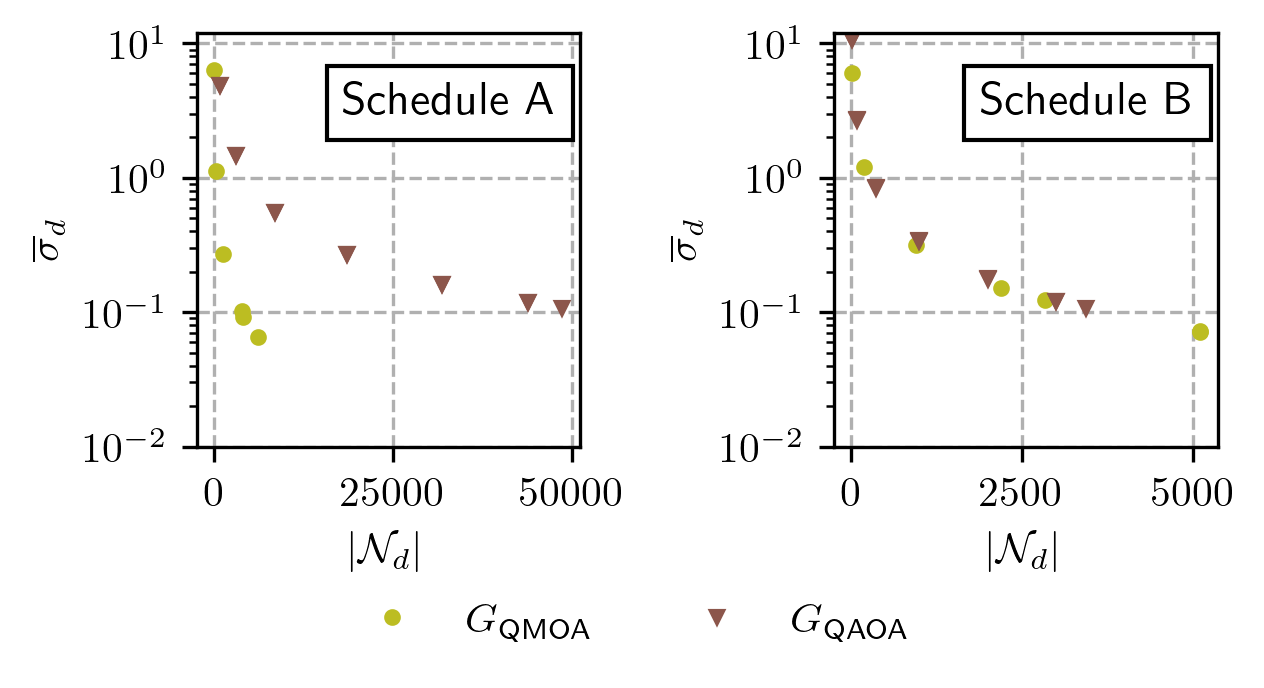}
    \caption[Shell variance of cost embeddings (parallel machine scheduling)]{QMOA and QAOA MSV of the Schedule A and Schedule B parallel machine scheduling problems where $G_\text{QMOA}$ and $\hat{G}_\text{QAOA}$ refer to the defining graphs of the mixing unitaries of the QMOA and QAOA respectively.}
    \label{fig:PMS-MSV}
\end{figure}

The shell variance, defined in \cref{eq:MSV} (\cref{sec:MSV}), of the QMOA and QAOA graph embeddings of the parallel machine scheduling problem solution space is shown in \cref{fig:PMS-MSV}. The QMOA maps the solution costs of Schedule A and Schedule B into the $(6, 5)$ and $(7, 4)$ Hamming graphs, respectively. The QAOA maps solution costs to the $(n, 2)$ Hamming graph, with $n=18$ for Schedule A and $n=14$ for Schedule B. The QMOA graph MSV is $1.33$ (Schedule A) and $1.14$ (Schedule B). The QAOA MSV is higher at $16.5$ (Schedule A) and $2.31$ (Schedule B). \cref{fig:portfolio-shell-var} depicts the shell variance for the QVAs applied to the optimisation of Portfolios A and B (excluding the QWOA). The constrained permutation graph has the lowest MSV, followed by the parity and permutation graphs. The MSV for the complete graph of the QWOA mixing unitary is $3.96$ (Portfolio A) and $166$ (Portfolio B). In each case, the MSV is inversely correlated with the mean approximation ratio and the maximum observed convergence to the global minimum (see \cref{fig:portfolio-depth-plots,fig:portfolio-states}). 

\begin{figure*}
    \centering
    \includegraphics[width = \linewidth]{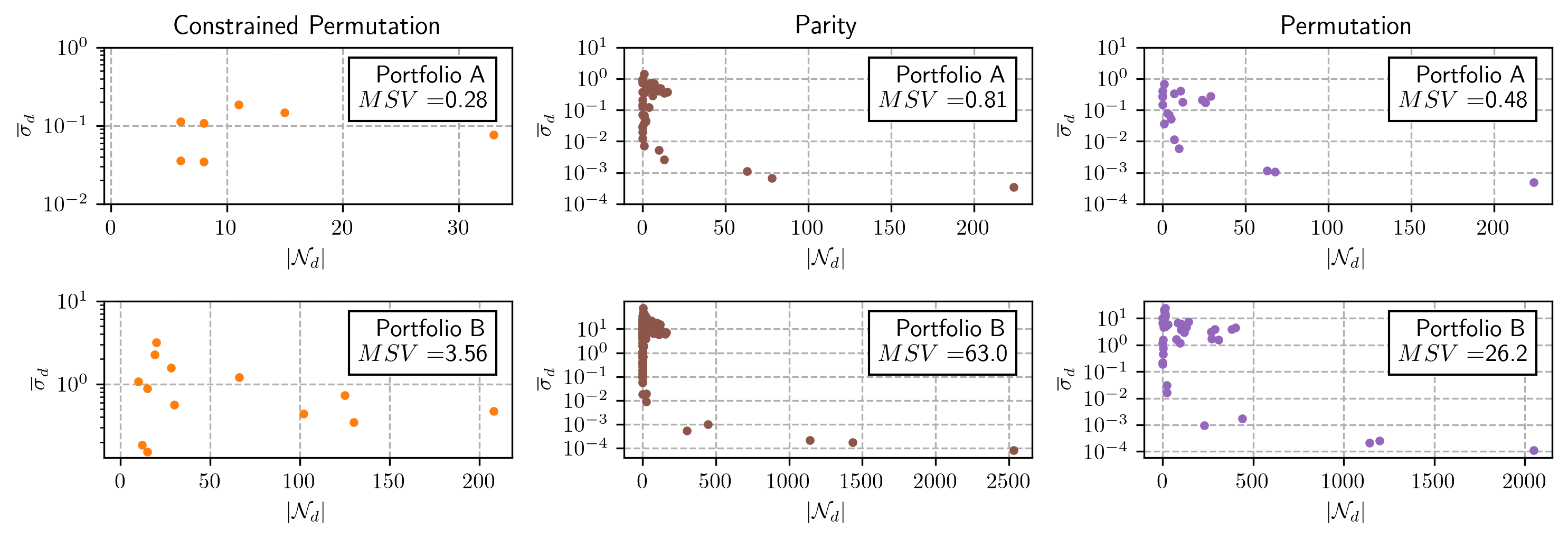}
    \caption[Shell variance of solution cost embeddings (portfolio rebalancing)]{Shell variance and MSV of the embedding of the solution costs of the portfolio rebalancing problem into defining graphs of the QWOA-CS~(disjoint), QAOAz~(XY-parity) and QAOAz~(XY-complete) mixing unitaries, as defined in \cref{sec:figures_of_merit}. Multiple points at the same $|\mathcal{N}_d|$ indicate shells of the same size in disjoint subgraphs.}
    \label{fig:portfolio-shell-var}
\end{figure*}

\subsubsection{Phase Distribution in QVA states}
\label{sec:phase-distribution-results}

By \cref{th:min_var}, a QVA that exhibits efficient convergence is expected to produce a phase distribution that clearly distinguishes the state with the highest amplification. We interpret this clustering as indicative of a QVA's ability to support phase coherence across the subshells of its mixing unitary over repeated ansatz iterations. A wider range of angles over states with significant amplification is expected to lead to a more diffusive evolution compared to phase distributions that narrowly group states with a similar phase-encoded cost.

\cref{fig:pms-phase,fig:portfolio-phase} show the probability ratio of QVA states for the parallel machine scheduling and portfolio optimisation problems against the phase angle relative to the state with the highest probability of measurement, denoted as $\bm{s}^*$. A depth of $p=3$ was selected to assess these properties at a point of intermediate convergence across all the considered algorithms. In each case, a narrower range of phase differences at $p=3$ is associated with higher convergence at $p=5$ and a smaller MSV (see also \cref{fig:PMS-MSV,fig:portfolio-shell-var,fig:pms-states,fig:portfolio-states}).

\begin{figure}[t!]
    \centering
    \includegraphics[width=8.85679cm]{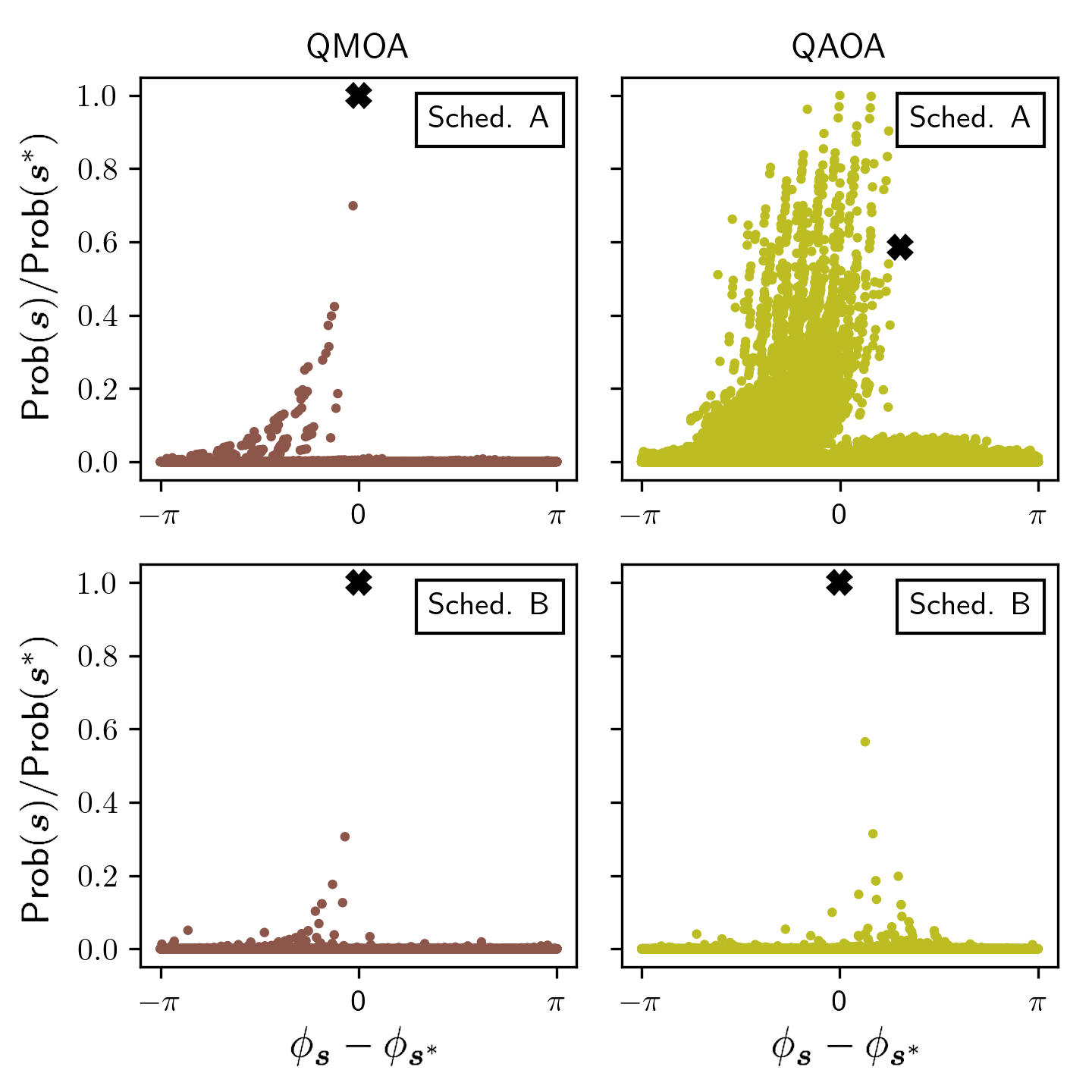}
    \caption[Probability ratio and phase angle (parallel machine scheduling)]{Probability ratio versus phase angle $\phi$ for states $\bm{s}$, relative to the state with the highest measurement probability in the optimised state vectors of QVAs solving the parallel machine scheduling problem. Each distribution corresponds to the highest approximation ratio depicted in \cref{fig:portfolio-depth-plots} at $p=3$. The black cross indicates the globally optimal solution $\bm{s}_{\text{opt}}$, which is unique in each instance. The upper row corresponds to Schedule A (Sched. A) and the lower to Schedule B (Sched. B). The measurement probability of $\bm{s}^*$ for each QVA are as follows. Schedule A: $0.0606$ (QMOA), $0.000559$ (QAOA). Schedule B $0.206$ (QMOA), $0.103$ (QAOA).}
    \label{fig:pms-phase}
\end{figure}
\begin{figure*}
    \centering
    \includegraphics[width = \linewidth]{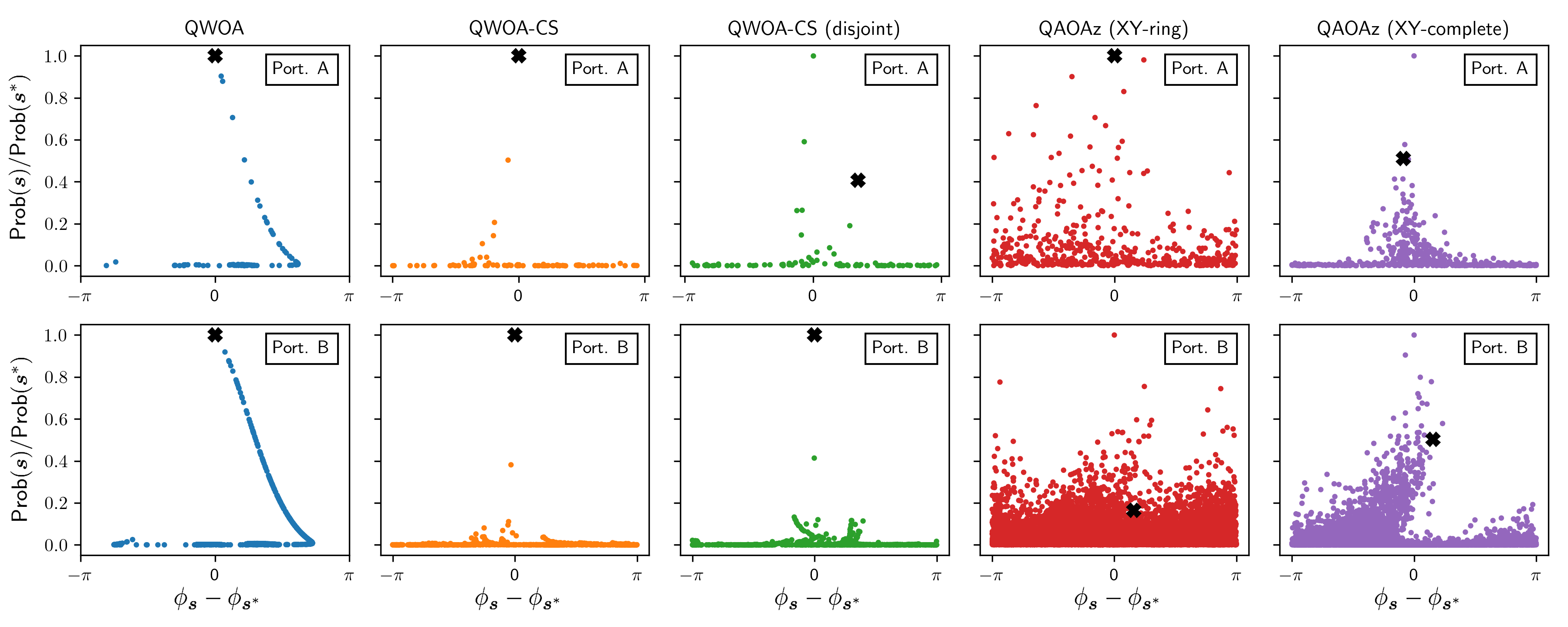}
    \caption[QVA probability ratio and phase angle (portfolio rebalancing)]{Probability ratio versus phase angle $\phi$ for states $\bm{s}$, relative to the state with the highest measurement probability in the optimised state vectors of QVAs solving the portfolio reballancing problem. Each distribution corresponds to the highest approximation ratio depicted in \cref{fig:portfolio-depth-plots} at $p=3$. The black cross indicates the globally optimal solution $\bm{s}_{\text{opt}}$, which is unique in each instance. The upper row corresponds to Portfolio A (Port. A) and the lower to Portfolio B (Port. B). The measurement probability of $\bm{s}^*$ for each QVA are as follows. Portfolio A: $0.118$ (QWOA), $0.369$ (QWOA-CS), $0.229$ (QWOA-CS~(disjoint)), $0.0151$ (QAOAz~(XY-parity)), $0.0269$ (QAOAz~(XY-complete)). Portfolio B: $0.0147$ (QWOA), $0.213$ (QWOA-CS), $0.103$ (QWOA-CS~(disjoint)), $0.00144$ (QAOAz~(XY-parity)), $0.00325$ (QAOAz~(XY-complete)).}
    \label{fig:portfolio-phase}
\end{figure*}

\subsubsection{Hybrid Parameter Optimisation}
\label{sec:hybrid-optimisation-results}

\begin{figure}[!h]
    \centering
    \includegraphics[width=8.85679cm]{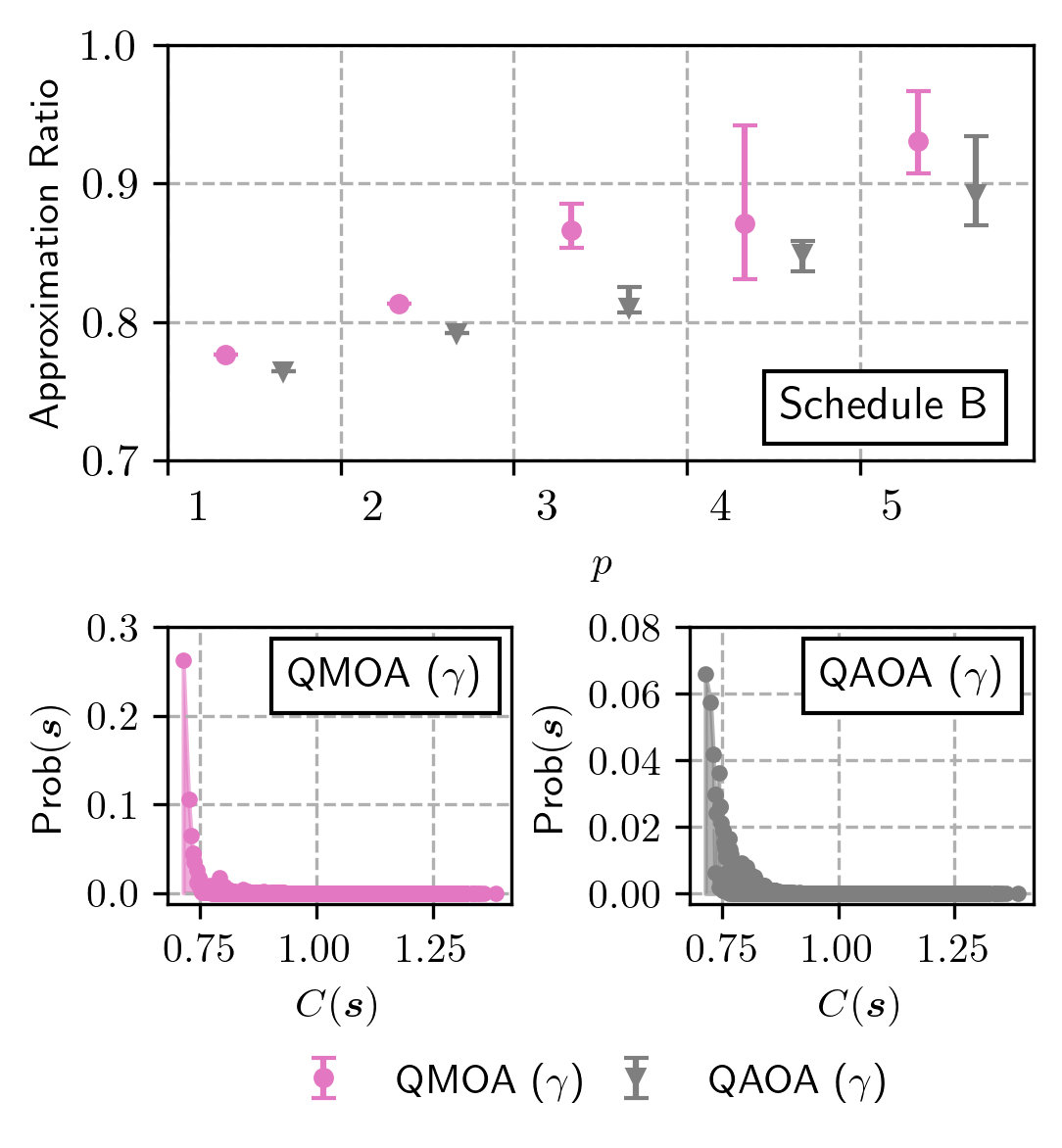}
    \caption[Optimisation with fixed walk times (parallel machine scheduling)]{Optimiastion of parallel machine scheduling problem Schedule B with fixed subshell-optimal walk times. The top shows the mean approximation ratio with increasing ansatz depth $p$. Each point is the mean of five repeats, and the bars show the minimum and maximum observed at the given depth. The bottom row shows the states with the highest approximation ratio at $p=5$. From lowest to highest $d$, the QMOA~($\gamma$) subshell-optimal walk times were: 1.79661, 1.90339, 1.14226, 2.10132, 2.46619, 0.785398, 0.785398. For QAOA~($\gamma$), these were: 0.27055, 0.387597, 0.481275, 0.563943, 0.640522, 0.713724, 0.785398, 0.857072, 0.930274, 1.00685, 2.05207, 1.95839, 1.30025, 1.5708.}
    \label{fig:pms-hybrid-opt}
\end{figure}

In \cref{sec:CTQW-subshells}, \cref{eq:hamming-coefficeints} opens the possibility of developing hybrid parameter optimisation schemes in which closed-form solutions for the graph subshell coefficients are used to derive near-optimal walk times, independent of any specific problem instance. A directed application of this observation is as follows:

\begin{enumerate}
    \item For subshells at $d > 0$, obtain $|\mathcal{N}_{d,0}|$ subshell-optimal walk times, $(t_{1, 0}, \dots, t_{0, |\mathcal{N}_{d,0}|-1})$, by numerical optimisation of $\max |w_{d,k}(t_{d,k})|$ where $t > 0$.
    \item Define the fixed mixing time for ansatz iteration $i$ as $t_{i + 1,k}$ where $k = i \mod |\mathcal{N}_{d, k}|$ and $ i = 0$ to $p - 1$.
    \item Minimse $\langle \hat{Q} \rangle$ by optimisation over the reduced set of variational parameters $\bm{\theta} = (\gamma_0, \gamma_1, \dots, \gamma_{p-1})$.
\end{enumerate}

This scheme was applied to the Schedule B parallel machine scheduling problem, for which the QMOA mixing unitary has $8$ subshell coefficients and the QAOA has $15$. We denote these QVA variants as QMOA~($\gamma$) and QAOA~($\gamma$). Numerical optimisation of $\max |w_{d,k}(t_{d,k})|$ followed the method described in \cref{sec:convergence}.

The approximation ratios and their corresponding probability distributions with the highest ratio at $p=5$ are illustrated in \cref{fig:pms-hybrid-opt}. At this depth, the mean approximation ratios are $0.934$ for QMOA~($\gamma$) and $0.911$ for QAOA~($\gamma$). Both distributions show convergence towards the globally optimal solution. The probability of measuring the global minimum from the prepared QMOA~($\gamma$) state is $0.263$, representing $54.7\%$ of the maximum probability achieved by optimising both $\gamma$ and $t$. In contrast, the probability in the prepared QAOA~($\gamma$) state is $0.0658$, which constitutes $26.7\%$ of the optimal value obtained when both $\gamma$ and $t$ are optimised (see \cref{fig:pms-states}).

\section{Discussion}
\label{sec:discussion}

The results presented in \cref{sec:CTQW-subshells,sec:convergence} support both the theoretical description of QVA dynamics in \cref{sec:theory} and establish that the defining graphs of the QAOA, QMOA, QWOA, QWOA-CS, and QAOAz mixing unitaries have the potential to produce a high degree of constructive interference at a target state. The graph convergence potentials reported in \cref{table:graph_characteristics_ordered_deg} indicate that vertex degree negatively correlates with the maximum convergence obtainable with a single ansatz iteration. We suggest that this follows from the unitary evolution of the CTQW, which requires that probability amplitude transfer between $\bm{s}$ and $\bm{s}^\prime$ be equally weighted over all $\bm{s}^\prime$ within the same subshell. Consequently, amplitude transfer between $\bm{s}$ and $\bm{s}^\prime$ is also negatively correlated with the size of the subshell containing $\bm{s}^\prime$. However, as shown in \cref{fig:transfer-and-discrepancy}, a high degree of transfer is still possible due to the cumulative contribution of all states within the same subshell.

In \cref{fig:transfer-and-discrepancy}, we also see a clear relationship between phase discrepancy and the ability of a mixing unitary to realise its convergence potential. For each graph, the maximum convergence occurred only with zero subshell variance. The extent of this impact appears to be influenced by graph structure. As given in \cref{table:graph_characteristics_ordered_deg}, the $(3, 5)$ Hamming graph has a convergence that is greater than half of its full convergence potential over the largest range of subshell variance, closely followed by the $(7, 3)$ Hamming graph.

Identifying \cref{eq:hamming-coefficeints} allows for a broader investigation of the convergence potential of the $(n, m)$ Hamming graphs. Our results in \cref{sec:conv_hamming} indicate that the convergence potential of such graphs is determined by their largest unstructured subspace. These are characterised by complete graphs of size $m$, for which convergence is upper bound by the theoretical limits of an unstructured quantum search. Accordingly, for $m \leq 4$, complete convergence appears to be attainable with one iteration of $\hat{U}_Q$ and $\hat{U}_W$ independent of $n$, which would provide an exponential advantage if achieved in practice without significant parameter optimisation overhead. In contrast, for $m > 4$, the convergence potential decreases polynomially in $m$ and exponentially in $n$. However, as the amplification grows like $\mathcal{O}(e^{n \log m})$, we still observe an exponential advantage relative to an unbiased sampling of the solution space.

Whether this scaling behaviour can be realised in practice depends on the extent to which problems of interest can approximate this idealised phase distribution with a suitable choice of variational parameters. To this end, we note the observed effectiveness of the QAOA in tackling COPs with binary variables, such as the maxcut and Boolean satisfiability problems~\cite{wurtz_maxcut_2021,basso2021quantum,boulebnane10solving,sureshbabu_parameter_2024,headley_problem-size-independent_2023}. Our results indicate that similar performance is attainable for unconstrained COPs with arties of three and four. The convergence potential results for $m > 4$ also corroborate the observation of exponential scaling in $n$ of the QMOA when applied to optimisation of multivariable functions in previous work ~\cite{matwiejew_quantum_2023}.

Quantum variational algorithm simulation and analysis in \cref{sec:applications} supports the generality of the results in \cref{sec:CTQW-subshells,sec:convergence,sec:conv_hamming}. For both the parallel machine scheduling and portfolio optimisation problems, QVAs with mixing unitaries defined by graphs with lower MSV consistently achieved the highest approximation ratio and the strongest amplification of the globally optimal solution (see \cref{fig:pms-depth-plot,fig:portfolio-states,fig:pms-states}). The correspondence between this and a reduction in subshell variance is further emphasised by \cref{fig:pms-phase,fig:portfolio-phase}, in which an ordered phase distribution over a comparatively narrow range is associated with a higher approximation ratio.    
The numerical simulation results in \cref{sec:pms-results} show that the QMOA outperforms the QAOA in the parallel machine scheduling problem. This is most pronounced in Schedule A, where the state preparation and mixing unitary of the generalised QMOA reduces the quantum search space by 16 times. At face value, one might conclude that this advantage is primarily due to the reduction in size; however, since the defining graph of the QAOA mixing unitary, the $(n, 2)$ Hamming graph, has a convergence potential that is exponentially higher than Hamming graphs with $m > 4$, we suggest that the advantage of the QMOA for this problem should be attributed at least in part to the alignment of its graph with the global structure in the problem solution space. The results obtained for Schedule B offer further support for this position, where, as shown in \cref{fig:pms-depth-plot,fig:pms-states}, the QMOA consistently outperforms the QAOA with a quantum search space of equivalent size.

In the portfolio optimisation problem, the QWOA-CS is consistently the best-performing QVA. As shown in \cref{fig:portfolio-depth-plots,fig:portfolio-states}, it has the highest mean approximation ratio at each depth and, at $p=5$, exhibits the highest convergence to the globally optimal solution. By combining the QWOA indexing unitary with a two-part mixing process over constrained permutation graphs and a $K$-partite graph, it outperforms the QWOA. The efficient amplification of low-cost solutions by the QWOA-CS~(disjoint) supports the efficacy of the hueristic for mixing unitary design and emphasises the cosupporting roles of the QWOA-CS submixers.

The performance of the QAOAz~(XY-parity) is consistent with previous work~\cite{slate_quantum_2021}. In \cref{fig:portfolio-depth-plots}, it has the lowest mean approximation ratio at all depths, except for the QAOAz~(XY-complete) at $p=4$. The QAOAz~(XY-complete) performs slightly better on average but poorly overall compared to the QWOA, QWOA-CS, and QWOA-CS~(disjoint). \cref{table:graph_characteristics_ordered_deg,fig:transfer-and-discrepancy} indicate that these XY-based mixing unitaries can perform an efficient quantum search but are disadvantaged in this case because their search space is larger due to the inclusion of degenerate solutions. In fact, in \cref{fig:portfolio-states,fig:qaoaz_degen}, it appears that the presence of degenerate solutions biases the convergence of the QAOAz~(XY-complete) away from the globally optimal solution. This also appears to be the case for both the QAOAz~(XY-parity) (Portfolio B) and QAOAz~(XY-complete) in \cref{fig:portfolio-phase}.

Altogether, these results indicate that the heuristic for mixing unitary design presented in \cref{sec:mixer-design} leads to QVAs that are highly efficient in terms of their ability to amplify the probability of low-cost solutions at low circuit depth, together with an underlying theoretical framework with strong explanatory potential. By the subshell interpretation of the QVA mixing unitary, with the appropriate choice of graph structure, we can view the $t$ parameter as selecting for transfer between one of a small number of subshells which, given low enough subshell variance, are semi-independent from the problem instance. As demonstrated in \cref{sec:hybrid-optimisation-results}, this opens the potential for developing problem-agnostic methods to mitigate the computational overhead of QVA parameter optimisation.

As detailed in \cref{sec:cart_qva}, the QMOA is a highly efficient QVA with a gate complexity of $\mathcal{O}(n^2)$, which is only a polynomial factor higher than the QAOA~\cite{bennett_quantum_2021}. In \cref{sec:perm_qva}, we find that the indexing process is the determining factor in the gate complexity of the QWOA-CS, which is $\mathcal{O}(n^4)$. In comparison, the QAOAz~(XY-parity) requires $\mathcal{O}(n)$ and the QAOAz~(XY-complete) $\mathcal{O}(n^2)$ fundamental quantum gates. The overhead of the QWOA-CS indexing process may be a not-insignificant increase in complexity for near-term devices. However, as the space of valid solutions is generally exponentially smaller than the full solution space, methods based on post-selection or amplitude amplification of the valid solution space are not scalable. While such choices must be made case by case, our results indicate that a polynomial increase in gate complexity may be an acceptable trade-off, given the potential for a significant reduction in the number of ansatz iterations required for acceptable convergence.

In \cref{fig:transfer-and-discrepancy}, the magnitude of potential amplitude transfer between graph subshells as a function of inter-node distance is completely symmetric for the $(7, 2)$ Hamming graph and the permutation graph. The parity, $(3, 5)$ Hamming, and constrained permutation graphs show a higher potential for amplitude transfer between closer shells to varying extents. Neverthless, we have demonstrated that an efficient mixing unitary relies on the compounding effects of constructive interference over large portions of the solution space, and by this, we expect that the diameter of any suitable graph will grow at most polynomially with the number of combinatorial variables. This underscores our position that, in the design of QVAs, researchers are best served by considering the minimisation of subshell variance, as opposed to a viewpoint that emphasises a measure of distance between solutions. After all, quantum mechanics is inherently non-local.

\section{Conclusion}
\label{sec:conclusion}

We have presented an analytical description of quantum variational algorithms for combinatorial optimisation problems based on a continuous-time quantum walk interpretation of mixing unitaries based on vertex-transitive graphs. This model describes the amplitude transfer induced by these unitaries as a linear combination of time-dependent complex subshell coefficients, which are associated with subspaces corresponding to the automorphism groups of a defining graph. The graphs used by many QVAs of practical interest are distance-transitive or, at least, vertex-transitive; hence, this approach is highly generalisable.

Our theoretical description establishes that a mixing unitary's ability to produce convergence to high-quality solutions depends on minimising the average variance within these subshells, which informs our heuristic for mixing unitary design. By constructing a graph according to the minimum non-zero Hamming distance between solutions in the valid solution space, solutions with similar costs are, on average, mapped to the same subshell. This reduces subshell variance and thus supports efficient convergence to optimal solutions at low circuit depths.

We applied our heuristic for mixing unitary design to constrained and unconstrained combinatorial optimisation problems. The unconstrained case produces mixing unitaries defined by the Hamming graphs. This broad class of graphs encompasses the complete, hypercube, and higher-order Hamming graphs that characterise the mixing unitaries of several established algorithms, including the QWOA, QAOA, and QMOA. By this observation, we generalise the QMOA, initially developed to optimise continuous multivariable functions, to unconstrained combinatorial optimisation problems, including problems with variables that do not have a one-to-one mapping with an integer number of qubits.

In the case of constrained optimisation, we demonstrate that the solution space is fully partitioned into permutations of multisets, where valid solutions are related by a minimum non-zero Hamming distance of two. This partitioning scheme is general to problems with integer equality constraints and is efficiently implementable, as a subset of generating multisets for the valid solution space can be identified in at most polynomial time for all problems of this type.

Following from this, we introduce the QWOA-CS. This algorithm extends the QWOA by constructing an indexing scheme and mixing unitary that preserves the structure of the partitioned solution space. The QWOA-CS mixing unitary encapsulates two sequentially applied mixing unitaries. The first is defined by a continuous-time quantum walk over a graph that connects solutions related by the transposition of two of their elements, which we name the constrained permutation graph. We prove that the maximum degree of this graph grows quadratically with the number of solution variables, meaning that a walk on this graph can be efficiently implemented via sparse Hamiltonian simulation using quantum signal processing. The second mixing unitary enables transfer between the disjoint sets by a continuous-time quantum walk over a non-homogeneous partite graph. For this, we derive an efficient method for its exact simulation based on the quantum Fourier transform. Our complexity analysis shows that the generalised QMOA and QWOA-CS are scalable algorithms with polynomial gate complexity.

Based on our theoretical description of amplitude transfer and the concept of phase discrepancy, we introduce two figures of merit: convergence potential and mean shell variance. Convergence potential provides an upper bound on the convergence supported by the defining graph of a mixing unitary in a single ansatz iteration. We apply this metric to study the convergence characteristics of mixing unitaries based on the graphs that underpin the QAOA, QWOA, QAOAz, QMOA, and QWOA-CS. Overall, the convergence potential is higher in graphs with a small diameter and vertex degree.

Mean shell variance quantifies the average variance in solution costs over the shells of the defining graph of the mixing unitary. This classically compatible measure successfully predicted the relative performance of all algorithms considered. In algorithms utilising graphs with low mean shell variance, we demonstrated that the optimised ansatz unitary produces a phase distribution that distinguishes the globally optimal solution from the bulk of the solution space, even at low ansatz depth.

A closed-form expression for the subshell coefficients of the Hamming graphs supported a detailed analysis of their convergence characteristics. We found that these graphs can provide a degree of amplification that scales exponentially with the number of combinatorial variables. In particular, our results indicate that Hamming graphs over variables with alphabets of size two to four can support complete convergence to a target solution independent of problem size. This result supports the observed effectiveness of the QAOA on problems with binary variables and suggests that similar results might be observed for problems with a variable arity of three and four.

Numerical results simulating the application of the generalised QMOA and the QWOA-CS to the parallel machine scheduling and portfolio rebalancing problems demonstrate that these algorithms offer a significant improvement over pre-existing alternatives. They achieved the highest approximation ratio in each case and consistently converged to the globally optimal solution. 

Finally, we present a parameter setting scheme in which a closed-form solution for the subshell coefficients generates fixed mixing parameters that are independent of the problem instance. Under this scheme, our numerical results show that the QMOA can achieve efficient convergence to the globally optimal solution by varying only the phase-shift parameters.

Our findings suggest several avenues for further investigation, including analysis of the dynamics and convergence characteristics of continuous-time quantum walks on the non-homogeneous $K$-partite graph and the convergence potential of graphs at deeper ansatz depths. The latter must contend with the degrees of freedom introduced by each additional ansatz iteration; however, for the case of unconstrained optimisation, a geometrically informed analysis of the problem solution space may be a fruitful approach to this challenge. There is clearly potential for the development of parameter-setting schemes that are more refined than our illustrative example. These would be supported by an analytical or empirical derivation of a formula for the subshell coefficients of the constrained permutation graph, which might be attainable given our comprehensive description of the structural properties of this graph.

We also note that the state preparation and mixing unitary circuits for the generalised QMOA might serve as the basis for a QWOA-CS circuit that does not depend explicitly on the quantum Fourier transform, which may assist with its execution on near-term quantum processors. One difficulty in achieving this is that, while the generalised QMOA utilises an approximation of a CTQW on a complete graph that is exact up to a global phase, the simulation of a CTQW on a $K$-partite graph, based on sequential walks over the Laplacian of disjoint complete graphs and a fully connected complete graph, must account for relative phase differences between the partite sets. However, an approximation of this walk may be sufficient.

In conclusion, we have demonstrated that a continuous-time quantum walk description of the dynamics of quantum variational algorithms offers a highly explanatory lens through which researchers can tackle the core challenges that arise in designing these algorithms. We believe that this supports the hope that quantum variational algorithms will provide a significant quantum advantage in optimising NP-hard problems that are of practical importance to a wide range of domains.

%% file: appendix.tex
\section{Parallel Machine Scheduling}
\label{app:parallel_machine_data}
Solution costs for the parallel machine scheduling problem were computed as described in \cref{sec:pms-results}. The problem data and model parameters for the two considered instances are given below.

\filbreak

\vspace{1em}

\noindent Schedule A:
\begin{itemize}
    \item $n = 6$,  $m = 5$,  $\eta = 0.5$, $\alpha = 2$, $a = 100$
    \item $w = (3, 6, 1, 4, 5, 2)$
    \item $\tau = (21, 22, 13, 14,  5, 15)$
    \item $\kappa_{\text{QMOA}} = (65, 61, 41, 36, 79)$
    \item $\kappa_{\text{QAOA}} = (65, 61, 41, 36, 79, 41, 41, 41)$ 
\end{itemize}

\filbreak

\vspace{1em}

\noindent Schedule B:
\begin{itemize}
    \item $n = 7$,  $m = 4$,  $\eta = 0.5$, $\alpha = 2$
    \item $w = (7, 3, 2, 4, 1, 6, 5)$
    \item $\tau = (23, 9, 11, 17, 6, 11, 12)$
    \item $\kappa = (71, 62, 50, 97)$
\end{itemize}

\section{Portfolio Rebalancing}
\label{app:portfolio_data}
Covariance and mean return for portfolios of $n$ randomly selected assets were computed as described in \cref{sec:portfolio-results} between 28/12/2022 and 29/12/2023. The asset sets and model parameters for the two considered portfolios are given below.

\vspace{1em}

\filbreak

\noindent Portfolio A:
\begin{itemize}
    \item $n=6$, $A=2$, $\eta = 0.5$
    \item Assets: AMC (Amcor Ltd.), IPH (IPH Ltd.), TLS (Telstra Corporation Ltd.), CMW (Cromwell Property Group), CBA (Commonwealth Bank of Australia), RIO (Rio Tinto Ltd.)
\end{itemize}

\vspace{1em}

\filbreak

\noindent Portfolio B:
\begin{itemize}
    \item $n=8$, $A=2$, $\eta = 0.5$
    \item Assets: TLS (Telstra Corporation Ltd.), SGR (The Star Entertainment Group Ltd.), QAN (Qantas Airways Ltd.), BXB (Brambles Ltd.), RIO (Rio Tinto Ltd.), ANZ (Australia and New Zealand Banking Group Ltd.), AMP (AMP Ltd.), BKW (Brickworks Ltd.)
\end{itemize}

%% file: main.bbl
\begin{thebibliography}{10}
\expandafter\ifx\csname url\endcsname\relax
  \def\url#1{\texttt{#1}}\fi
\expandafter\ifx\csname urlprefix\endcsname\relax\def\urlprefix{URL }\fi
\expandafter\ifx\csname href\endcsname\relax
  \def\href#1#2{#2} \def\path#1{#1}\fi

\bibitem{Amaro2022Dec}
D.~Amaro, M.~Rosenkranz, N.~Fitzpatrick, K.~Hirano, M.~Fiorentini, A case study of variational quantum algorithms for a job shop scheduling problem, EPJ Quantum Technology 9~(1) (2022) 5.

\bibitem{bennett_quantum_2021}
T.~Bennett, E.~Matwiejew, S.~Marsh, J.~B. Wang, Quantum walk-based vehicle routing optimisation, Frontiers in Physics 9 (2021) 730856.

\bibitem{slate_quantum_2021}
N.~Slate, E.~Matwiejew, S.~Marsh, J.~B. Wang, Quantum walk-based portfolio optimisation, Quantum 5 (2021) 513.

\bibitem{lozano_combinatorial_2019}
R.~C. Lozano, M.~Carlsson, G.~H. Blindell, C.~Schulte, Combinatorial register allocation and instruction scheduling, ACM Transactions on Programming Languages and Systems (TOPLAS) 41~(3) (2019) 1--53.

\bibitem{kell_scientific_2012}
D.~B. Kell, Scientific discovery as a combinatorial optimisation problem: how best to navigate the landscape of possible experiments?, Bioessays 34~(3) (2012) 236--244.

\bibitem{korte_2022}
B.~Korte, J.~Vygen, Combinatorial Optimization Theory and Algorithms, 5th Edition, Algorithms and Combinatorics, 21, Springer Berlin Heidelberg, Berlin, Heidelberg, 2012.

\bibitem{pardalos_complexity_1992}
P.~M. Pardalos, S.~Jha, Complexity of uniqueness and local search in quadratic 0--1 programming, Operations research letters 11~(2) (1992) 119--123.

\bibitem{CZ12}
J.~I. Cirac, P.~Zoller, Goals and opportunities in quantum simulation, Nature physics 8~(4) (2012) 264--266.

\bibitem{FH16}
E.~Farhi, A.~W. Harrow, Quantum supremacy through the quantum approximate optimization algorithm, arXiv preprint arXiv:1602.07674 (2016).

\bibitem{boulebnane10solving}
S.~Boulebnane, A.~Montanaro, Solving boolean satisfiability problems with the quantum approximate optimization algorithm, arXiv preprint arXiv:2208.06909 (2022).

\bibitem{maslov2021quantum}
D.~Maslov, J.-S. Kim, S.~Bravyi, T.~J. Yoder, S.~Sheldon, Quantum advantage for computations with limited space, Nature Physics 17~(8) (2021) 894--897.

\bibitem{bravyi2020quantum}
S.~Bravyi, D.~Gosset, R.~K{\"{o}}nig, M.~Tomamichel, Quantum advantage with noisy shallow circuits, Nature Physics 16~(10) (2020) 1040--1045.

\bibitem{wu2021strong}
Y.~Wu, W.-S. Bao, S.~Cao, F.~Chen, M.-C. Chen, X.~Chen, T.-H. Chung, H.~Deng, Y.~Du, D.~Fan, et~al., Strong quantum computational advantage using a superconducting quantum processor, Physical review letters 127~(18) (2021) 180501.

\bibitem{basso2021quantum}
J.~Basso, E.~Farhi, K.~Marwaha, B.~Villalonga, L.~Zhou, The quantum approximate optimization algorithm at high depth for maxcut on large-girth regular graphs and the {S}herrington-{K}irkpatrick model, arXiv preprint arXiv:2110.14206 (2021).

\bibitem{Farhi14}
E.~Farhi, J.~Goldstone, S.~Gutmann, A quantum approximate optimization algorithm, arXiv preprint arXiv:1411.4028 (2014).

\bibitem{marsh_quantum_2019}
S.~Marsh, J.~Wang, A quantum walk-assisted approximate algorithm for bounded {NP} optimisation problems, Quantum Information Processing 18~(3) (2019) 61.

\bibitem{marsh_combinatorial_2020}
S.~Marsh, J.~B. Wang, Combinatorial optimization via highly efficient quantum walks, Physical Review Research 2~(2) (2020) 023302.

\bibitem{HWO+19}
S.~Hadfield, Z.~Wang, B.~O’gorman, E.~G. Rieffel, D.~Venturelli, R.~Biswas, From the quantum approximate optimization algorithm to a quantum alternating operator ansatz, Algorithms 12~(2) (2019) 34.

\bibitem{Preskill18}
J.~Preskill, Quantum computing in the nisq era and beyond, Quantum 2 (2018) 79.

\bibitem{symons_practitioners_2023}
B.~C. Symons, D.~Galvin, E.~Sahin, V.~Alexandrov, S.~Mensa, A practitioner’s guide to quantum algorithms for optimisation problems, Journal of Physics A: Mathematical and Theoretical 56~(45) (2023) 453001.

\bibitem{zhou2020quantum}
L.~Zhou, S.-T. Wang, S.~Choi, H.~Pichler, M.~D. Lukin, Quantum approximate optimization algorithm: Performance, mechanism, and implementation on near-term devices, Physical Review X 10~(2) (2020) 021067.

\bibitem{Brandhofer2023}
S.~Brandhofer, D.~Braun, V.~Dehn, G.~Hellstern, M.~Huls, Y.~Ji, I.~Polian, A.~S. Bhatia, T.~Wellens, {Benchmarking the performance of portfolio optimization with {QAOA}}, Quantum Information Processing 22 (2023) 25.

\bibitem{CEB20}
J.~Cook, S.~Eidenbenz, A.~B\"{a}rtschi, The quantum alternating operator ansatz on maximum k-vertex cover, in: 2020 IEEE International Conference on Quantum Computing and Engineering (QCE), IEEE, 2020, pp. 83--92.

\bibitem{ruan_quantum_2023}
Y.~Ruan, Z.~Yuan, X.~Xue, Z.~Liu, Quantum approximate optimization for combinatorial problems with constraints, Information Sciences 619 (2024) 98--125.

\bibitem{Grover97}
L.~K. Grover, Quantum mechanics helps in searching for a needle in a haystack, Physical review letters 79~(2) (1997) 325.

\bibitem{zalka_grovers_1999}
C.~Zalka, Grover’s quantum searching algorithm is optimal, Physical Review A 60~(4) (1999) 2746.

\bibitem{matwiejew_quantum_2023}
E.~Matwiejew, J.~Pye, J.~B. Wang, Quantum optimisation for continuous multivariable functions by a structured search, Quantum Science and Technology 8~(4) (2023) 045013.

\bibitem{bridi_analytical_2024}
G.~A. Bridi, F.~d.~L. Marquezino, Analytical results for the quantum alternating operator ansatz with {G}rover mixer, arXiv preprint arXiv:2401.11056 (2024).

\bibitem{bennett2021quantum}
T.~Bennett, J.~B. Wang, Quantum optimisation via maximally amplified states, arXiv preprint arXiv:2111.00796 (2021).

\bibitem{wurtz_maxcut_2021}
J.~Wurtz, P.~Love, Maxcut quantum approximate optimization algorithm performance guarantees for p> 1, Physical Review A 103~(4) (2021) 042612.

\bibitem{sureshbabu_parameter_2024}
S.~H. Sureshbabu, D.~Herman, R.~Shaydulin, J.~Basso, S.~Chakrabarti, Y.~Sun, M.~Pistoia, Parameter setting in quantum approximate optimization of weighted problems, Quantum 8 (2024) 1231.

\bibitem{headley_problem-size-independent_2023}
D.~Headley, F.~K. Wilhelm, Problem-size-independent angles for a {G}rover-driven quantum approximate optimization algorithm, Physical Review A 107~(1) (2023) 012412.

\bibitem{wang_quantum_2018}
Z.~Wang, S.~Hadfield, Z.~Jiang, E.~G. Rieffel, Quantum approximate optimization algorithm for {M}axcut: A fermionic view, Physical Review A 97~(2) (2018) 022304.

\bibitem{ADZ93}
Y.~Aharonov, L.~Davidovich, N.~Zagury, Quantum random walks, Physical Review A 48~(2) (1993) 1687.

\bibitem{venegas-andraca_quantum_2012}
S.~E. Venegas-Andraca, Quantum walks: a comprehensive review, Quantum Information Processing 11~(5) (2012) 1015--1106.

\bibitem{QWbook2014}
K.~Manouchehri, J.~B. Wang, Physical implementation of quantum walks, Springer, 2014.

\bibitem{CG04}
A.~M. Childs, J.~Goldstone, Spatial search by quantum walk, Physical Review A 70~(2) (2004) 022314.

\bibitem{QMW+22}
D.~Qu, S.~Marsh, K.~Wang, L.~Xiao, J.~Wang, P.~Xue, Deterministic search on star graphs via quantum walks, Physical Review Letters 128~(5) (2022) 050501.

\bibitem{razzoli_universality_2022}
L.~Razzoli, P.~Bordone, M.~G. Paris, Universality of the fully connected vertex in laplacian continuous-time quantum walk problems, Journal of Physics A: Mathematical and Theoretical 55~(26) (2022) 265303.

\bibitem{novo_systematic_2015}
L.~Novo, S.~Chakraborty, M.~Mohseni, H.~Neven, Y.~Omar, Systematic dimensionality reduction for quantum walks: Optimal spatial search and transport on non-regular graphs, Scientific reports 5~(1) (2015) 13304.

\bibitem{marsh_deterministic_2021}
S.~Marsh, J.~B. Wang, Deterministic spatial search using alternating quantum walks, Physical Review A 104~(2) (2021) 022216, publisher: American Physical Society.

\bibitem{SKW03}
N.~Shenvi, J.~Kempe, K.~B. Whaley, Quantum random-walk search algorithm, Physical Review A 67~(5) (2003) 052307.

\bibitem{xiao_branch_2021}
Y.~Xiao, Y.~Zheng, Y.~Yu, L.~Zhang, X.~Lin, B.~Li, A branch and bound algorithm for a parallel machine scheduling problem in green manufacturing industry considering time cost and power consumption, Journal of Cleaner Production 320 (2021) 128867.

\bibitem{henning_complexity_2020}
S.~Henning, K.~Jansen, M.~Rau, L.~Schmarje, Complexity and inapproximability results for parallel task scheduling and strip packing, Theory of Computing Systems 64~(1) (2020) 120--140.

\bibitem{azizoglu1999minimization}
M.~Azizoglu, O.~Kirca, On the minimization of total weighted flow time with identical and uniform parallel machines, European Journal of Operational Research 113~(1) (1999) 91--100.

\bibitem{anand2012resource}
S.~Anand, N.~Garg, A.~Kumar, Resource augmentation for weighted flow-time explained by dual fitting, in: Proceedings of the twenty-third annual ACM-SIAM symposium on Discrete Algorithms, SIAM, 2012, pp. 1228--1241.

\bibitem{Markowitz1952}
H.~M. Markowits, Portfolio selection, Journal of finance 7~(1) (1952) 71--91.

\bibitem{buonaiuto_best_2023}
G.~Buonaiuto, F.~Gargiulo, G.~De~Pietro, M.~Esposito, M.~Pota, Best practices for portfolio optimization by quantum computing, experimented on real quantum devices, Scientific Reports 13~(1) (2023) 19434.

\bibitem{baker_wasserstein_2022}
J.~S. Baker, S.~K. Radha, Wasserstein solution quaility and the quantum approximate optimization algorithm: A portfolio optimization case study, arXiv preprint arXiv:2202.06782 (2 2022).

\bibitem{hodson_portfolio_2019}
M.~Hodson, B.~Ruck, H.~Ong, D.~Garvin, S.~Dulman, Portfolio rebalancing experiments using the quantum alternating operator ansatz (2019).

\bibitem{OML19}
R.~Or{\'{u}}s, S.~Mugel, E.~Lizaso, Quantum computing for finance: Overview and prospects, Reviews in Physics 4 (2019) 100028.

\bibitem{qu_experimental_2024}
D.~Qu, E.~Matwiejew, K.~Wang, J.~Wang, P.~Xue, Experimental implementation of quantum-walk-based portfolio optimization, Quantum Science and Technology 9~(2) (2024) 025014.

\bibitem{Rebentrost18}
P.~Rebentrost, S.~Lloyd, {Quantum computational finance: quantum algorithm for portfolio optimization}, arXiv (11 2018).

\bibitem{nielsen_quantum_2010}
M.~A. Nielsen, I.~L. Chuang, Quantum computation and quantum information, Cambridge university press, 2010.

\bibitem{matwiejew_quop_mpi_2022}
E.~Matwiejew, J.~B. Wang, {Q}u{O}p\_{MPI}: A framework for parallel simulation of quantum variational algorithms, Journal of Computational Science 62 (2022) 101711.

\bibitem{crescenzi_structure_1999}
P.~Crescenzi, V.~Kann, R.~Silvestri, l.~L. Trevisan, Structure in {Approximation} {Classes}, SIAM Journal on Computing 28 (1999) 24.

\bibitem{guastaroba_models_2009}
G.~Guastaroba, R.~Mansini, M.~G. Speranza, Models and simulations for portfolio rebalancing, Computational Economics 33 (2009) 237--262.

\bibitem{green15}
A.~Green, {XVA}: credit, funding and capital valuation adjustments, John Wiley \& Sons, 2015.

\bibitem{RL18}
P.~Rebentrost, S.~Lloyd, {Options futures and other derivatives.}, arXiv (11 2018).

\bibitem{bener_2015}
S.~G.~W. Edward A.~Bender, Lists, Decisions and Graphs, University of California at San Diego, 2015.

\bibitem{biggs_1974}
N.~Biggs, Algebraic graph theory, 2nd Edition, Cambridge University Press, 1974.

\bibitem{fortunato_community_2012}
S.~Fortunato, C.~Castellano, Community structure in graphs, arXiv preprint arXiv:0712.2716 (2007).

\bibitem{jafarizadeh_investigation_2007}
M.~A. Jafarizadeh, S.~Salimi, Investigation of continuous-time quantum walk via spectral distribution associated with adjacency matrix, Annals of physics 322~(5) (2007) 1005--1033.

\bibitem{childs_limitations_nodate}
A.~M. Childs, R.~Kothari, Limitations on the simulation of non-sparse hamiltonians, arXiv preprint arXiv:0908.4398 (2009).

\bibitem{berry2009black}
D.~W. Berry, A.~M. Childs, Black-box hamiltonian simulation and unitary implementation, arXiv preprint arXiv:0910.4157 (2009).

\bibitem{low2017optimal}
G.~H. Low, I.~L. Chuang, Optimal hamiltonian simulation by quantum signal processing, Physical review letters 118~(1) (2017) 010501.

\bibitem{low_hamiltonian_2019}
G.~H. Low, I.~L. Chuang, Hamiltonian simulation by qubitization, Quantum 3 (2019) 163.

\bibitem{zhou_quantum_2017}
S.~Zhou, T.~Loke, J.~A. Izaac, J.~Wang, Quantum {F}ourier transform in computational basis, Quantum Information Processing 16~(3) (2017) 82.

\bibitem{zhou_efficient_2017}
S.~Zhou, J.~Wang, Efficient quantum circuits for dense circulant and circulant like operators, Royal Society open science 4~(5) (2017) 160906.

\bibitem{davis_1979}
P.~J. Davis, Circulant Matrices, 1st Edition, Pure \& Applied Mathematics, John Wiley \& Sons Inc, 1979.

\bibitem{wang_xy-mixers_2020}
N.~Rubin, Z.~Wang, E.~Rieffel, J.~M. Dominy, {XY}-mixers: analytical and numerical results for {QAOA}, Phys. Rev. A 101 (2020) 012320.

\bibitem{cleve2000}
R.~Cleve, J.~Watrous, Fast parallel circuits for the quantum fourier transform, in: Proceedings 41st Annual Symposium on Foundations of Computer Science, 2000, pp. 526--536.

\bibitem{ruiz-perez_quantum_2017}
L.~Ruiz-Perez, J.~C. Garcia-Escartin, Quantum arithmetic with the quantum {F}ourier transform, Quantum Information Processing 16 (2017) 1--14.

\bibitem{khosropour_quantum_2011}
A.~Khosropour, H.~Aghababa, B.~Forouzandeh, Quantum division circuit based on restoring division algorithm, in: 2011 Eighth International Conference on Information Technology: New Generations, IEEE, 2011, pp. 1037--1040.

\bibitem{cuccaro_new_2004}
S.~A. Cuccaro, T.~G. Draper, S.~A. Kutin, D.~P. Moulton, A new quantum ripple-carry addition circuit (2004).

\bibitem{douglas2009efficient}
B.~Douglas, J.~Wang, Efficient quantum circuit implementation of quantum walks, Physical Review A 79~(5) (2009) 052335.

\bibitem{andries_e_brouwer_spectra_2012}
{Andries E. Brouwer}, {Willem H. Haemers}, Spectra of Graphs, Springer, 2012.

\bibitem{childs2004quantum}
A.~M. Childs, Quantum information processing in continuous time, Ph.D. thesis, Massachusetts Institute of Technology (2004).

\bibitem{yoder2014fixed}
T.~J. Yoder, G.~H. Low, I.~L. Chuang, Fixed-point quantum search with an optimal number of queries, Physical review letters 113~(21) (2014) 210501.

\bibitem{medvidovic2021classical}
M.~Medvidovi{\'c}, G.~Carleo, Classical variational simulation of the quantum approximate optimization algorithm, npj Quantum Information 7~(1) (2021) 101.

\bibitem{heubach_2009}
T.~M. Silvia~Heubach, Combinatorics of Compositions and Words, 1st Edition, Discrete Mathematics and Its Applications, Chapman and Hall CRC, 2009.

\bibitem{earl_2021}
R.~Earl, J.~Nicholson, The concise {O}xford dictionary of mathematics, Oxford University Press, 2021.

\bibitem{horvat2023igraphm}
S.~Horvát, J.~Podkalicki, G.~Csárdi, T.~Nepusz, V.~Traag, F.~Zanini, D.~Noom, {IG}raph/{M}: graph theory and network analysis for {M}athematica, Journal of Open Source Software 8~(81) (2023) 4899.

\bibitem{Mathematica}
W.~R. Inc., Mathematica, {V}ersion 13.2, champaign, IL, 2022 (2022).

\bibitem{jones_scipy_2001}
E.~Jones, T.~Oliphant, P.~Peterson, \href{http://www.scipy.org/}{{SciPy}: {Open} source scientific tools for {Python}} (2001).
\newline\urlprefix\url{http://www.scipy.org/}

\bibitem{gao2012implementing}
F.~Gao, L.~Han, Implementing the nelder-mead simplex algorithm with adaptive parameters, Computational Optimization and Applications 51~(1) (2012) 259--277.

\bibitem{yahoo_finance_2024}
{Y}ahoo finance – stock market live, quotes, business \& finance news, \url{https://finance.yahoo.com/}, accessed: 2024-04-07 (2024).

\end{thebibliography}
